\documentclass[a4paper, USenglish]{lipics-v2021}
\nolinenumbers

\usepackage[linesnumbered,lined,boxed,commentsnumbered]{algorithm2e}
\usepackage{tabularx}
\usepackage{booktabs}

\title{Colorless Tasks and Extension-Based Proofs}

\author{Yusong Shi}{Department of Computer Science and Technology, Tsinghua University, China}{shiys20@mails.tsinghua.edu.cn}{}{}

\author{Weidong Liu}{Department of Computer Science and Technology, Tsinghua University, China \and Zhongguancun Laboratory, China}{liuwd@mail.tsinghua.edu.cn}{}{}

\authorrunning{Y. Shi and W. Liu} 

\Copyright{Yusong Shi and Weidong Liu} 

\ccsdesc[500]{Theory of computation~Interactive proof systems}
\ccsdesc[500]{Theory of computation~Distributed algorithms}
\ccsdesc[500]{Theory of computation~Distributed computing models}
\ccsdesc[500]{Theory of computation~Problems, reductions and completeness}

\keywords{Colorless tasks, Impossibility proofs, Extension-based proof}

\category{} 

\relatedversion{} 

\acknowledgements{We would like to thank Faith Ellen and Jason Liu for helpful discussions and the anonymous reviewers for their comments. }

\newtheorem{property}{Property}
\newcommand*{\FULL}{}

\begin{document}
\maketitle

\begin{abstract}
The concept of extension-based proofs models the idea of a valency argument, which is widely used in distributed computing. Extension-based proofs are limited in power: it has been shown that there is no extension-based proof of the impossibility of a wait-free protocol for $(n,k)$-set agreement among $n > k \geq 2$ processes. There are only a few tasks that have been proven to have no extension-based proof of the impossibility, since the techniques in these works are closely related to the specific task.

We give a necessary and sufficient condition for colorless tasks to have no extension-based proofs of the impossibility of wait-free protocols in the NIIS model. We introduce a general adversarial strategy decoupled from any concrete task specification. 
In this strategy, some properties of the chromatic subdivision that is widely used in distributed computing are proved. 
 
\end{abstract}

\newpage
  \small\tableofcontents
\newpage

\section{Introduction}
\label{sec:introduction}

The consensus task is one of the most important tasks in distributed computing. There are $n$ faulty processes, each of which starts with a private input, communicates with each other, and finally agrees on a single value. A wait-free protocol to solve the consensus task has to satisfy three properties. Every non-faulty process eventually decides some value (Termination). If all non-faulty processes have the same input value $v$, then any non-faulty process should output the value $v$ (Validity).
Every non-faulty process must output the same value (Agreement). One of the most important results in distributed computing, due to Fischer, Lynch, and Paterson \cite{Fischer85}, is that there is no deterministic protocol that solves the consensus task in the asynchronous message passing system. The key idea of their proof is called a valency argument, which proves the existence of an infinite execution in which no process terminates. 

The $(n, k)$-set agreement task, which is a generalization of the consensus task, was first proposed by Chaudhuri \cite{Chaudhuri93}. In this task, there are $n$ processes, each starting with an input in $\{0, 1, \cdots k\}$, where $1 \leq k < n$. Each process that does not crash must output the input value of some process that it has seen and, at most $k$ different values can be output by all processes collectively. The consensus task is a special case where $k = 1$.
The $(n, k)$-set agreement task was independently shown to have no wait-free protocol by Borowsky and Gafni \cite{Borowsky93}, Herlihy and Shavit \cite{Herlihy99}, and Saks and Zaharoglou \cite{Michael00}. Topological techniques were used to prove these results. Borowsky and Gafni used a powerful simulation method that allows $N$-process protocols to be executed by fewer processes in a resilient way. The paper by Saks and Zaharoglou constructed a topological structure that captures processors' local views of the whole system. The proof presented a relation between the $(n, k)$-set agreement task and Sperner's lemma or the Brouwer fixed point theorem. The paper by Herlihy and Shavit introduced a more general formalism based on topology to discuss computation in asynchronous distributed systems. They used topological structures, known as simplicial complexes, to model tasks and protocols. In this framework, they proved one of the most important theorems in distributed computing, the asynchronous computability theorem, which is a necessary and sufficient condition for a task to be solvable in shared memory model by a wait-free protocol. The impossibility of the $(n, k)$-set agreement task was also proved using only combinatorial techniques in \cite{Attiya02, Attiya11, Attiya12, Shi20}.

In \cite{Alistarh19}, Alistarh, Aspnes, Ellen, Gelashvili and Zhu pointed out the differences between valency arguments and combinatorial or topological techniques. In the proof by Fischer, Lynch and Paterson, an infinite execution can be constructed by extending an initial execution infinitely often. In contrast, in those proofs using combinatorial techniques, the existence of a bad execution is proved, but not explicitly constructed. \cite{Alistarh19} generalized this type of proof and called it an extension-based proof. An extension-based proof is defined as an interaction between a prover and a protocol that claims to solve a task. The prover tries to find out some errors in the protocol by submitting queries to the protocol. If the prover manages to do so, then the prover wins against the protocol.
If there exists a prover that can win against any protocol that claims to solve a task, we say that this task has an extension-based impossibility proof. The proof of the impossibility of consensus is an example of an extension-based proof. In the same paper, they showed that there are no extension-based proofs for the impossibility of a wait-free protocol for the $(n,k)$-set agreement in the non-uniform iterated immediate snapshot (NIIS) model. The same result was proved in the non-uniform iterated snapshot (NIS) model in the journal version \cite{Alistarh23}. Some tasks \cite{Alistarh21, Liu22} that are closely related to the set agreement task and 1-dimensional colorless tasks have also been shown to have no extension-based proofs.

Do other tasks also have no extension-based impossibility proofs? One possible way to extend existing results is via reductions. A reduction from task $T$ to task $S$ is to construct a protocol to solve task $T$ using a protocol that solves $S$. Brusse and Ellen \cite{Brusse21} gave the first result about reductions: It is possible to transform extension-based impossibility proofs if the reduction is limited to the use of only one instance of $S$, and extension-based proofs are augmented. Using this result, they proved that there are no (augmented) extension-based proofs of the impossibility of wait-free protocols in the NIS model for a number of other distributed computing problems.

Another way to generate new results is to find a condition that characterizes the tasks that have extension-based impossibility proofs. A task is specified by a tuple $(\mathcal{I}, \mathcal{O}, \Delta)$ \label{syb:task}. A protocol solves a task $(\mathcal{I}, \mathcal{O}, \Delta)$ if, starting with any input values in $\mathcal{I}$, processes decide on output values in $\mathcal{O}$ after communicating with each other for some steps according to the protocol, respecting the input/output relation $\Delta$. Both $\mathcal{I}$ and $\mathcal{O}$ are closed under containment, since processes are assumed to be faulty and may crash at any time. For example, in the 2-process ($p_{0}$ and $p_{1}$) binary consensus task, $\mathcal{I} = \{\{(p_{0}, 0), (p_{1}, 0)\}, \{(p_{0}, 0), (p_{1}, 1)\}, \allowbreak \{(p_{0}, 1), (p_{1}, 0)\}, \{(p_{0}, 1), (p_{1}, 1)\}, \{(p_{0}, 0)\}, \{(p_{0}, 1)\}, \{(p_{1}, 0)\}, \{(p_{1}, 1)\} \}$ and $\mathcal{O} = \{\{(p_{0}, 0), (p_{1}, 0)\}, \allowbreak \{(p_{0}, 1), (p_{1}, 1)\}, \allowbreak \{(p_{0}, 0)\}, \{(p_{1}, 1)\} \}$. The task specification is captured by $\Delta$, which maps each possible input vector to the allowed output vectors. If two processes $p_{0}$ and $p_{1}$ have input values $\{(p_{0}, 0), (p_{1}, 1)\}$ and do not crash, then their allowed output values $\Delta(\{(p_{0}, 0), (p_{1}, 1)\})$ are $\{\{(p_{0}, 0), (p_{1}, 0)\}, \{(p_{0}, 1), (p_{1}, 1)\} \}$. If the process $p_{0}$ has an input $0$ but the process $p_{1}$ crashes, its allowed output value $\Delta(\{(p_{0}, 0)\})$ is $\{(p_{0}, 0)\}$.  We can show that a task $(\mathcal{I}, \mathcal{O}, \Delta)$ has no extension-based proofs if we can design an adversarial strategy that can construct an adaptive protocol that wins against any extension-based prover.

In this paper, we focus on a subset of tasks called {\em colorless tasks}. A colorless task is defined only in terms of input and output values, without process ids. Many important tasks in distributed computing, including the consensus task, are colorless tasks. Note that in the previous work about colorless tasks(introduced in \cite{Bor01}, recently used in \cite{Attiya23_2}), colorless tasks are usually defined in a different way from tasks, and there are some consequences related to this definition \cite{Herlihy10, Herlihy17, Attiya23_2}. A colorless task is specified by a tuple $(\mathcal{I^{*}}, \mathcal{O^{*}}, \Delta^{*})$. \label{syb:colorless_task} $\mathcal{I^{*}}$ (or $\mathcal{O^{*}}$) contains only sets of input values (or output values), regardless of the number of processes involved and regardless of which process has a particular input or output value. As an example, for the n-process binary consensus task, $\mathcal{I^{*}} = \{\{0\}, \{1\}, \{0, 1\} \}$ and $\mathcal{O^{*}} = \{\{0\}, \{1\}\}$. The task specification is expressed as $\Delta^{*}(\{0\}) = \{0\}$, $\Delta^{*}(\{1\}) = \{1\}$ and $\Delta^{*}(\{0, 1\}) = \{\{0\}, \{1\}\}$. 

\subsection{The reason to use the colorless condition}
In this paper, {\em all} our discussions use the definition and related consequences of tasks rather than those specified for colorless tasks. Given a colorless task $(\mathcal{I^{*}}, \mathcal{O^{*}}, \Delta^{*})$, we can transform it into the form $(\mathcal{I}, \mathcal{O}, \Delta)$. The relations between the two definitions are discussed in more detail in Appendix \ref{appendix:transformation}. We will only use the format $(I, O, \Delta)$ in which process ids are included rather than $(\mathcal{I^{*}}, \mathcal{O^{*}}, \Delta^{*})$ that does not contain process ids.

So why do we talk about colorless tasks while adopting the form of general tasks? Part of our design needs a property (Property \ref{pro:colorless_property}) of the input/output relation $\Delta$.
\begin{property}
\label{pro:colorless_property}
In any possible execution, if a process is allowed to output a value $v$, then any other process that has seen a superset of the values seen by this process is also allowed to output the value $v$.
\end{property}
This property is intrinsic for colorless tasks. We regard colorless tasks as tasks having this property(i.e. a subset of tasks). In fact, all our discussions work for all tasks that satisfy this property, which is a superset of colorless tasks.

\subsection{Our contributions}
In this paper, we show which types of colorless tasks whose impossibility to have wait-free protocols can only be proved using combinatorial arguments but not valency arguments. In other words, we give a necessary and sufficient condition for a colorless task to have no extension-based impossibility proofs. Using this condition, we design an adversarial strategy that can win against any extension-based prover for those tasks. This is the first result to generalize the adversarial strategy proposed for set-agreement \cite{Alistarh19} to any colorless task, an improvement to previous results which are specified for some concrete tasks \cite{Alistarh19, Alistarh21} or 1-dimensional colorless tasks \cite{Attiya23_2}. Finally, we illustrate the results introduced in this paper by applying them to some colorless tasks.


\ifdefined\FULL
This paper will be organized as follows. Since we will use the results of \cite{Herlihy99}, the combinatorial topology and models used in this paper are briefly introduced in Section \ref{sec:preliminaries}. Section \ref{sec:related_work} is about related work. Section \ref{sec:motivation_and_summary} describes the motivation behind our necessary and sufficient condition. In Section \ref{sec:preparations}, we introduce the key concepts used in our adversarial strategy. In Section \ref{sec:nccondition} we give a necessary and sufficient condition for a colorless task to have no extension-based proofs. In Section \ref{sec:applications}, we apply the results in this paper to some colorless tasks. We derive some results similar to those in \cite{Attiya23_2}. Finally, Section \ref{sec:conclusions} summarizes our results and presents some potential extensions to our work. For readers who are not concerned with technical details, Section \ref{sec:motivation_and_summary} summarizes our ideas and main results in Sections \ref{sec:preparations} and Section \ref{sec:nccondition}. These readers can skip Sections \ref{sec:preparations} and Section \ref{sec:nccondition} and directly advance to the application and conclusion parts. There are many symbols used in this paper, we list the important ones in Appendix \ref{appendix:symbols}. 

\else
This paper will be organized as follows. Since we will use the results of \cite{Herlihy99}, the combinatorial topology and models used in this paper are briefly introduced in Section \ref{sec:preliminaries}. Section \ref{sec:related_work} is about related work. Section \ref{sec:motivation_and_summary} describes a high-level idea behind our adversarial strategy. The details of our strategy are shown in Sections 5 and 6 of the extended paper. Finally, Section \ref{sec:conclusions} summarizes our results and presents some potential extensions to our work. 
\fi

\section{Preliminaries}
\label{sec:preliminaries}

\subsection{NIIS models}
\label{sec:preliminaries:NIIS_model}
In this paper, we consider the non-uniform iterated immediate snapshot (NIIS) \label{syb:NIIS} model introduced by Hoest and Shavit\cite{Hoest06}.

An {\em snapshot object} consists of an array and supports two types of operation: write and read. A process with id $i$ writes a value to the $i$-th cell of the array, or reads a snapshot of the array.
Similarly, an {\em immediate snapshot} (IS) object was introduced by Borowsky and Gafni in \cite{Borowsky93_IS}. An {\em IS object} \label{syb:IS} consists of an array and supports only one type of operation, called a {\em writeread} operation, where a process with id $i$ writes a value to the $i$-th cell of the array and returns a snapshot of the array immediately following the write. The writeread operations performed to some IS object by different processes are said to be concurrent if all snapshots occur after all writes to the array are finished.

The NIIS model assumes an unbounded sequence of IS objects $IS_{1}, IS_{2} \cdots$. $(n + 1)$ \label{syb:n_plus_1} sequential threads of control, called {\em processes}, $p_{0}, p_{1} \ldots p_{n}$ \label{syb:processes}, communicate through IS objects to solve decision tasks. Note that in the following sections we assume $n + 1$ processes rather than $n$ processes in Section \ref{sec:introduction} or previous work to simplify our notations. The set of all processes is denoted by $\Pi$ \label{syb:pi}.  A {\em protocol} is a distributed program to solve a task.
In any execution of a protocol in the NIIS model, each process $p_{i}$ performs a writeread operation on each IS object starting from $IS_{1}$. Initially, $p_{i}$'s state contains its identifier $i$ and its input value. Each time $p_{i}$ performs a writeread operation on some IS object $IS_{j}$ using its current state $s_{i}$ as argument, and sets its current state $s_{i}$ to its identifier $i$ and the response of its writeread operation. Then $p_{i}$ consults a map $\delta$ \label{syb:delta} to determine whether it should terminate and output a value. If $\delta(s_{i}) \neq \perp$, $p_{i}$ outputs $\delta(s_{i})$ and terminates. Otherwise, it continues to access this next IS object.  Therefore, each {\em NIIS protocol} is determined by a decision map $\delta$ from a local state to output values or $\perp$. Note that in the execution of an NIIS protocol, different processes may terminate after accessing different numbers of IS objects.

A {\em configuration} $C$ \label{syb:C} consists of the contents of each shared object and the state of each process. However, since each process remembers its entire history and only process $p_{i}$ can write to the $i$-th component of each IS object, a configuration is fully determined by the states of processes in this configuration. An initial configuration consists of the input values and process ids of all processes. A process is {\em active} in a configuration if it has not terminated. A configuration is terminated if all processes have terminated.

A schedule describes the order in which processes take steps. A {\em scheduler} repeatedly chooses a set of processes that are poised to perform writeread operations on the same IS object concurrently. A {\em schedule} $\alpha$ \label{syb:alpha} is an ordered sequence of sets of processes chosen by the scheduler. Given a configuration $C$ and a set $P$ of processes that are poised to access the same IS object, the {\em resulting configuration} $CP$ is the configuration reached from $C$ by having each process whose identifier is in the set write its current state to the IS object and then having all of these processes read the IS object. Let $C$ be a reachable configuration in which all active processes have accessed the same number of IS objects. For any set $P$ of processes, a {\em $P$-only 1-round} schedule from $C$ is an ordered partition of processes in $P$ that are active in $C$. A protocol is wait-free if there is no infinite schedule from any initial configuration. A {\em $P$-only r-round} schedule from $C$ is a schedule $\alpha_{1}\alpha_{2} \cdots \alpha_{r}$ such that each $\alpha_{i}$ is a $P$-only 1-round schedule from $C\alpha_{1}\cdots\alpha_{i - 1}$. A {\em full $r$-round schedule} from $C$ is a $P$-only r-round schedule from $C$ where $P = \Pi$. 

The NIIS model is known to be equivalent to the standard shared memory model for deterministic, wait-free compututation. This means that for a task, if there is a deterministic, wait-free protocol in one model, then there is a deterministic, wait-free  protocol in the other model.

\subsection{Topological concepts}
An {\em (abstract) simplex} is the set of all subsets of some finite set. 
There is a natural geometric interpretation of an (abstract) simplex. In this paper, we use the two definitions interchangeably. A vertex $\vec{v}$ is a point in Euclidean space. A set $\{ \vec{v}_0, \ldots \vec{v}_n \}$  of vertices is affinely independent if and only if the vectors $\{ \vec{v}_1 - \vec{v}_0, \ldots \vec{v}_n - \vec{v}_0 \}$ are linearly independent. An $n$-simplex $S$ \label{syb:simplex} spanned by $\{ \vec{v}_0, \ldots \vec{v}_n \}$ is defined to be the set of all points x such that $x = \sum_{i = 0}^{n}t_{i}\vec{v}_{i}$ where $\sum_{i = 0}^{n}t_{i} = 1$ and $t_{i} \geq 0$ for all $i$. The {\em dimension} of such a simplex $S$ is defined as $n$. Any simplex $T$ spanned by a subset of $\{ \vec{v}_0, \ldots \vec{v}_n \}$ is called a {\em face} of $S$. Faces of a simplex $S$ different from $S$ are called proper faces of $S$. 

An {\em (abstract) simplicial complex} \label{syb:complex} is a finite collection $\mathcal{K}$ of sets that is closed under subset: for any set $S \in \mathcal{K}$, if $S^{'} \subseteq S$, then $S^{'} \in \mathcal{K}$. The {\em dimension} \label{syb:dim} of a complex $\mathcal{K}$, denoted by dim($\mathcal{K}$), is the highest dimension of its simplices. A simplex $S$ in $\mathcal{K}$ is a {\em facet} if it is not a proper face of any simplex in $\mathcal{K}$. The dimension of $\mathcal{K}$ is the maximum dimension of any of its facets.
If $\mathcal{L}$ is a subcollection of simplices in $\mathcal{K}$ that is also closed under containment and intersection, then $\mathcal{L}$ is a complex called a {\em subcomplex} of $\mathcal{K}$. The set of simplices of $\mathcal{K}$ of dimension at most $\ell$ is a subcomplex of $\mathcal{K}$, called the $\ell{-}skeleton$ of $\mathcal{K}$, denoted by $skel^{\ell}(\mathcal{K})$. For example, elements of $skel^{0}(\mathcal{K})$ are just the vertices of $\mathcal{K}$. 

There are two standard constructions that characterize the neighborhood of a vertex or a simplex: the star and the link. The {\em star} of a simplex $ S \in \mathcal{K}$, denoted as $St(S, \mathcal{K})$, consists of all simplices $S^{'}$ that contain $S$ and all simplices included in such a simplex $S^{'}$. The link of a simplex $S \in \mathcal{K}$ \label{syb:link}, denoted by $lk(S, \mathcal{K})$, is the subcomplex of $\mathcal{K}$ consisting of all simplices in $St(S, \mathcal{K})$ that do not share common vertices with $S$.

Now, we define a way to join simplices. Let $S$ be the simplex spanned by $(\vec{s}_{0}, \ldots \vec{s}_{p})$ and $T$ be the simplex spanned by $(\vec{t}_{0}, \ldots \vec{t}_{q})$. Then the {\em joining} of $S$ and $T$, \label{syb:joining} denoted as $S * T$, is the simplex constructed from the set of vertices $(\vec{s}_{0}, \ldots \vec{s}_{p}, \vec{t}_{0}, \ldots \vec{t}_{q})$. Note that our definition is slightly different from the usual definition of joining two simplices which requires the vertices of $(\vec{s}_{0}, \ldots \vec{s}_{p}, \vec{t}_{0}, \ldots \vec{t}_{q})$ to be affinely independent. We allow a vertex to appear in both simplices $S$ and $T$ to simplify some expressions.

Let $\mathcal{K}$ and $\mathcal{L}$ be complexes, possibly of different dimensions.
A {\em vertex map} $\mu: skel^{0}(\mathcal{K}) \rightarrow skel^{0}(\mathcal{L}) $ is a map that takes the vertices of $\mathcal{K}$ to the vertices of $\mathcal{L}$. {\em A simplicial map} is a vertex map that carries each simplex of $\mathcal{K}$ to a simplex of $\mathcal{L}$.
A simplicial map $\mu:\mathcal{K} \rightarrow \mathcal{L}$ is {\em non-collapsing} if it preserves dimension, that is, for all $S \in \mathcal{K}: dim(\mu(S)) = dim(S)$. An important non-collapsing simplicial map is a {\em coloring} of an $n$-dimensional complex $\mathcal{K}$ that labels the vertices of the complex with a value in ${0,1,\ldots, n}$ (denoting a process id) so that no two neighboring vertices have the same color. A {\em chromatic complex} or {\em colored complex} $(\mathcal{K}, \zeta_{\mathcal{K}})$ is a complex $\mathcal{K}$ together with a coloring $\zeta_{\mathcal{K}}$. A simplicial map between colored complexes is {\em color-preserving} if it maps each vertex to a vertex of the same color.

A geometric complex $\sigma(\mathcal{K})$ is a {\em subdivision} of a geometric complex $\mathcal{K}$ if 
\begin{itemize}
    \item[-] Each simplex of $\sigma(\mathcal{K})$ is contained in a simplex of $\mathcal{K}$;
    \item[-] Each simplex of $\mathcal{K}$ is the union of finitely many simplices of $\sigma(\mathcal{K})$.
\end{itemize}
Let $S$ be a simplex in a subdivision $\sigma(\mathcal{K})$ of $\mathcal{K}$. {\em The carrier} of $S$ in $\mathcal{K}$ \label{syb:carrier}, denoted by $carrier(S,\allowbreak \mathcal{K})$, is the smallest simplex $T$ of $\mathcal{K}$ such that $S$ is contained in $T$.
A chromatic complex $(\sigma(\mathcal{K}), \zeta_{\sigma(\mathcal{K})})$ is a {\em chromatic subdivision} of $(\mathcal{K}, \zeta_{\mathcal{K}})$ if $\sigma(\mathcal{K})$ is a subdivision of $\mathcal{K}$ and for all $S \in \sigma(\mathcal{K})$, $\zeta_{\sigma(\mathcal{K})}(S)$ $\subseteq$ $\zeta_{\mathcal{K}}(carrier(S, \mathcal{K}))$. The {\em standard chromatic subdivision} is the most important type of chromatic subdivision that we will use. Let $(\mathcal{K}, \zeta_{\mathcal{K}})$ be a chromatic simplicial complex. Its {\em standard chromatic subdivision} \label{syb:Ch} $Ch(\mathcal{K})$ is the simplicial complex whose vertices have the form $(i, S_{i})$, where $i \in \{0, \ldots n\}$, $S_{i}$ is a non-empty face of some simplex $S \in \mathcal{K}$, and $i \in \zeta_{\mathcal{K}}(S_{i})$. $(k + 1)$ vertices $((i_{0}, S_{i_{0}}), \ldots , (i_{k}, S_{i_{k}}))$, where each $i_{j}$ is an index in $\{0, \ldots n\}$, forms a simplex of $Ch(\mathcal{K})$ if and only if

\begin{itemize}
    \item[-] There exists a permutation map $M$ from $\{i_{0}, \ldots i_{k}\}$ to $\{i_{0}, \ldots i_{k}\}$ such that $S_{M(i_{0})} \subseteq \ldots \subseteq S_{M(i_{k})}$; 
    \item[-] For $0 \leq i_{j_{1}}, i_{j_{2}} \leq n$, if $i_{j_{1}} \in \zeta_{\mathcal{K}}(S_{i_{j_{2}}})$, then $S_{i_{j_{1}}} \subseteq S_{i_{j_{2}}}$.
\end{itemize}
and to make the subdivision chromatic, we define the coloring $\zeta_{Ch(\mathcal{K})}$ as $\zeta_{Ch(\mathcal{K})}(i, S_{i}) = i$. We give an example of the chromatic subdivision in Figure \ref{img:introductions:standard_chromatic_subdivisions}.

\begin{figure}[h]
  \centering
  \includegraphics[width=0.5\linewidth]{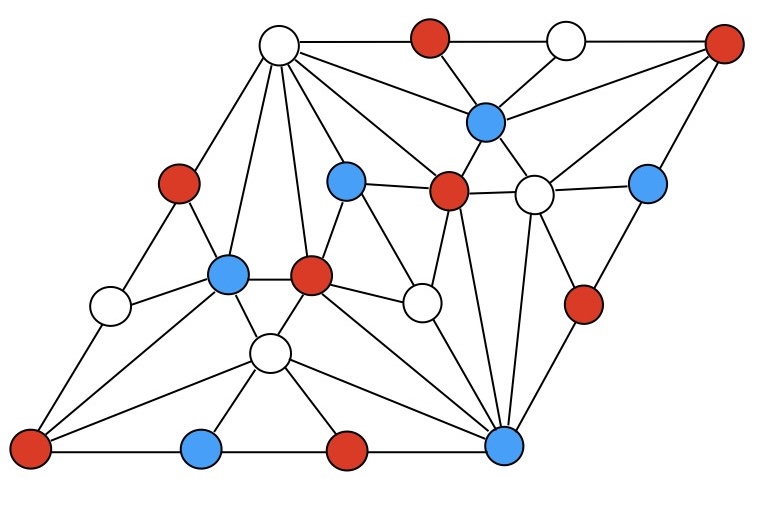}
  \caption{The standard chromatic subdivision of a 2-dimensional complex}
  \label{img:introductions:standard_chromatic_subdivisions}
\end{figure}

\subsection{Topological representation of tasks}
For a (colored) task $(\mathcal{I}, \mathcal{O}, \Delta)$, all possible input or output values can be represented by a simplicial complex, called an {\em input complex} $\mathcal{I}$ or an {\em output complex} $\mathcal{O}$. Each vertex $\vec{s}$ of these simplices is labeled with a process id and a value that are denoted by $ids(\vec{v})$ and $vals(\vec{s})$, respectively. An input assignment for a set $P$ of processes, represented by a simplex, is a set of pairs $\{(p_{j},v_{j})\}$, where each process $p_{j} \in P$ appears exactly once, but the input values $v_j$ need not be distinct. For a simplex $S$, $ids(S)$ is defined as $\{ p_{j} | (p_{j}, v_{j}) \in S\}$ and $vals(S)$ is defined as $\{ v_{j} | (p_{j}, v_{j}) \in S\}$.
Given two simplicial complexes $\mathcal{K}$ and $\mathcal{L}$, a {\em carrier map} $\Phi$ from $\mathcal{K}$ to $\mathcal{L}$ takes each simplex $S \in \mathcal{K}$ to a subcomplex $\Phi(S) $ of $\mathcal{L}$ such that for all $S, S^{'} \in \mathcal{K}$ such that $S \subseteq S^{'}$, we have $\Phi(S) \subseteq \Phi(S^{'})$.
The {\em topological task specification} corresponding to the task specification $\Delta$ is defined as a carrier map that carries each simplex $S$ of the input complex to a subcomplex of the output complex, which contains all valid output values when the input values are represented by $S$. We represent a decision task by a tuple $(\mathcal{I}, \mathcal{O}, \Delta)$ consisting of an input complex $\mathcal{I}$, an output complex $\mathcal{O}$, and a carrier map $\Delta$. 

In this paper, we are particularly interested in a subset of tasks called {\em colorless tasks}. A colorless input assignment is defined by an input assignment by discarding the process names. Note that the different occurrences of each input value do not matter. A colorless output assignment is defined by analogy with (colorless) input assignments. Similarly, all possible colorless input or output values can be represented by a simplicial complex, called an {\em colorless input complex} $\mathcal{I}^{*}$ or an {\em colorless output complex} $\mathcal{O}^{*}$. A colorless task $(\mathcal{I}^{*}, \mathcal{O}^{*}, \Delta^{*})$ is characterized by a colorless input complex $\mathcal{I}^{*}$, a colorless output complex $\mathcal{O}^{*}$, and a relation $\Delta^{*}$ that maps each simplex $\sigma \in \mathcal{I}^{*}$ to a subcomplex $\Delta^{*}(\sigma)$ of $\mathcal{O}^{*}$ such that $dim(\Delta^{*}(\sigma)) \leq dim(\sigma)$ and for all $S, S^{'} \in \mathcal{I^{*}}$ where $S \subseteq S^{'}$, we have $\Delta^{*}(S) \subseteq \Delta^{*}(S^{'})$. Colorless tasks are a central family of tasks in distributed computing. The consensus task, the set agreement task, and the approximate agreement task are examples of colorless tasks. 

As we mentioned earlier, we regard general tasks as a subset of colored tasks and will use the form of colored tasks. It has been shown that a colorless task defined in the form $(\mathcal{I}^{*}, \mathcal{O}^{*}, \Delta^{*})$ can be transformed into a colored form $(\mathcal{I}, \mathcal{O}, \Delta)$. We give an implementation in Appendix \ref{appendix:transformation}.

\subsection{Topological representation of an NIIS protocol}
In the NIIS model, a reachable configuration $C$ of a protocol $\delta$ is represented by a simplex $S$. The vertices of the simplex are the states of the processes in the configuration. Given a protocol $\delta$, if each process in $C$ is active, the {\em non-uniform chromatic subdivision} of $S$, denoted by $\chi(S)$ \label{syb:chi}, is defined to be the standard chromatic subdivision of $S$. Otherwise, let $T$ be the set of terminated vertices and let $\overline{T}$ be the simplex spanned by active vertices. The non-uniform chromatic subdivision of $S$ is defined as the complex, each facet of which contains the vertices of $T$ and the vertices of a facet in the standard chromatic subdivision of $\overline{T}$. The {\em non-uniform chromatic subdivision of a complex $\mathcal{K}$} is defined as $\cup_{S \in \mathcal{K}} \chi(S)$. In the example of Figure \ref{img:introductions:non_uniform_chromatic_subdivisions}, all processes are active in the configuration represented by the 2-simplex $S_{1}$. In the configuration represented by the 2-simplex $S_{2}$, the process represented by the red color is terminated, but the two other processes are active.

\begin{figure}[h]
  \centering
  \includegraphics[width=\linewidth]{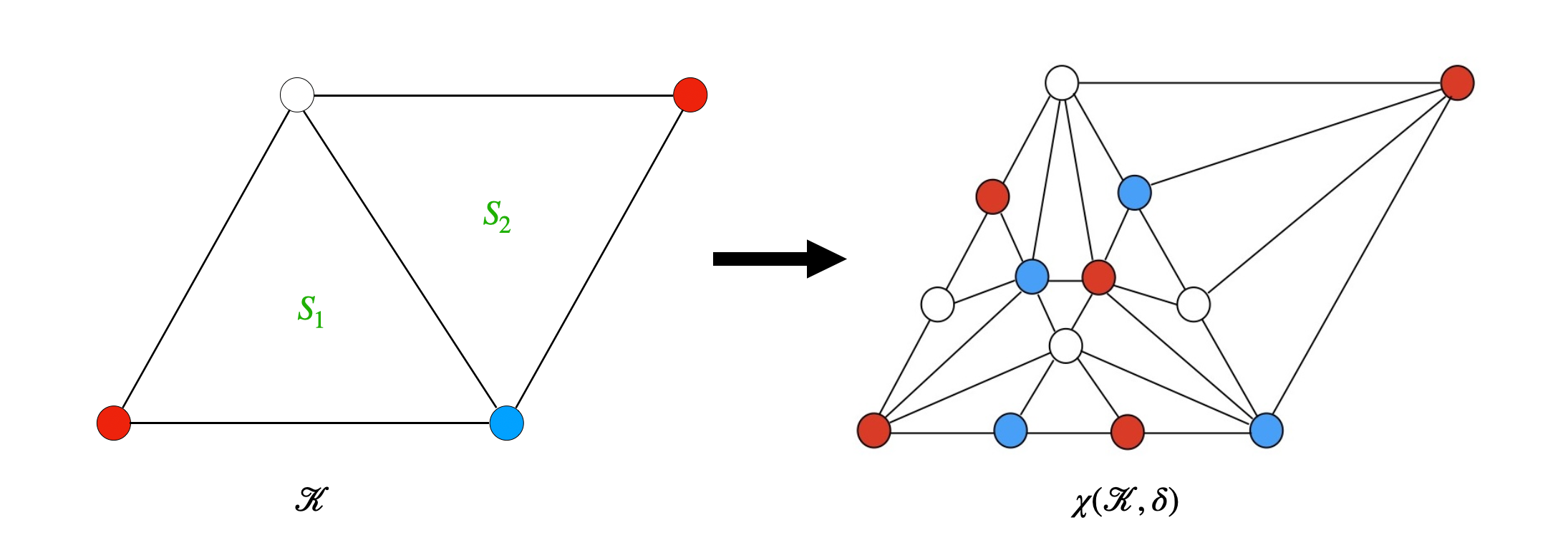}
  \caption{A non-uniform chromatic subdivision of a 2-dimensional complex}
  \label{img:introductions:non_uniform_chromatic_subdivisions}
\end{figure}

The {\em iterated standard chromatic subdivision} or the {\em iterated non-uniform chromatic subdivision}, denoted by $Ch^{n}$ or $\chi^{n}$, is the result of iterating the standard chromatic subdivision or the non-uniform chromatic subdivision n times, i.e., $Ch^{n}(\mathcal{K}) = Ch(Ch^{n  - 1}(\mathcal{K}))$ and $\chi^{n}(\mathcal{K}) = \chi(\chi^{n  - 1}(\mathcal{K}), \delta)$.

Like tasks, protocols can be represented in terms of combinatorial topology. The {\em $i$-th protocol complex} consists of all simplices, represents configurations that are reachable from some initial configuration by a $i$-round schedule. The {\em $i$-th execution map} is a carrier map that carries each initial configuration to all configurations reached from it in the $i$-th protocol complex. A {\em protocol} is represented by $(\mathcal{I}, \mathcal{P}, \Xi)$ and a simplicial map $\delta:\mathcal{P} \rightarrow \mathcal{O}$ where $\mathcal{I}$ is the input complex, $\mathcal{P}$ is the $i$-th protocol complex, and $\Xi$ is the $i$-th execution map, for some non-negative integer $i$. We say that a protocol $(\mathcal{I}, \mathcal{P}, \Xi)$ {\em solves a task} $(\mathcal{I}, \mathcal{O}, \Delta)$ if $\delta(\Xi(s^{k}))$ is in $\Delta(s^{k})$ for each $s^{k} \in \mathcal{I}$. 

Hoest and Shavit \cite{Hoest06} showed that the $i$-th protocol complex of an NIIS protocol is equal to $\chi^{i}(\mathcal{I}, \delta)$, where $\chi$ is constructed from the NIIS protocol as described above. Therefore, the map $\delta$ of an NIIS protocol is a function from $\chi^{i}(\mathcal{I}, \delta)$ to $\mathcal{O}$ for some iterated non-uniform chromatic subdivision $\chi^{i}$. This is equal to our definition of $\delta$ in Section \ref{sec:preliminaries:NIIS_model} if we let $\delta$ return $\perp$ when the input value is not a state in $\chi^{i}(\mathcal{I}, \delta)$: If a process has a state represented by a vertex $v$ in $\chi^{i}(\mathcal{I}, \delta)$, it will terminate its execution and output $\delta(v)$. Otherwise, a process has a state represented by a vertex $v^{'}$ in $\chi^{i^{'}}(\mathcal{I}, \delta)$ for some $i^{'} < i$ and is still active, because a terminated vertex in $\chi^{i^{'}}(\mathcal{I}, \delta)$ will also be in $\chi^{i}(\mathcal{I}, \delta)$. Therefore, $\delta(v^{'}) = \perp$ and the process will continue its execution.

Now we introduce a core concept in this paper, which is a variant of the concept defined above in the NIIS model. Let $U$ be any simplex in $\mathcal{I}$. \label{syb:U} A \emph{partial protocol $\delta_{U}$} with respect to $U$ \label{syb:delta_U} specifies whether a process should output a value(and which output if so) in each configuration reached from an initial configuration that contains $U$, by a schedule in which $ids(U)$ is the first set of processes. A partial protocol with respect to $U$ is {\em wait-free} if it has no infinite execution in which $ids(U)$ is the first set of processes to take a step.
For convenience, we require that no process terminates in a partial protocol before taking at least one step.

The i-th {\em protocol complex of a partial protocol $\delta_{U}$ with respect to $U$} is defined as follows. 
\begin{itemize}
\item
$\mathbb{F}_{0}(U)$ is the set of all simplices in $\mathcal{I}$ that contain $U$. 
\item
$\mathbb{F}_{1}(U)$ is the subcomplex of $\chi(\mathbb{F}_{0}(U),\delta_{U})$ consisting of all simplices representing configurations reachable via 1-round schedules in which the processes in $ids(U)$ have the input values $vals(U)$ and $ids(U)$ is the first set of processes to take a step.
\item
For $i \geq 1$,  $\mathbb{F}_{i+1}(U) = \chi(\mathbb{F}_{i}(U), \delta_{U})$
consists of all simplices representing configurations reachable via $(i+1)$-round schedules in which the processes in $ids(U)$ have the input values $vals(U)$ and $ids(U)$ is the first set of processes to take a step.
\end{itemize}
\label{syb:mathbb_F_U}

\begin{figure}[h]
  \centering
  \includegraphics[width=\linewidth]{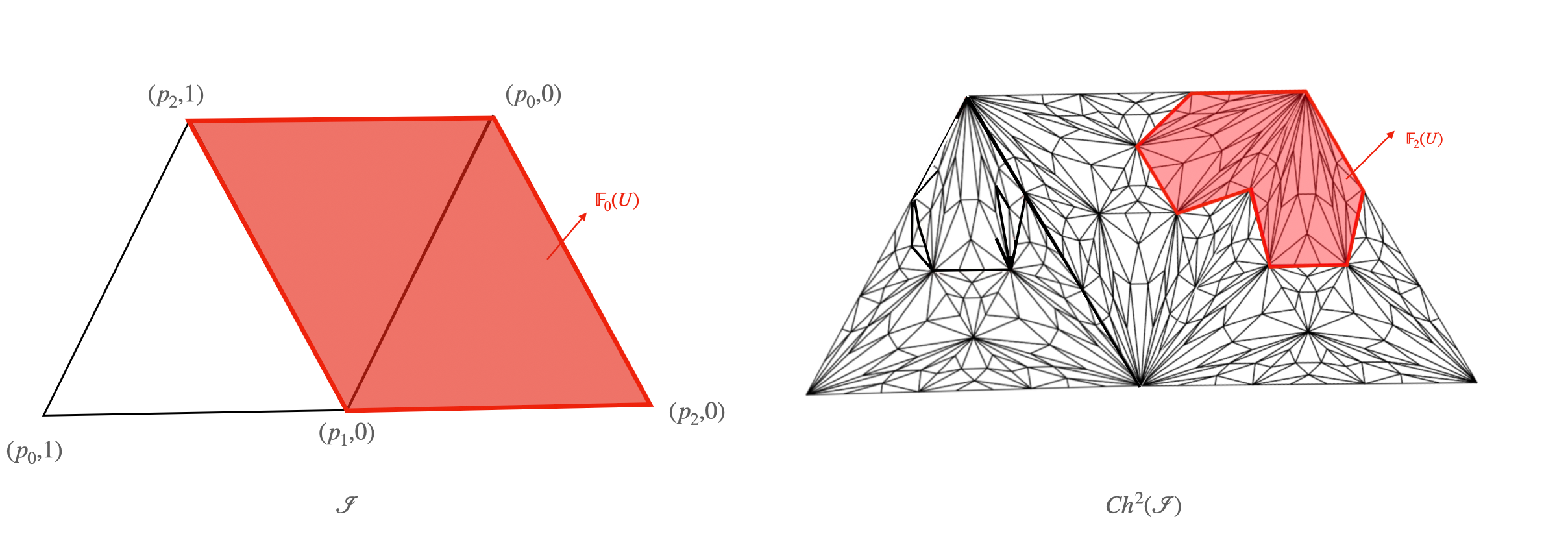}
  \caption{The protocol complexes $\mathbb{F}_{0}(U)$ and $\mathbb{F}_{2}(U)$ with respect to $U = \{(p_{0}, 0)\}$}
  \label{img:introduction:protocol_complex}
\end{figure}

In Figure 3, we give an example in which $U = \{(p_0, 0)\}$ and no processes terminate in the first two rounds.

Similarly, a partial protocol with respect to $U$ can be represented topologically as $(\mathbb{F}_{0}(U), \mathbb{F}_{i}(U), \Xi)$. The $i$-th {\em execution map} $\Xi$ is a carrier map that carries each initial configuration in $\mathbb{F}_{0}(U)$ to all configurations reached from it in $\mathbb{F}_{i}(U)$.
We say that a partial protocol $\delta_{U}$ with respect to $U$ satisfies the task specification $\Delta$ if $\delta_{U}(\Xi(s^{k}))$ is in $\Delta(s^{k})$ for each $s^{k} \in \mathcal{I}$ where $\Xi(s^{k})$ is not empty.

Let $[U_{1}, U_{2} \cdots U_{k}]$ be an ordered sequence of simplices in $\mathcal{I}$, such that $ids(U_{1}), ids(U_{2}) \cdots ids(U_{k})$ is a prefix of a full $r$-round schedule $\alpha$ for some integer $r$ and all $U_{i}$ are subsimplices of some $n$-simplex $s^{n}$ in $\mathcal{I}$. Let $U_{sum}$ be the simplex in $\mathcal{I}$ whose ids are in $\cup_{1}^{k}(ids(U_{i}))$ such that each process has the same input value as in $s^{n}$. A \emph{partial protocol with respect to $[U_{1}, U_{2} \cdots U_{k}]$}, denoted by $\delta_{[U_{1}, U_{2} \cdots U_{k}]}$, \label{syb:delta_U_1_2} specifies the state of each process after each execution starting from each initial configuration that contains $U_{sum}$ by a schedule in which $ids(U_{i})$ is the $i$-th set of processes to take a step. The i-th {\em protocol complex of a partial protocol with respect to $[U_{1}, U_{2} \cdots U_{k}]$ } \label{syb:mathbb_F_U_1_2} is defined in the same way as above.

\begin{itemize}
\item
$\mathbb{F}_{0}([U_{1}, U_{2} \cdots U_{k}])$ is the set of all simplices in $\mathcal{I}$ that contain $U_{sum}$.
\item
For $1 \leq i \leq r$, $\mathbb{F}_{i}([U_{1}, U_{2} \cdots U_{k}])$ is the subcomplex of $\chi(\mathbb{F}_{0}([U_{1}, U_{2} \cdots U_{k}]), \delta_{[U_{1}, U_{2} \cdots U_{k}]})$ consisting of all simplices representing configurations reachable via schedules in which the processes in $ids(U_{i})$ have the input values $vals(U_{i})$ and $ids(U_{i})$ is the $i$-th set of processes to take a step for each $1 \leq i \leq k^{'}$, where $U_{1}, U_{2} \cdots U_{k^{'}}$ is the longest prefix of the first full $i$-round schedule of $\alpha$.
\item
For $i > r$,  $\mathbb{F}_{i+1}([U_{1}, U_{2} \cdots U_{k}]) = \chi(\mathbb{F}_{i}([U_{1}, U_{2} \cdots U_{k}]), \delta_{[U_{1}, U_{2} \cdots U_{k}]})$.
\end{itemize}

Herlihy and Shavit\cite{Herlihy99} proved the asynchronous computability theorem, a necessary and sufficient condition for a task to be solvable in a wait-free manner by asynchronous processes that communicate by reading and writing a shared memory. 
\begin{theorem}[Asynchronous computability theorem]
A decision task $(\mathcal{I}, \mathcal{O}, \Delta)$ has a wait-free protocol in the read-write memory model if and only if there exists a chromatic subdivision $\sigma$ of $\mathcal{I}$ and a color-preserving simplicial map $\mu:\sigma(\mathcal{I}) \rightarrow \mathcal{O}$ such that for each simplex S in $\sigma(\mathcal{I})$, $\mu(S) \in \Delta(carrier(S, \mathcal{I}))$.
\end{theorem}

\section{Related work}
\label{sec:related_work}

\subsection{Extension-based proofs}
In \cite{Alistarh19}, Alistarh, Aspnes, Ellen, Gelashvili and Zhu formally defined the class of {\em extension-based proofs} that generalizes the impossibility proof techniques used in the FLP impossibility result. The proof of impossibility is modeled as an interaction between a prover and an NIIS protocol that claims to solve the task.
The protocol is defined by a map $\delta$ from the process states to the output values or a special value $\perp$. We do not use this notation $\Delta$ from \cite{Alistarh19} to denote this map, since $\Delta$ is also commonly used in the topological representation of distributed tasks. Instead, in this paper, we will use the lowercase $\delta$ to represent this map. 
The prover asks queries to learn information about the protocol, such as the $\delta$ values of the states of processes in some reached configuration, in an effort to find a violation of the task specification or construct an infinite execution.

If there exists a prover that can defeat any protocol that claims to solve the task, we say that the task has an extension-based impossibility proof. But if for each prover, an adversary can adaptively design a protocol based on the queries made by the prover such that the prover cannot win in the interaction, there is no extension-based proof for the impossibility of the task. The adaptive protocol constructed by the adversary will be referred to as an adversarial protocol as in \cite{Alistarh19}.

There are three classes of extension-based proofs (restricted extension-based proofs, extension-based proofs, and augmented extension-based proofs) defined in \cite{Alistarh19, Brusse21}. We will first introduce the restricted extension-based proof, which is the weakest one. The interaction between the prover and a protocol proceeds in phases. In each phase $\varphi$ the prover starts with a finite schedule $\alpha(\varphi)$ and a set of configurations $\mathcal{A}(\varphi)$ \label{syb:initial_configurations_of_phase_varphi} reached from some initial configurations of the task by $\alpha(\varphi)$ \label{syb:initial_schedule_of_phase_varphi} that only differ in the input values of processes that are not in the schedule $\alpha$. For $\varphi = 1$, the schedule $\alpha(1)$ is the empty schedule and $\mathcal{A}(1)$ contains all initial configurations. The set of configurations reached in phase $\varphi$ is denoted by $\mathcal{A}^{'}(\varphi)$ \label{syb:reached_configurations_in_phase_varphi} which is empty at the beginning of phase $\varphi$. The prover queries the protocol by choosing a configuration $C \in \mathcal{A}(\varphi) \bigcup \mathcal{A}^{'}(\varphi)$ and a set of processes $P$ that are poised to apply writeread to the same IS object in $C$. The protocol replies with the value of $\delta$ of each process in $P$ in the resulting configuration $C^{'}$ and the prover adds the configuration $C^{'}$ to $\mathcal{A}^{'}(\varphi)$.  A chain of queries is a sequence of queries (finite or infinite) such that if $(C_{i}, P_{i})$ and $(C_{i+1}, P_{i+1})$ are consecutive queries in the chain, then $C_{i+1}$ is the configuration resulting from scheduling $P_{i}$ from $C_{i}$. The prover is allowed to construct finitely many chains of queries in a phase.

When the output values from the response of a query are not valid output values, we say that the prover finds a violation of the task specification. For example, if a prover finds $(k + 1)$ different values(excluding $\perp$) in the response of a query to a protocol that claims to solve the $(n, k)$-set agreement task, the protocol must admit that it cannot solve the $(n, k)$-set agreement task. Similarly, if a prover can construct an infinite chain of queries or infinite phases, the protocol must admit that it is not wait-free. In this case, we say that the prover finds a liveness issue.

If the prover does not find a violation or construct an infinite execution, it must end the phase and choose some configuration $C^{'} \in \mathcal{A}^{'}(\varphi)$. Let $\alpha^{'}$ be the schedule such that $C^{'}$ is reached from some configuration in $\mathcal{A}(\varphi)$. The schedule $\alpha(\varphi)\alpha^{'}$ will be used as the initial schedule in the next phase. In other words, the prover sets $\alpha(\varphi + 1) = \alpha(\varphi)\alpha^{'}$. Suppose that the configuration $C^{'}$ is reached from an initial configuration $C_{ini}$. Then $\mathcal{A}(\varphi + 1)$ consists of all configurations reached by the schedule $\alpha(\varphi + 1)$ from initial configurations that only differ from $C_{ini}$ by the input values of those processes that do not appear in $\alpha(\varphi + 1)$.

At the start of a phase, if every process has terminated in all configurations of $\mathcal{A}(\varphi)$,  no extension will be possible, indicating the end of the interaction and failure of the prover. But if the prover manages to construct an infinite chain of queries or an infinite number of phases, the prover wins.

Extension-based proofs allow for an additional type of query. The prover performs output queries by choosing a configuration $C \in \mathcal{A}(\varphi) \bigcup \mathcal{A}^{'}(\varphi)$, a set of processes $P$ that are poised to access the same IS object, and a value $y \in \{0, 1,..., k\}$. If there is a $P$-only schedule from $C$ that results in a configuration in which a process in $P$ outputs $y$, the protocol will return some such schedule. Otherwise, the protocol returns $NULL$. 
Augmented extension-based proofs adopt a more general query, called an assignment query. An assignment query consists of a configuration $C \in \mathcal{A}(\varphi) \bigcup \mathcal{A}^{'}(\varphi)$, a set of processes $P$, and an assignment function $f$ from a subset $Q$ of $P$ to the output values. If there exists a $P$-only schedule from $C$ that results in a configuration in which the output value of each process $q$ in $Q$ is $f(q)$, the protocol will return some such schedule. Otherwise, the protocol will return $NULL$. It is shown that an output query $(C, P, y)$ can be simulated by a sequence of assignment queries $(C, P, f_{q})$ for each process $q$ in \cite{Brusse21}. Therefore, we will discuss only augmented extension-based proofs in subsequent sections.


\subsection{Limitations of extension-based proofs}

In \cite{Alistarh19}, Alistarh, Aspnes, Ellen, Gelashvili and Zhu proved that there are no extension-based proofs of the impossibility of solving the $(n, k)$-set agreement task in the NIIS model. They designed an adversarial strategy that constructs an adaptive protocol according to the queries made by a prover. A key observation is that the prover cannot learn all the information about the protocol during phase 1, since it is only allowed to make a finite number of queries. The adversary designs a protocol in which the output values of some (but not all) configurations are defined in phase 1. For comparison, a real protocol defines the output values of all possible configurations. Things will change when the prover has to end phase 1, since the configurations that can be reached in later interactions are now limited. In \cite{Alistarh19} it is proved that after phase one of the interaction, the adversary can defeat any extension-based prover, since it can now define a partial protocol for the reachable configurations in the rest of the interaction. This is why a partially specified protocol can exist even if the task actually does not have a wait-free protocol. This observation is essential for understanding extension-based impossibility proofs for a task, as we shall discuss later.

Alistarh, Ellen, and Rybicki \cite{Alistarh21} proved that there is no extension-based proof of the impossibility of solving the approximate agreement task on the 4-cycle in the NIS model. They proposed a new strategy to construct a partial protocol after phase 1.

\section{Motivation and summary}
\label{sec:motivation_and_summary}
In this section, we give a description of the motivation behind Section \ref{sec:nccondition}, in which we give a necessary and sufficient condition for a colorless task defined by $(\mathcal{I}, \mathcal{O}, \Delta)$ to have no extension-based proofs.

The $(n, k)$-set agreement task is the first task that was shown to have no extension-based impossibility proofs. As shown in \cite{Alistarh19}, given any extension-based prover, the adversary will pretend to have a protocol for the $(n, k)$-set agreement task during phase 1 of the interaction. But after the prover chooses a schedule at the end of phase 1, the adversary can assign a valid output value to each undefined configuration that the prover can reach in the later phases. In other words, the adversary has a partial protocol compatible with the existing assigned values that satisfies the task specification of the $(n, k)$-set agreement after phase 1.

Before we present our idea behind this necessary and sufficient condition, it is better to explain how an adversary wins against any extension-based prover in our terminology. 
We divide the adversarial strategy into two parts: In this first part, the adversary adaptively defines a protocol in response to any specific prover's queries during the first $r$ phases. In the second part, the adversary designs a partial protocol after the end of phase $r$ so that the prover is doomed to lose. By assumption, the task $(\mathcal{I}, \mathcal{O}, \Delta)$ cannot be solved by a wait-free protocol (otherwise it is meaningless to talk about extension-based proofs.). However, there may be a protocol when restricted to a limited set of executions. This is exactly what happens after the end of each phase. When the prover chooses a configuration $C^{'}$ in $\mathcal{A}^{'}(\varphi)$ and ends a phase $\varphi$, any possible execution in later phases is reached from some initial configuration by a schedule that is an extension of $\alpha(\varphi + 1)$. Suppose that $C^{'}$ is reached from an initial configuration $C$. These limited possible configurations can be represented by a protocol complex of a partial protocol with respect to $[U_{1}, U_{2} \cdots U_{k}]$, where $U_{i}$ is a subsimplex of the $n$-simplex representing $C$ and $ids(U_{i})$ is the $i$-th set of processes in $\alpha(\varphi + 1)$. The assignment of a valid value of $\delta$ to each vertex in this protocol complex is equivalent to building a partial protocol with respect to $[U_{1}, U_{2} \cdots U_{k}]$.

\subsection{The definition of finalization}
If the adversary can prevent the prover from finding a problem in the first $r$ phases and construct a partial protocol after phase $r$, no matter what queries the prover makes during the first r phases and which configurations the prover has chosen at the end of the first r phases, we say that the adversary can {\em finalize after phase $r$}.
We start by showing that the adversary can win against any extension-based prover, if and only if the adversary can finalize after phase $r$ for some positive integer r. 

Suppose that the adversary can finalize after some phase $r$. Then, eventually, the prover chooses a configuration at the end of some phase $r^{'} > r$ in which every process has terminated. The prover loses in the next phase.

Suppose that the adversary cannot finalize after phase $r$ for each integer $r$. This means that the prover is able to choose a configuration $C^{'}$ that is reached from some initial configuration $C$ by the schedule $\alpha(r + 1)$ at the end of phase $r$ for each integer $r$, such that there is no partial protocol with respect to $[U_{1}, U_{2} \cdots U_{r}]$, where $U_{i}$ is a subsimplex of the $n$-simplex representing $C$ and $ids(U_{i})$ is the $i$-th set of processes in $\alpha(r + 1)$. 
So, at the end of phase $r$, the prover can choose a configuration such that $\alpha(r + 1) = \{U_{1}, U_{2} \cdots U_{r}\}$. In phase $r + 1$, there exists a configuration in $\mathcal{A}(r + 1)$ in which some processes are not terminated. By definition, the protocol wins when at the beginning of some phase $r + 1$, every process has terminated in every configuration in $\mathcal{A}(r + 1)$. This means that the protocol constructed by the adversary cannot win.

Then we consider the case $r = 1$, i.e. a necessary and sufficient condition for a colorless task $(\mathcal{I}, \mathcal{O}, \Delta)$ to have an adversary that can win against any restricted extension-based prover by constructing a partial protocol after the end of phase 1. Most of the techniques used in the proof for larger values of r are introduced in the proof of this case, which are in Sections \ref{sec:preparations} - \ref{sec:nccondition:finalization_after_phase_1}. Then in Section \ref{sec:nccondition:augmented_extension_based_proofs}, we show that assignment queries do not give the prover more power. Therefore, the necessary and sufficient condition works for all versions of extension-based proofs.
Finally, in Section \ref{sec:nccondition:finalization_after_phase_r}, we consider the case $r > 1$. We will show that $r$ is at most $n + 1$(the number of processes). When $1 < r \leq n + 1$, we can use the techniques in case $r = 1$ to prove a similar condition.

So, the most important part of this paper is to find a necessary and sufficient condition for a colorless task $(\mathcal{I}, \mathcal{O}, \Delta)$ to have an adversary that can construct a partial protocol after phase 1. It is hard to find a necessary and sufficient condition directly since what we have is just a tuple of topological structures. We start by finding a necessary condition, make it stronger, and then show that this condition is also sufficient.

\subsection{Finalization after phase 1}
\label{sec:motivation_and_summary:finalization_after_phase_1}

\subsubsection{A necessary condition}
At first, we consider a special case where no queries are submitted by the prover during phase 1, which means that the adversary has not decided the output values of any configuration before the end of phase 1.  
A necessary condition for finalization after phase 1 is that no matter which configuration the prover chooses at the end of phase 1, the adversary can construct a partial protocol. In other words, for a task $(\mathcal{I}, \mathcal{O}, \Delta)$, for each simplex $U$ in $\mathcal{I}$, there exists a partial protocol with respect to $U$. 

What is the difference between a real protocol and a set of partial protocols? We know that there is no real wait-free protocol for the $(n, k)$-set agreement task. But as shown in \cite{Alistarh19}, there exists a set of partial protocols, each of which corresponds to a simplex in $\mathcal{I}$. A real protocol implies a set of partial protocols, but not vice versa. An important observation is that some configurations may be shared by the protocol complex with respect to $U$ and the protocol complex with respect to $U^{'} \neq U$. These configurations can be assigned different output values by different partial protocols, which is not allowed in a real protocol.

So, how can the adversary construct a partial protocol with respect to any simplex $U \in \mathcal{I}$? We now temporarily forget the existing assigned values and extension-based proofs. We will discuss a condition that characterizes the existence of a partial protocol with respect to each simplex $U$ in $\mathcal{I}$. We can use the asynchronous computability theorem to give this topological condition. Details are discussed in Section
\ifdefined\FULL
\ref{sec:preparations:general_tasks_with_partial_input_information} 
\else
5.1
\fi
in the extended paper. Let $U$ be a simplex representing the states of some processes in an initial configuration. 
Given a simplex $U$ in the input complex $\mathcal{I}$, for each facet of the subcomplex of $\mathcal{O}$ whose process ids are in $ids(U)$, we use it as a label to define a complex whose facets are the $n$-simplices in the output complex $\mathcal{O}$ that contain this facet. The output complex restricted to the simplex $U$, denoted by $\mathcal{O}_{U} = \{\mathcal{O}_{U}^{0}, \mathcal{O}_{U}^{1}, \cdots \mathcal{O}_{U}^{m}\}$, is defined as the disjoint union of the complexes whose labels are in $\Delta(U)$, if we associate each label with a unique index.

Intuitively, the existence of a partial protocol with respect to $U$ is equivalent to an assignment of output values to all configurations reached by schedules starting with processes in $ids(U)$ accessing the first IS object concurrently. We will show that we only have to consider the configurations reached by schedules starting with processes in $ids(U)$ running concurrently until termination (not just the first round). The output values of the processes in $ids(U)$ are represented by some facet of $\delta(U)$. The remaining part can be seen as a task of processes in $\Pi - ids(U)$ which has input complex $lk(U, \mathcal{I})$ and output complex $\mathcal{O}_{U}^{i}$ for some $i$.

\begin{theorem}[Theorem \ref{the:partial_info_theorem}]
    For a task $(\mathcal{I}, \mathcal{O}, \Delta)$, there exists a partial protocol with respect to any simplex in $\mathcal{I}$ if and only if for each simplex $U \in \mathcal{I}$, there exists $\mathcal{O}_{U}^{i}$ for some $0 \leq i \leq m$, a chromatic subdivision $\sigma$ of $lk(U, \mathcal{I})$ and a color-preserving simplicial map $\mu:\sigma(lk(U, \mathcal{I})) \rightarrow \mathcal{O}_{U}^{i}$,  such that for each simplex $S$ in $\sigma(lk(U, \mathcal{I}))$, $\mu(S) \in \Delta_{U}(carrier(S, lk(U, \mathcal{I})))$ where $\Delta_{U}$ carries a simplex $U^{'}$ $\subseteq$ $lk(U, \mathcal{I})$ to the complex $\Delta(U^{'} * U) \backslash U$, which is generated from $\Delta(U^{'} * U) $ by removing vertices whose ids are in $ids(U)$.
\end{theorem}

\subsubsection{An enhanced necessary condition}
The existence of a set of partial protocols is a necessary (but not necessarily sufficient) condition for finalization after phase 1. Note that processes may not terminate after the same round in such a partial protocol.  However, we can construct a new set of partial protocols by postponing the termination of processes that terminate earlier. We give an explicit construction in Section \ref{sec:preparations:canonical_neighbor} to ensure that all these partial protocols terminate after the same round $r_{m}$ \label{syb:r_m}. In this case, the union of $\mathbb{F}_{r_{m}}(U)$ for all $U \in \mathcal{I}$ is equal to $Ch^{r_{m}}(\mathcal{I})$. The standard chromatic subdivision has some good properties that we will use to build our adversarial strategy.

Now, we return to extension-based proofs. A set of partial protocols, each of which corresponds to a simplex $U$ in $\mathcal{I}$, is not enough. In the protocol complex of a partial protocol, the output values of some configurations may already be determined during the interaction of phase 1. For two simplices $U_{1}$ and $U_{2}$ in $\mathcal{I}$ and each simplex $S$ in $\mathbb{F}_{1}(U_{1}) \cap \mathbb{F}_{1}(U_{2})$, we consider the configuration,  denoted by $CEN(S)$ \label{syb:CEN_S}, reached from $S$ via a schedule that repeats the set of processes $ids(S)$ until all processes in $ids(S)$ terminate. We say that two partial protocols $\delta_{U_{1}}$ and $\delta_{U_{2}}$ are {\em compatible} if the output values of $CEN(S)$ are the same under $\delta_{U_1}$ and $\delta_{U_2}$ for each possible $S$. A set of partial protocols is {\em compatible} if any two partial protocols are compatible. An enhanced necessary condition for finalization after phase 1 is that the set of partial protocols is compatible. This can be proved by having the prover submit a chain of queries to ask the output values of $CEN(S)$ in phase 1. At the end of phase 1, the prover can choose $\alpha(2)$ to be $U_{1}$ or $U_{2}$, which means that the partial protocol with respect to $U_{1}$ or $U_{2}$, constructed by the adversary at the end of phase 1, has to assign the same output values to $CEN(S)$. An assumption here is that task $(\mathcal{I}, \mathcal{O}, \Delta)$ has finite initial configurations.

The high-level idea of our argument is as follows: we only require that the output value of {\em one} configuration ($CEN(S)$ in our construction) reached from $S$ is consistent with different partial protocols. A real protocol that solves the task can also be divided into a set of partial protocols. But in that case {\em every} configuration reached from $S$ should be consistent. Details are given in Section
\ifdefined\FULL
\ref{sec:nccondition:glue_protocols}
\else
6.1
\fi in the extended paper.

\subsubsection{The enhanced condition is also sufficient}
\label{sec:motivation_and_summary:finalization_after_phase_1:enhanced_condition_is_sufficient}

Now we prove that a task $(\mathcal{I}, \mathcal{O}, \Delta)$ has a set of compatible partial protocols $\{\delta_{U} | {U} \in \mathcal{I}\}$ then the adversary can always finalize after phase 1. We say that a partial protocol {\em has a label $U$} if it is a partial protocol with respect to $U$. Note that each partial protocol is an NIIS protocol in which processes may terminate after accessing different numbers of IS objects. We can assume without loss of generality that all processes terminate after the same round $r_{m}$ in all partial protocols (by postponing the termination of processes that terminate earlier). 

We redesign the adversarial strategy in \cite{Alistarh19}. So similarly, we use an infinite sequence of complexes $S^{0}, S^{1} ...$ \label{syb:complexes} and an integer $t$ (current complex) \label{syb:t} to represent the adaptive protocol, in which $S^{0} = \mathcal{I}$. Once the adversary has defined $\delta$ for each vertex in $S^{t}$, it performs a non-uniform chromatic subdivision of $S^{t}$ and constructs $S^{t+1} = \chi(S^{t}, \delta)$. After that, the adversary increments $t$.

Our adversary maintains several invariants in the interaction with an extension-based prover, which generalizes the invariants in \cite{Alistarh19}.
The first invariant is kept: For each $0 \leq r < t$ and each vertex $v \in S^{r}$ , $\delta(v)$ is defined. If $v$ is a vertex in $S^{t}$, then $\delta(v)$ is undefined or $\delta(v) \neq \perp$. If a vertex $s$ represents the state of a process in a configuration that has been reached by the prover, then $\delta(v)$ is defined. The second invariant is about the safety of the adaptive protocol. The output values defined by the adaptive protocol will not violate the task specification. To achieve this, the adversary defines the $\delta$ values using the output values obtained from the set of partial protocols. The third invariant is about configurations whose output values are obtained from different partial protocols. Since we now have a set of compatible partial protocols, our third invariant requires that the active distance between a configuration terminated with output values given by $\delta_{U_{1}}$ and a configuration terminated with output values given by $\delta_{U_{2}}$, where $U_{2} \neq U_{1}$ is at least 3. Here, the active distance is defined as the minimal number of edges between non-terminated vertices in any path from a subgraph to another subgraph. 

For a task $(\mathcal{I}, \mathcal{O}, \Delta)$, what we have is just a compatible set of partial protocols. Before any query submitted by the prover, the adversary sets $\delta(v) = \perp$ for each vertex v in the complexes $S^{0}, S^{1} ... S^{r_{m}}$, where $S^{i} = Ch(S^{i - 1})$ for each $i \leq r_{m}$ (note that we use uniform chromatic subdivision here, and have to use the non-uniform chromatic subdivision in later complexes). Suppose that the invariants are satisfied immediately prior to a query $(C, P)$ by the prover, where $C$ is a configuration previously reached by the prover and $P$ is a set of active processes in $C$
poised to access the same immediate snapshot object. There is an integer $0 \leq r \leq t$ such that the states of the processes in $P$ correspond to a simplex in $S^{r}$ by the first invariant. Since each process in $P$ is still active, $\delta(v) = \perp$ for each vertex $v$ that represents the state of a process in $C$, which means $0 \leq r \leq t - 1$ by the first invariant.  If $0 \leq r < t - 1$, then $\delta$ is already defined for the resulting configuration according to the first invariant. If $r = t - 1$, the adversary may have to decide the $\delta$ values of the vertices in the resulting configuration. If $\delta(v)$ is set to $\perp$ for each vertex $v$ in the resulting configuration, then the invariant (1) is maintained. But if $\delta(v)$ is set to a value that is not $\perp$ for some vertex $v$ in the resulting configuration, the adversary sets $S^{t + 1} = \chi(S^{t}, \delta)$ and increments $t$, in order to maintain the invariant (1).

The remaining discussions are about how to decide $\delta$ for a vertex in $S^{r}$ where $r > r_{m}$ while maintaining the second and third invariants. We start with a simple strategy: Each terminated vertex has a simplex $U$ of $\mathcal{I}$ as its label, indicating which partial protocol its $\delta$ value is from. If the adversary has to define $\delta$ for a {\em terminated} vertex $v$ in the complex $S^{r}$ where $r > r_{m}$, where $v$ is reached from some $n$-simplex $s^{n}$ in $S^{r_{m}}$, the adversary can use the value of $\delta_{U}(v^{'})$ where $s^{n} \in \mathbb{F}_{r_{m}}(U)$ and $v^{'}$ is the vertex of $s^{n}$ with the same process id as $v$. The output value of $v$ is fully determined by the simplex $U$. We say that the vertex $v$ has a possible label $U$. Note that some vertex may have multiple possible labels. Suppose that the adversary is going to decide on the assigned value for a vertex $v$ in $S^{t}$ ($t$ is the current complex). If $v$ has a possible label $U$ and the active distance between $v$ and any configuration terminated with a different label is at least 3, the adversary terminates $v$ with the label $U$. Otherwise, the adversary sets $\delta(v) = \perp$. 

A problem with our first strategy is that the prover can construct an infinite execution. Let $s^{k}$ be a shared simplex of a $n$-simplex $s^{n}_{1}$ of $\mathbb{F}_{1}(U_{1})$ and an $n$-simplex $s^{n}_{2}$ of $\mathbb{F}_{1}(U_{2})$. The prover can submit a chain of queries, as a result of which all processes in $ids(s^{k})$ terminate in some configuration reached from $s^{k}$ via a $ids(s^{k})$-only schedule $\beta$. If this configuration is terminated with the label $U_{1}$, consider the vertex that represents the state of some process $p \notin ids(s^{k})$ in the configuration reached from $s^{n}_{2}$ by $\beta \cdot \{p\}\cdot \{p\} \cdots$. This vertex will not be terminated by the adversary using our first strategy. So a special rule is needed to assign an output value to this type of vertex to avoid infinite executions.

We introduce the following rule: If a vertex $v$ in $\mathbb{F}_{t}(U)$, but not in $\mathbb{F}_{t}(U^{'})$, is adjacent to a terminated vertex with the label $U^{'}$, the adversary still terminates $v$ with the label $U^{'}$. But since $\delta_{U^{'}}$ is not defined directly for $v$, how can the adversary choose the output value of $v$? We use a canonical neighbor, which is the key concept of this paper, to handle this problem. If an $n$-simplex $s^{n}$ in $Ch^{r_{m}}(\mathcal{I})$ is not in $\mathbb{F}_{r_{m}}(U^{'})$, but shares a simplex $s_{s}$ with $\mathbb{F}_{r_{m}}(U^{'})$, we define an $n$-simplex in $\mathbb{F}_{r_{m}}(U^{'})$ as the canonical neighbor of $s^{n}$ with the label $U^{'}$, denoted by $N(s^{n}, U^{'})$ \label{syb:canonical_neigbhor}. The {\em canonical neighbor} $N(s^{n}, U^{'})$ of a simplex $s^{n}$ with the label $U^{'}$ is defined to have four requirements:

\begin{itemize}
    \item [1)] $N(s^{n}, U^{'})$ is in $\mathbb{F}_{r_{m}}(U^{'})$.
    \item [2)] $N(s^{n}, U^{'})$ contains the simplex $s_{s}$. (This is why we call it a neighbor.) 
    \item [3)] For each vertex $v \in s^{n} - s_{s}$, $carrier(v, \mathcal{I}) = carrier(v^{'}, \mathcal{I})$, where $v^{'}$ is the vertex with the same process id as $v$ in $N(s^{n}, U^{'})$.
    \item [4)] If $s^{n}_{1}$ and $s^{n}_{2}$ are $n$-simplices in $Ch^{r_{m}}(\mathcal{I})$ that have canonical neighbors with label $U^{'}$ (Note that $s^{n}_{1}$ and $s^{n}_{2}$ are not necessarily in the same $\mathbb{F}_{r_{m}}(U)$.), and share a vertex with the process id $p$, then $N(s^{n}_{1}, U^{'})$ and $N(s^{n}_{2}, U^{'})$ share a vertex with the process id $p$.
\end{itemize}

Suppose that a vertex $v$ is reached from some $n$-simplex $s^{n}$ in $S^{r_{m}}$. When the vertex $v$ is terminated with the label $U^{'}$ and $U^{'}$ is not a possible label of $v$, the adversary sets $\delta(v)$ as $\delta_{U^{'}}(v^{'})$ where $v^{'}$ is the vertex in $N(s^{n}, U^{'})$ having the same process id as $v$. The output values given by this special rule will not cause a safety problem, since $v^{'}$ and $v$ have the same carrier in $\mathcal{I}$. In other words, $v^{'}$ and $v$ see the same set of input values, which means that an output value for one vertex is a valid output value for the other vertex. The fourth requirement of canonical neighbors ensures that there is no conflict of the value assigned to some vertex $v$ when $v$ can be reached from multiple n-simplices in $Ch^{r_{m}}(\mathcal{I})$. We provide an algorithm for calculating canonical neighbors in Section
\ifdefined\FULL
\ref{sec:preparations:canonical_neighbor}
\else
5.2
\fi of the extended paper. An example, where $r_{m} = 2$, is given in Figure \ref{img:motivations_and_summary:canonical_neighbor}.

\begin{figure}[h]
  \centering
  \includegraphics[width=\linewidth]{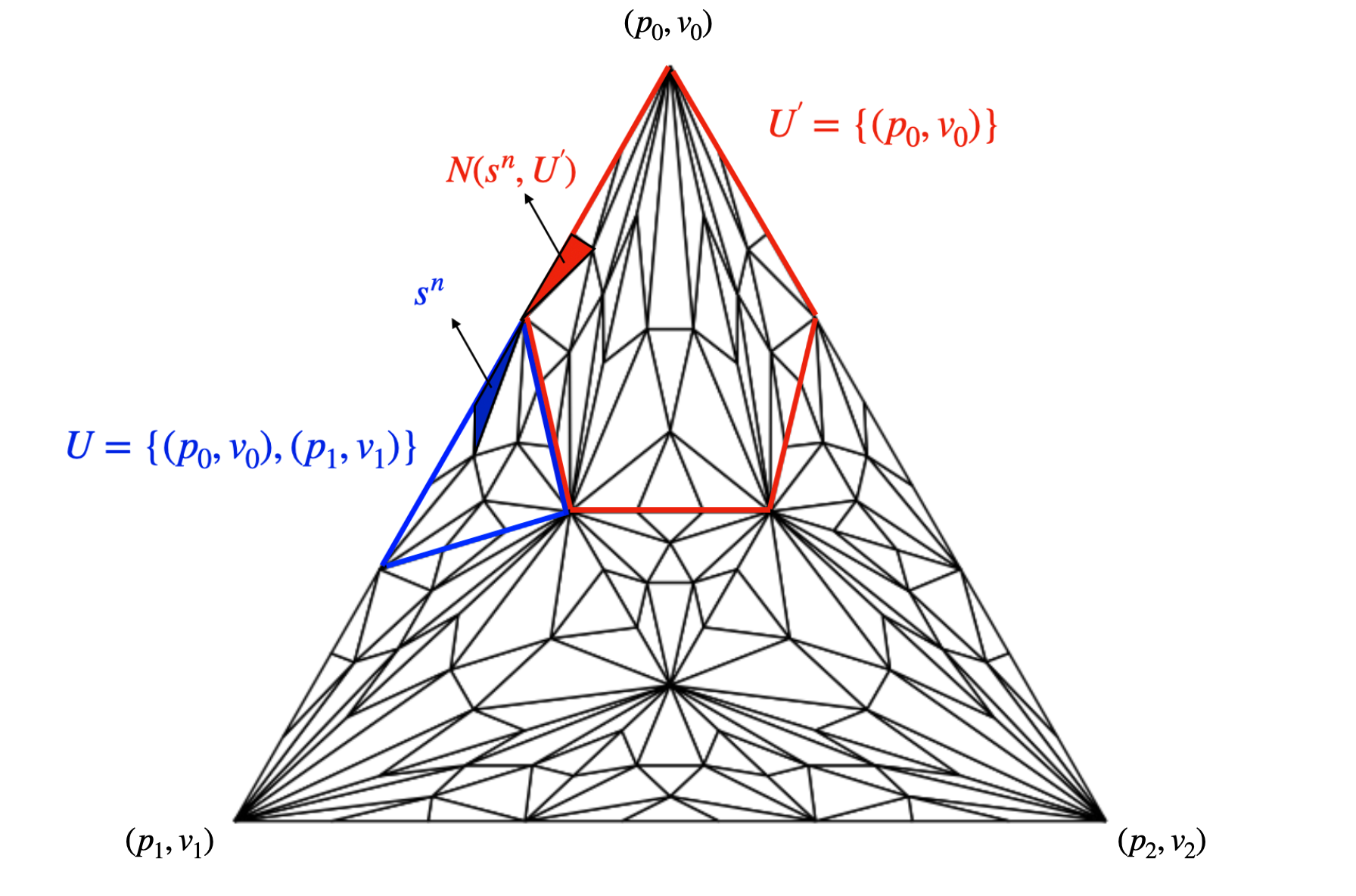}
  \caption{An example of our implementation of canonical neighbors when $r_{m} = 2$}
  \label{img:motivations_and_summary:canonical_neighbor}
\end{figure}


We emphasize that {\em not every} simplex in $Ch^{r_{m}}(\mathcal{I})$ has a canonical neighbor with the label $U^{'}$. Only those simplices that intersect with $\mathbb{F}_{r_{m}}(U^{'})$ have canonical neighbors with the label $U^{'}$. This means that we have to impose a restriction on the usage of the special rule that defines the $\delta$ value of a vertex using the canonical neighbor: If the adversary uses this rule in $S^{t}$, then it can use this rule after $r_{s}$ \label{syb:r_s} rounds of non-uniform chromatic subdivision of $S^{t}$. 

The result of this restriction is that the adversary will not terminate a vertex in the subdivision of an $n$-simplex in $Ch^{r_{m}}(\mathcal{I})$ that has no canonical neighbor with the label $U^{'}$, with the label $U^{'}$. 
A proof of why this restriction works is given in Section
\ifdefined\FULL
\ref{sec:nccondition:adversarial_strategy}
\else
6.2
\fi of the extended paper. 

Using canonical neighbors and a restriction on the special rule, the adversary sets $\delta$ values for some vertices in $S^{t}$ while leaving other vertices undefined. Any prover cannot find a violation of the task specification. And the prover cannot construct an infinite execution in phase 1, as proven in Section \ref{sec:nccondition:prover_not_win_in_first_phase}. Therefore, the prover cannot win in phase 1.

When the prover chooses a configuration $C^{'}$ that is reached from some initial configuration $C$ at the end of phase 1, the schedule $\alpha(2)$ that can be assumed to contain only one set of processes without loss of generality. Let $U$ be a simplex in $\mathcal{I}$ such that $ids(U)$ is equal to the set of processes in $\alpha(2)$ and $val(U)$ are obtained from the input values of these processes in the configuration $C$. We now present our adversarial strategy that constructs a partial protocol with respect to $U$. Note that this constructed partial protocol is different from $\delta_{U}$ since in phase 1, the adversary defines the $\delta$ values of some configurations with output values given by $\delta_{U^{'}}$ in the subdivision of $n$-simplices in $\mathbb{F}_{r_{m}}(U)$ that have an intersection with $\mathbb{F}_{r_{m}}(U^{'})$. How can the adversary define the $\delta$ values of the remaining undefined vertices in $\mathbb{F}_{t}(U)$? $\delta_{U}(s^{n})$ and $\delta_{U^{'}}(N(s^{n}, U^{'}))$ do not necessarily have the same output value for even one process. But we know that $\delta_{U}$ and $\delta_{U^{'}}$ are assumed to be compatible, i.e. the output of a simplex $CEN(S)$ in $\mathbb{F}_{r_{m}}(U) \cap \mathbb{F}_{r_{m}}(U^{'})$ is identical using $\delta_{U}$ or $\delta_{U^{'}}$. 

In order to give a motivation behind our adversarial strategy, we consider a simple case where a set $S$ of vertices terminated with the label $U^{'}$ is in the subdivision of a $n$-simplex $s^{n} \in \mathbb{F}_{r_{m}}(U)$. We say that a vertex of $S$ is in the outermost layer of $S$ if it is adjacent to some undefined vertices. The adversary will terminate all the vertices adjacent to the outermost layer using some output values, after which these newly terminated vertices become the outermost layer of $S$. The adversary will repeat this addition many times until the output values of the outermost layer of $S$ are $\delta_{U}(s^{n})$. Due to the restriction that we have imposed on the special rule, $s^{n}$ has an intersection $s_{s}$ with $\mathbb{F}_{r_{m}}(U^{'})$. The adversary can construct a $dim(CEN(S))$-dimensional path $P_{1} = s_{0}, s_{1} \cdots s_{e}$ in $\mathbb{F}_{r_{m}}(U)$ from $s^{n}$ to $CEN(S)$ for some simplex $S \in \mathbb{F}_{1}(U) \cap \mathbb{F}_{1}(U^{'})$, where $s_{0}$ is some subsimplex of $s^{n}$ and $s_{e}$ is $CEN(S)$. We can prove that, in the configuration represented by $s^{n}$, the processes in $ids(s_{0})$ have the minimum carrier in $\mathcal{I}$ and each process has the same carrier in $\mathcal{I}$ in the configurations representing $s_{0}, s_{1} \cdots s_{e}$. The simplex $S \in \mathbb{F}_{1}(U)$ is chosen according to some rule using $s^{n}$ as an argument. The adversary then constructs the same path $P_{2} = s_{e}, s_{e - 1} \cdots s_{0}$ from $CEN(S)$ to $s^{n}$. The details are given in Section \ref{sec:nccondition:finalization_after_phase_1}. Note that the dimension of $s^{n}$ is $n$, while the dimension of $CEN(S)$ is less than $n$. 

More formally, the adversary {\em adds a layer of $k$-dimensional output} $\{(p_{0}, o_{0}), (p_{1}, o_{1}) \cdots (p_{k}, o_{k})\}$ to $S$ in some $n$-simplices of $Ch^{r_{m}}(\mathcal{I})$ (Note that these $n$-simplices are in the complex $S^{r_{m}}$ rather than $S^{t}$.) if given a vertex $v$ (in $S^{t}$) of the process id $p$ adjacent to $S$ in the subdivision of these $n$-simplices, the adversary sets $\delta(v) = o_{i}$ if $p = p_{i}$ for some $i$. Otherwise, the adversary sets $\delta(v)$ as $o_{i}$ where $i$ is chosen according to some rule. After the adversary adds a layer of output $\{(p_{0}, o_{0}), (p_{1}, o_{1}) \cdots (p_{k}, o_{k})\}$ to $S$, the outermost layer of $S$ in the subdivision of these $n$-simplices consists of the output values in $\{o_{0}, o_{1} \cdots o_{k}\}$. 
The adversary will subdivide the complex $S^{t}$ enough times and then add layers of output values $\delta_{U^{'}}(s_{0}), \delta_{U^{'}}(s_{1}) \cdots \delta_{U^{'}}(s_{e}) = \delta_{U}(s_{e}), \delta_{U}(s_{e - 1}) \cdots \delta_{U}(s_{0})$ to the set $S$. The colorless condition is used here since the dimension of $s_{i}$ is less than $n$. Finally, the output values of the outermost layer of $S$ consist of the values in $\delta_{U}(s^{n})$. The adversary simply defines the $\delta$ value of each undefined vertex $v$ in the subdivision of $s^{n}$ using $\delta_{U}(v^{'})$, where $v^{'}$ is the vertex of $s^{n}$ having the same process id as $v$. So the adversary has constructed a partial protocol with respect to $U$. 

We give an example of the path $s_{0}, s_{1} \cdots s_{e}$ in Figure \ref{img:motivations_and_summary:example}, where $r_{m} = 3$. Let $U = \{(p_{0}, v_{0})\}$, and $U^{'} = \{(p_{0}, v_{0}), (p_{1}, v_{1}), (p_{2}, v_{2})\}$. Suppose that a set of vertices in $S^{t}$ is contained in the subdivision of a 2-simplex $s^{2} \in \mathbb{F}_{3}(U)$, then the adversary chooses $S$ as the 1-simplex in $\mathbb{F}_{1}(U) \cap \mathbb{F}_{1}(U^{'})$. There exists a path $s_{0}, s_{1} \cdots s_{e}$ in the subdivision of $S$, where $s_{0}$ is a subsimplex of $s^{2}$ and $s_{3}$ is equal to $CEN(S)$. 

Note that in the above discussion we discuss a simple case. Generally, a set $S$ terminated with the label $U^{'}$ may be in the subdivision of multiple $n$-simplices in $\mathbb{F}_{r_{m}}(U)$. A general discussion is given in Section \ref{sec:nccondition:finalization_after_phase_1}.

\begin{figure}[h]
  \centering
  \includegraphics[width=\linewidth]{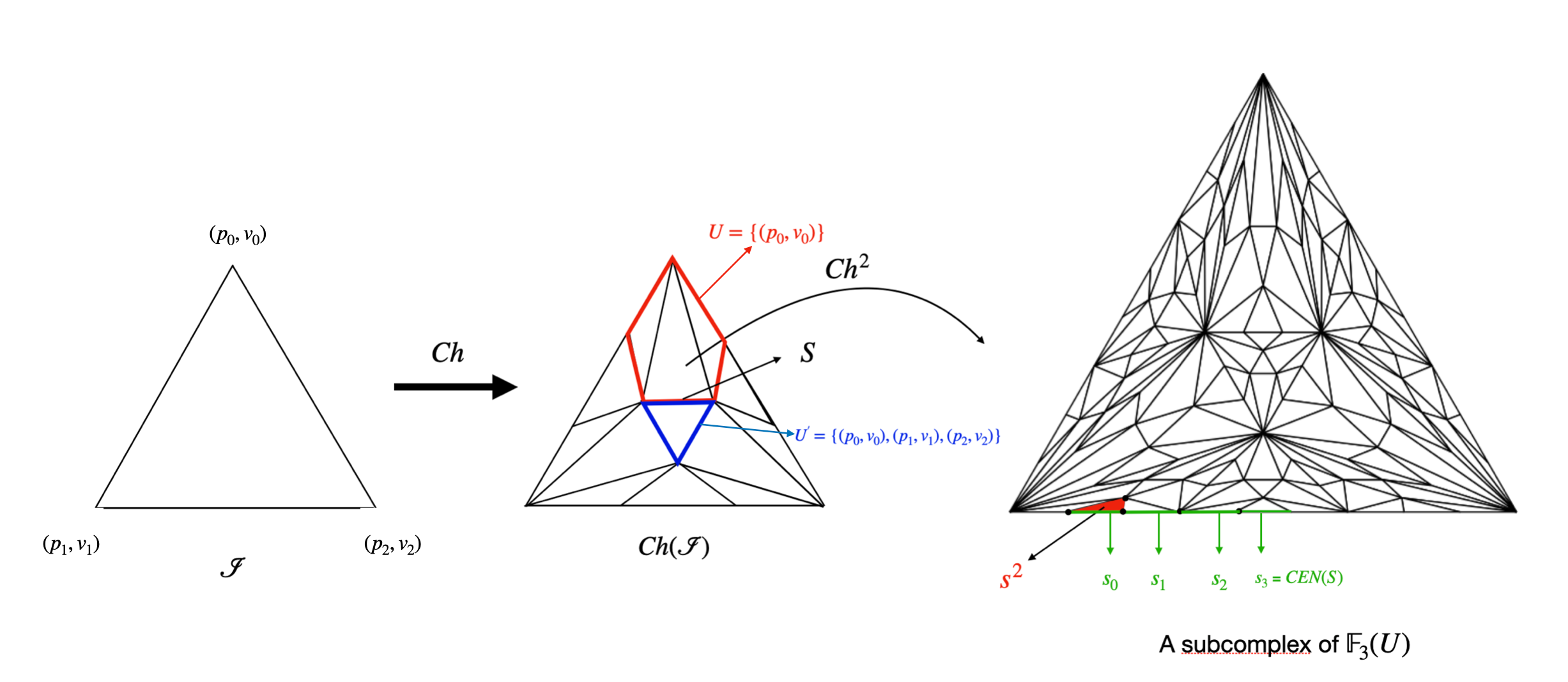}
  \caption{An example of $s_{0}, s_{1} \cdots s_{e}$ when $r_{m} = 3$}
  \label{img:motivations_and_summary:example}
\end{figure}

\begin{theorem}[Theorem \ref{the:adversary_finalize_after_the_first_round_theorem}]
For a colorless task $(\mathcal{I}, \mathcal{O}, \Delta)$, there exists an adversary that can finalize after the first round to win against any restricted extension-based prover if and only if there exists a compatible set of partial protocols, each of which corresponds to a simplex $U \in \mathcal{I}$.
\end{theorem}

\subsection{Different versions of extension-based proofs are equivalent in power}
Recall that there are multiple definitions of extension-based proofs according to the types of queries allowed during the interaction. We have consider the weakest version of extension-based proofs. Now we show that the assignment queries,  which are allowed in the strongest definition of extension-based proofs, will not give a prover more power. 

Since the adversary adds some layers to some set of vertices at the end of phase 1, the adversary has to consider these output values in its response to an assignment query. In a subdivision of an $n$-simplex $s^{n}$ in $Ch^{r_{m}}(\mathcal{I})$, the output value of a set of vertices is $\delta_{U}$ or $\delta_{U^{'}}(N(s^{n}, U^{'}))$ during phase 1. The adversary uses some other output values(i.e. the output values of $\delta_{U^{'}}(s_{1}), \delta_{U^{'}}(s_{2}) \cdots$) after the end of phase 1 for each other label $U^{'}$ to fill the gap between a configuration with output $\delta_{U^{'}}(N(s^{n}, U^{'}))$ and a configuration with output $\delta_{U}(s^{n})$.

When an assignment query is submitted to the adversary asking whether there is a configuration whose output values are those introduced after the end of phase 1. A key observation is that any such output value is $\delta_{U^{'}}(s_{i})$ or $\delta_{U^{'}}(s_{i})$, where $s_{i}$  are connected to a subsimplex $s_{0}$ of $s^{n}$.  The adversary starts with a configuration whose output values are $\delta_{U^{'}}(N(s^{n}, U^{'}))$. Then it adds layers of output values, just as it does after the end of phase 1 such that a configuration is assigned with the queried output value. All terminated vertices involved in these newly added layers are assigned the label $U^{'}$. The adversary returns this configuration that is terminated with required output value.

However, the correctness of our adversarial strategy is based on the condition that any vertex assigned with the label $U^{'}$ uses the output value of $\delta_{U^{'}}(N(s^{n}, U^{'}))$. The adversary will "hide" the newly added layers by adding these layers in reverse order, so that the outermost layer consists of output values of $\delta_{U^{'}}(N(s^{n}, U^{'}))$. Therefore, assignment queries can be integrated into our adversarial strategy, which means that different types of extension-based prover are equivalent in power. The details are given in Section \ref{sec:nccondition:augmented_extension_based_proofs}.

\begin{theorem}[Theorem \ref{the:assignment_queries_no_more_power}]
Assignment queries do not give an extension-based prover more power in its interaction with a protocol.
\end{theorem}

\subsection{Finalization after phase r}
\label{sec:motivation_and_summary:finalization_after_phase_r}
In the discussions of finalization after phase 1, the adversary uses a compatible set of partial protocols to pretend that it has a protocol for some colorless task. In fact, this procedure can continue.
Let $U_{1}$ and $U_{2} \subset s^{n} - U_{1}$, where $s^{n}$ is an $n$-simplex in $\mathcal{I}$ that contains $U_{1}$. 
The adversary can use a set of partial protocols with respect to $[U_{1}, U_{2,0}], [U_{1}, U_{2,1}] \cdots$ to pretend to have a partial protocol with respect to $U_{1}$. 

Will this type of divisions have an end? We can show that it is not necessary to divide a partial protocol when its protocol complex $\mathbb{F}_{1}([U_{1}, U_{2} ... U_{k}])$ , where $ids(U_{1}), ids(U_{2}) \cdots ids(U_{k})$ is a full 1-round schedule, consists of only one simplex representing a configuration reachable via a 1-round schedule $ids(U_{1}), ids(U_{2}) ... ids(U_{k})$ from the initial configuration $U_{1} * U_{2} \cdots U_{k}$. This is proven by showing that there is a set of "smaller" compatible partial protocols $[U_{1}, U_{2} ... U_{k}, U_{k + 1}]$ for each subsimplex $U_{k + 1}$ of $U_{1} * U_{2} \cdots * U_{k}$ if and only if there is a partial protocol with respect to $[U_{1}, U_{2} ... U_{k}]$.

Therefore, the phase $r$ is at most the number of processes $n + 1$. If we want to find a necessary and sufficient condition for finalization after phase $r$, where $1 < r \leq n + 1$, all techniques in finalization after phase 1 can be used (although we need to make some slight changes). 
So, the finalization after phase $r$ has many similarities with the finalization after phase 1.

\begin{theorem}[Theorem \ref{the:adversary_finalize_theorem}]
For a colorless task$(\mathcal{I}, \mathcal{O}, \Delta)$, there exists an adversary that can win against any extension-based prover if and only if there exists a partial protocol with respect to $[U_{1}, \allowbreak ... U_{k}]$ for each possible sequence $[U_{1}, \allowbreak ... U_{k}]$, where $ids(U_{1}), ids(U_{2}) \cdots ids(U_{k})$ is a full 1-round schedule, and all partial protocols are compatible.
\end{theorem}

\ifdefined\FULL

\section{Preparations}
\label{sec:preparations}

\subsection{Partial protocols with respect to simplices in $\mathcal{I}$}
\label{sec:preparations:general_tasks_with_partial_input_information}
In this section, we give a necessary condition for the finalization after phase 1.  At that time, the prover has chosen a configuration $C^{'}$ $\in$ $\mathcal{A}^{'}(1)$ and let $\alpha(2)$ denote the schedule from some initial configuration $C \in \mathcal{A}(1)$ to $C^{'}$. All initial configurations in phase 2, denoted by $\mathcal{A}(2)$, are the configurations reached by performing $\alpha(2)$ from the initial configurations that differ only from $C$ by input values of processes not in the schedule $\alpha(2)$. 
Consider the schedule $\alpha^{''}$ consisting only of the first set of processes in $\alpha(2)$, if the adversary can finalize in this case, then there must be a finalization for the original schedule $\alpha(2)$ because all the configurations reached from a configuration $C$ by the schedule $\alpha(2)$ are contained in the configurations reached from $C$ by $\alpha(2)$. Therefore, we can assume, without loss of generality, that the schedule $\alpha(2)$ contains only one set of processes. 

If the adversary can finalize after phase 1, then for each simplex $U \in \mathcal{I}$, it can set $\delta$ values for the vertices in $\mathbb{F}_{t}(U)$ such that $\delta$ values will not violate the task specification. In the proof of non-existence of extension-based impossibility proofs for $(n, k)$-set agreement in \cite{Alistarh19}, this is simply done by setting $\delta(v) = a$ ($a$ is some value in $vals(U)$) for each undefined vertex. At most 2 different values will be output in any output configuration, which is ensured by some additional properties maintained by their adversarial protocol at the end of phase 1. 

Consider the $n$-simplices in $\mathbb{F}_{1}(U)$. Note that $n$-simplices can be reached from multiple initial configurations in $\mathcal{I}$. A necessary condition for the finalization after phase 1 is that there is a partial protocol with respect to each $U \in \mathcal{I}$. This is not a sufficient condition, since the output values of some configurations in $\mathbb{F}_{t}(U)$ have been determined in phase 1.

Before any process not in $ids(U)$ can perform an operation, processes in $ids(U)$ will writeread the first IS object concurrently using their input values. Therefore, each process will read the input values of all processes in $ids(U)$ from the return value of the operation on the first IS object. We define some notations. Given a simplex $U$ in the input complex $\mathcal{I}$, for each facet of the subcomplex of $\mathcal{O}$ whose process ids are in $ids(U)$, we use it as a label to define a complex whose facets are the $n$-simplices in the output complex $\mathcal{O}$ that contain this facet. The output complex restricted to the simplex $U$, denoted by $\mathcal{O}_{U} = \{\mathcal{O}_{U}^{0}, \mathcal{O}_{U}^{1}, \cdots \mathcal{O}_{U}^{m}\}$, is defined as the disjoint union of the complexes whose labels are in $\Delta(U)$, if we associate each label with a unique index. 
The notation $\mathcal{O}_{U}\backslash U$ is defined as a simplicial complex generated by removing vertices whose ids are in $ids(U)$ from $\mathcal{O}_{U}$.

\begin{lemma}
\label{the:partial_info_lemma}
    For a task $(\mathcal{I}, \mathcal{O}, \Delta)$, there exists a partial protocol with respect to any simplex in $\mathcal{I}$ if and only if for each simplex $U \in \mathcal{I}$, there exists $\mathcal{O}_{U}^{i}$ for some $0 \leq i \leq m$, such that task $(lk(U, \mathcal{I})$, $\mathcal{O}_{U}^{i} \backslash U$, $\Delta_{U})$ has a protocol, where $\Delta_{U}$ carries a simplex $U^{'}$ $\subseteq$ $lk(U, \mathcal{I})$ to the complex $\Delta(U^{'} * U) \backslash U$, which is generated from $\Delta(U^{'} * U) $ by removing vertices whose ids are in $ids(U)$.
\end{lemma}

\begin{proof}
We first prove necessity. Suppose that there is a partial protocol with respect to $U$. Consider the execution in which only processes in $ids(U)$ run and decide. The output values of $ids(U)$ must be in $\Delta(U)$, and therefore be associated with some index $i$. We denote this configuration by $C_{1}$. At this time, active processes are those not in $ids(U)$. Let all the remaining processes run and decide. The output values of the processes not in $ids(U)$ will be simplices in $\mathcal{O}_{U}^{i} \backslash U$.

Let $C_{0}$ be an initial configuration of the protocol $(lk(U, \mathcal{I})$, $\mathcal{O}_{U}^{i} \backslash U$, $\Delta_{U})$. We can construct a protocol for the task $(lk(U, \mathcal{I})$, $\mathcal{O}_{U}^{i} \backslash U$, $\Delta_{U})$ as follows: For any schedule $\alpha$ from $C_{0}$ consisting of processes in $\Pi - ids(U)$, we can return the output values of processes in $\Pi - ids(U)$ given by the partial protocol with respect to $U$, in the configuration reached from the configuration $C_{1}$ by $\alpha$.

Next, we prove sufficiency. If there exists a protocol $P$ for the task $(lk(U, \mathcal{I}), \mathcal{O}_{U}^{i}\backslash U, \Delta_{U})$ for some index i, we can construct a partial protocol by letting the processes in $U$ decide the label values associated with the index $i$ after accessing the first IS object, while the processes not in $U$ execute the protocol $P$ and simply ignore the information from the processes in $U$. The output values of an n-dimensional simplex are included in $\mathcal{O}_{U}^{i}$. Therefore, the constructed protocol is a valid partial protocol with respect to $U$.
\end{proof}

For each simplex $U \in \mathcal{I}$, we can determine the output values of the processes in $U$ before the protocol starts, and this is why we define $\mathcal{O}_{U}$. Other processes can execute a protocol and ignore any information from processes in $U$. These output values of the processes in $U$ have to be the same for all initial configurations. By this lemma and the asynchronous computability theorem, we can obtain the following theorem.

\begin{lemma}
\label{the:partial_info_theorem}
    For a task $(\mathcal{I}, \mathcal{O}, \Delta)$,  there exists a partial protocol with respect to any simplex in $\mathcal{I}$ if and only if for each simplex $U \in \mathcal{I}$, there exists $\mathcal{O}_{U}^{i}$ for some $0 \leq i \leq m$, a chromatic subdivision $\sigma$ of $lk(U, \mathcal{I})$ and a color-preserving simplicial map $\mu:\sigma(lk(U, \mathcal{I})) \rightarrow \mathcal{O}_{U}^{i}$,  such that for each simplex $S$ in $\sigma(lk(U, \mathcal{I}))$, $\mu(S) \in \Delta_{U}(carrier(S, lk(U, \mathcal{I})))$ where $\Delta_{U}$ carries a simplex $U^{'}$ $\subseteq$ $lk(U, \mathcal{I})$ to the complex $\Delta(U^{'} * U) \backslash U$, which is generated from $\Delta(U^{'} * U) $ by removing vertices whose ids are in $ids(U)$.
\end{lemma}

If a task $(\mathcal{I}, \mathcal{O}, \Delta)$ does not have extension-based impossibility proofs because the adversary can finalize after phase 1, then it must satisfy the conditions presented in Theorem \ref{the:partial_info_theorem}. Therefore, there exists a partial protocol with respect to $U$ which is denoted as $\delta_{U}$ in our subsequent discussions. In the example of $(n, 2)$-set agreement, letting $a$ denote some input value in $val(U)$, we can construct a color-preserving simplicial map $\mu$ from $lk(U, \mathcal{I})$ to $\mathcal{O}_{U}^{i}$ whose label is a tuple consisting of only value $a$. The simplicial map $\mu$ carries each vertex in $lk(U, \mathcal{I})$ to a vertex with the same color and value $a$. Since each vertex has seen the input value of $a$, this output value will not cause any violation of the task specification. Note that this condition is only a necessary one(not a sufficient one) for finalization after the end of phase 1. The consensus task that has an extension-based proof in \cite{Alistarh19} also satisfies the conditions of Theorem \ref{the:partial_info_theorem}. We will soon discuss their differences, which is a good example of how to understand our sufficient condition.

We have shown that the schedule $\alpha(2)$ chosen at the end of phase one can be assumed without loss of generality to contain only one set of processes. Furthermore, it can be assumed that this set of processes has only one process if we discuss only the existence of partial protocols. Suppose that $U_{1}$ is a subsimplex of $U_{2}$, and there is a partial protocol $\delta_{U_{1}}$, then $\delta_{U_{1}}$ can be applied as a partial protocol with respect to $U_{2}$. We can define the output values of the configuration reached by the schedule $\{U_{2}\}\gamma$ as the output values of the configuration reached by the schedule $\{U_{1}\}\{U_{2} - U_{1}\}\gamma$ given by the partial protocol $\delta_{U_{1}}$. Another observation is that $\mathbb{F}_{0}(U_{2})$ is a subset of $\mathbb{F}_{0}(U_{1})$.

\begin{theorem}
\label{the:partial_info_theorem_2}
    For a task $(\mathcal{I}, \mathcal{O}, \Delta)$, there exists a partial protocol with respect to any simplex in $\mathcal{I}$ if and only if for each vertex $v \in \mathcal{I}$, there exists $\mathcal{O}_{v}^{i}$ for some $0 \leq i \leq m$, such that task $(lk(v, \mathcal{I})$, $\mathcal{O}_{v}^{i} \backslash v$, $\Delta_{v})$ has a protocol, where $\Delta_{v}$ carries a simplex $\tau$ $\subseteq$ $lk(v, \mathcal{I})$ to complex $\Delta(\tau * v) \backslash v$, which is generated from $\Delta(\tau * v) $ by removing the vertices of $ids(v)$.
\end{theorem}

In other words, a partial protocol with respect to $U_{1}$ can be transformed into a partial protocol with respect to $U_{2}$ where $U_{1}$ is a subsimplex of $U_{2}$. 

\subsection{Canonical neighbor}
\label{sec:preparations:canonical_neighbor}
Suppose that the condition of Theorem \ref{the:partial_info_theorem} is satisfied. This means that for each simplex $U \in \mathcal{I}$, there is a partial protocol with respect to $U$. We say that this partial protocol has the label $U$. For each partial protocol $\delta_{U}$, there exists a round number $r(U)$ such that any execution of $\delta_{U}$ terminates before the round $r(U)$. We can choose the largest round of termination $r_{m}$ and set all the executions of all partial protocols to end in round $r_{m}$ by having all processes defer their termination to round $r_{m}$ and output the same values as in the original partial protocols. 

First, we define the {\em protocol label} of an $n$-simplex. For each $n$-simplex $s^{n}$ in $Ch^{r_{m}}(\mathcal{I})$, we can define a unique simplex $U$ as the protocol label for $s^{n}$, where $s^{n}$ is reached from some $n$-simplex in $\mathbb{F}_{1}(U)$. 
Let $s^{k}$ be a $k$-simplex in $Ch^{r_{m}}(\mathcal{I})$. The {\em possible labels} of $s^{k}$ are defined as protocol labels of those $n$-simplices containing $s^{k}$. Note that there is at least one possible label for each $k$-simplex. If a $k$-simplex is in the interior of some $\mathbb{F}_{r_{m}}(U)$, then it has only one possible label. But for a $k$-simplex at the intersection of $\mathbb{F}_{r_{m}}(U)$ and $\mathbb{F}_{r_{m}}(U^{'})$, it has multiple possible labels.  

Let $s^{k}$ be a $k$-simplex in $S^{t}$, where $t > r_{m}$. Note that $s^{k}$ is not a simplex in $S^{r_{m}}$, but a simplex in the subdivision of $S^{r_{m}}$. {\em Possible labels} of $s^{k}$ are defined as possible labels of $carrier(s^{k}, Ch^{r_{m}}(\mathcal{I}))$. A possible label of $s^{k2}$ is not necessarily a possible label of $s^{k}$, where $s^{k2}$ is a subsimplex of $s^{k}$.
Note that if $s^{k}$ is terminated with one of its possible labels, then there will be no violation of the task specification, since each output configuration is obtained from some partial protocol. Our definition of possible labels has some properties. For example, if each vertex of a simplex has some possible label $U$, then this simplex will also have a possible label $U$. 

We must assume that $r_{m} \geq 2$ for the next definition. This assumption is without loss of generality since, if $r_{m} = 1$, the termination of each partial protocol can be postponed by one round, as stated above. If an $n$-simplex $s^{n}$ in $\mathbb{F}_{r_{m}}(U)$ has an intersection with $\mathbb{F}_{r_{m}}(U^{'})$ for some label $U^{'}$, the vertices of $s^{n}$ can be sorted into two types: intersection vertices $V(s^{n}, U^{'})$ and non-intersection vertices. In other words, the intersection vertices have possible labels $U$ and $U^{'}$. 

Let $C$ be an initial configuration represented by an $n$-simplex $\sigma \in \mathcal{I}$. We restrict our view to the chromatic subdivision of $\sigma$. For an $n$-simplex $s^{n}_{1}$ in $\mathbb{F}_{r_{m}}(U)$, if $s^{n}_{1}$ has an intersection with $\mathbb{F}_{r_{m}}(U^{'})$, we define the {\em canonical neighbor} $N(s^{n}_{1}, U^{'})$ of $s^{n}_{1}$ with label $U^{'}$ as an $n$-simplex $s^{n}_{2}$ such that
\begin{itemize}
\item[1)]$s^{n}_{2}$ is an $n$-simplex in $\mathbb{F}_{r_{m}}(U^{'})$.
\item[2)]For each $v \in V(s^{n}_{1}, U^{'})$, the vertex $v$ is in $N(s^{n}_{1}, U^{'})$.
\item[3)] For each non-intersection vertex $v$ and $ids(v) = p$, there exists a vertex $v^{'}$ of $s^{n}_{2}$ with the process id $p$ such that $carrier(v, \sigma) = carrier(v^{'}, \sigma)$.
\item[4)]If an $n$-simplex $s^{n}_{3}$ (not necessarily in $\mathbb{F}_{r_{m}}(U)$) shares a vertex of the process id $p$ with $s^{n}_{1}$, $N(s^{n}_{1}, U^{'})$ shares a vertex with the process id $p$ with $N(s^{n}_{3}, U^{'})$.
\end{itemize}

\begin{lemma}
\label{the:carrier_of_union}
Let $s^{k}$ be a $k$-simplex in a subdivision of $\sigma$. Then $carrier(s^{k}, \sigma) = \bigcup\limits_{v \in s^{k}} carrier(v, \sigma)$. 
\end{lemma}
\begin{proof}

We use the geometric definition of simplices.
Let $\tau$ be the subsimplex of $\sigma$ whose set $V$ of vertices is $\bigcup\limits_{v \in s^{k}} carrier(v, \sigma)$. Each point $x$ in $s^{k}$ can be expressed as a combination of vertices of $s^{k}$, that is, $x = \sum_{i = 0}^{k}e_{i}v_{i}$ where $\sum_{i = 0}^{k}e_{i} = 1$ and $e_{i} \geq 0$ for all i. At the same time, each vertex $v_{i}$ of $s^{k}$ can be expressed as a combination of vertices of $carrier(v_{i}, \sigma)$, i.e. $v_{i} = \sum_{i = 0}^{k^{'}}e^{'}_{i}v^{'}_{i}$ where $\sum_{i = 0}^{k^{'}}e^{'}_{i} = 1$, $v^{'}_{i}$ are in $V$ and $e^{'}_{i} > 0$ for all i. Therefore, each point $x$ in $s^{k}$ can be expressed as a combination of vertices of $V$, i.e., $carrier(s^{k}, \sigma)$ is a subset of $\tau$.

We can show that $carrier(s^{k}, \sigma)$ is not a proper subset of $\tau$. Consider a vertex $x$ in $s^{k}$ such that each $e_{i} > 0$. When $x$ is expressed as a combination of vertices in $V$, the coefficient of each vertex $v_{i}$ is greater than 0. So $carrier(s^{k}, \sigma)$ is equal to $\tau$.

\end{proof}

By the definition of canonical neighbors and Lemma \ref{the:carrier_of_union}, each subsimplex of an $n$-simplex has the same carrier as the corresponding subsimplex of its canonical neighbor with any label, and we have the following lemma.

\begin{lemma}
\label{the:canonical_neighbor_no_safety_issue}
For an $n$-simplex $s^{n}$ in $\mathbb{F}_{r_{m}}(U)$, the output $\delta_{U^{'}}(N(s^{n}, U^{'}))$ is a valid output for configurations in the subdivision of $s^{n}$ that does not violate the carrier map.
\end{lemma}

We now give an implementation of the canonical neighbor. Let $C$ be an input configuration represented by a $n$-simplex $\sigma$. Each simplex in $Ch^{r}(\sigma)$ represents a configuration reached from $C$ via a full $r$-round schedule. Each round is a partition of the set of all processes, denoted by $\Pi$. A vertex in $Ch^{r}(\sigma)$ is uniquely determined by the set of processes that it reads in each partition of $\Pi$. So, if some process $p$ reads the same set of processes in each partition of $\Pi$ in two configurations in $Ch(\sigma)$, then the vertex with the process id $p$ is shared by two configurations. We will use this result in later proofs.

A notation from \cite[Chapter 10.2]{Herlihy13} is adopted in the next proof:
\begin{itemize}
    \item $(\mathbb{F}_{1} \downarrow P)(U)$ is the subcomplex generated from $\mathbb{F}_1(U)$ by removing every vertex whose process id is in $ids(P)$, for a simplex $P$.
    \item For $i \geq 1$, $(\mathbb{F}_{i+1} \downarrow P)(U) = \chi((\mathbb{F}_{i} \downarrow P)(U), \delta)$ consists of all simplices representing configurations reachable via $(\Pi - P)$-only $i$-round schedules from configurations in $(\mathbb{F}_{1} \downarrow P)(U)$.
\end{itemize}

\begin{lemma}
\label{the:canonical_neighbor_existence_lemma}
    For an $n$-simplex $s^{n}$ in $Ch^{r_{m}}(\sigma)$ that has an intersection with $\mathbb{F}_{r_{m}}(U)$, there exists an canonical neighbor of $s^{n}$ with the label $U$ if $r_{m} \geq 2$ that satisfies the first, second and third conditions.
\end{lemma}

\begin{proof}
For two simplices $U_{1}$ and $U_{2}$ that are subsimplices of $\sigma$, we can assume that $|ids(U_{1})| \leq |ids(U_{2})|$. Lemma 10.4.4 in \cite{Herlihy13} states that, for any $i$, the intersection $\mathbb{F}_{i}(U_{1}) \bigcap \allowbreak \mathbb{F}_{i}(U_{2})$ equals $(\mathbb{F}_{i} \downarrow W)(U_{1} * U_{2})$ where $W = U_{1}$ if $U_{1}$ is a subsimplex of $ U_{2}$, otherwise $W$ is the joining $U_{1} * U_{2}$ of $U_{1}$ and $U_{2}$. Recall that the joining of two simplices is defined as a simplex such that $ids(U_{1} * U_{2}) = ids(U_{1}) \cup ids(U_{1})$. We will discuss the two cases respectively.

We first discuss the case where $U_{1}$ is not a subsimplex of $U_{2}$. Consider the standard chromatic subdivision $Ch^{1}(\sigma)$. For any two simplices $U_{1}$ and $U_{2}$ that satisfy $|ids(U_{1})| \leq |ids(U_{2})|$, the intersection $\mathbb{F}_{1}(U_{1}) \bigcap \allowbreak \mathbb{F}_{1}(U_{2})$ is $(\mathbb{F}_{1} \downarrow U_{1} * U_{2})(U_{1} * U_{2})$. The intersection contains all the configurations reached by some $(\Pi - ids(U_{1} * U_{2}))$-only schedule when the input values of $ids(U_{1} * U_{2})$ have been read by all the processes in $(\Pi - ids(U_{1} * U_{2}))$. In other words, it is a simplicial complex consisting of $(|\Pi - ids(U_{1} * U_{2})|)$-simplices and their subsimplices, and each $(|\Pi - ids(U_{1} * U_{2})|)$-simplex corresponds to an ordered partition of processes in $\Pi - ids(U_{1} * U_{2})$.  

In $\mathbb{F}_{1}(U_{1})$, each n-simplex corresponds to an ordered partition of the processes in $(\Pi - ids(U_{1}))$. If an $n$-simplex representing some configuration $C_{1}$ in $\mathbb{F}_{1}(U_{1})$ contains a $|\Pi - ids(U_{1} * U_{2})|$-simplex at the intersection, then its schedule(from $\sigma$ to $C_{1}$) can be expressed as a concatenation of $ids(U_{1})$, an ordered partition of $ids(U_{2}) - ids(U_{1})$ and an ordered partition of $(\Pi - ids(U_{1} * U_{2}))$, or a special case where the last process set of the former partition is merged with the first process set of the latter partition. We can choose the n-simplex representing the configuration $C_{2}$ reached from $C$ by the schedule generated by concatenating $ids(U_{2})$, $ids(U_{1}) - ids(U_{2})$ and the same ordered partition of $(\Pi - ids(U_{1} * U_{2}))$. 
For any $n$-simplex in $\mathbb{F}_{r_{m}}(U_{1})$ representing some configuration $C_{3}$ reached from $C_{1}$ by a schedule $\gamma$ and having an intersection with $\mathbb{F}_{r_{m}}(U_{2})$, we choose the $n$-simplex representing some configuration $C_{4}$ reached from $C_{2}$ by the same schedule $\gamma$ as its canonical neighbor with the label $U_{2}$. Therefore, the $n$-simplex representing $C_{4}$ is in $\mathbb{F}_{r_{m}}(U_{2})$, and the first requirement of the canonical neighbor is satisfied. 

The carrier of each vertex of a configuration reached from $C$ depends only on the set of processes it has seen(directly or indirectly). The schedules from $C$ to $C_{3}$ and to $C_{4}$ only differ in the first partition of $\Pi$. In the first partition of $\Pi$, only a process in $ids(U_{1} * U_{2})$ can see a different set of input values in $C_{3}$ and $C_{4}$. If a process $p \in ids(U_{1} * U_{2})$ sees the input values of a subset of $ids(U_{1} * U_{2})$ in $C_{1}$, then $p$ sees the input values of a different subset of $ids(U_{1} * U_{2})$ in $C_{2}$. 

Consider the schedule $\gamma$ from $C_{1}$ to $C_{3}$. Since the $n$-simplex representing $C_{3}$ shares some vertices, whose process ids $P$ are a subset of $(\Pi - ids(U_{1} * U_{2}))$, with $\mathbb{F}_{r_{m}}(U_{2})$, $\gamma$ can be expressed as $\gamma_{1}\gamma_{2}$ where $\gamma_{1}$ is a $P$-only schedule and $\gamma_{2}$ does not contain any process in $P$. Therefore, the $n$-simplex representing $C_{4}$ contains all vertices of $C_{3}$ whose ids are in $P$. The second requirement of the canonical neighbor is satisfied. The shared vertices of $C_{3}$ and $C_{4}$ have the same carrier. However, it is possible that the vertices with process ids in $ids(U_{1} * U_{2})$ have different carriers in $C_{3}$ and $C_{4}$. This is exactly what happens to the configurations $C_{3} = C_{1}$ and $C_{4} = C_{2}$ when $r_{m} = 1$. But the mismatch will disappear after one more subdivision. In $Ch^{2}(\sigma)$, if an n-simplex $C_{3}$ reached from $C_{1}$ has an intersection with $\mathbb{F}_{2}(U_{2})$, then the second partition of $\Pi$ in the schedule from $C$ to $C_{3}$ must begin with a set of processes that contains some process $p$ in $P$. Since $P$ is a subset of $(\Pi - ids(U_{1} * U_{2}))$, this process $p$ has already seen input values of $ids(U_{1} * U_{2})$ in the first partition of $\Pi$. So each process in $C_{3}$ has seen the input values of $ids(U_{1} * U_{2})$. Therefore, the carrier of each vertex in $C_{3}$ will be the same as the carrier of the corresponding vertex in $C_{4}$ as it sees the same set of input values. The third requirement of the canonical neighbor is satisfied. Now we have shown that $C_{4}$ is the canonical neighbor of $C_{3}$ with the label $U_{2}$ that satisfies the first, second and third conditions. 

But in $Ch^{1}(\sigma)$, the n-simplices in $\mathbb{F}_{1}(U_{1})$ having an intersection with $\mathbb{F}_{1}(U_{2})$ are not limited to the $n$-simplices we have discussed. There could be an $n$-simplex representing $C_{1}$ that contains only a subsimplex of those $(|\Pi - U_{1} * U_{2}|)$-simplices at the intersection. The schedule from $C$ to $C_{1}$ can be expressed as a concatenation of $ids(U_{1})$, an ordered partition of $(ids(U_{2}) - ids(U_{1}))\bigcup D$ and an ordered partition of $(\Pi - ids(U_{1} * U_{2}) - D)$ for some subset $D$ of $(\Pi - ids(U_{1} * U_{2}))$. Consider the last set of processes $U_{l}$ of the ordered partition of $(ids(U_{2}) - ids(U_{1}))\bigcup D$. It must contain some process in $(ids(U_{2}) - ids(U_{1}))$, and let $D_{l}$ be the processes in $U_{l}$ that are also in $D$. The vertices with ids in $(\Pi - ids(U_{1} * U_{2}) - D) \bigcup D_{l}$ are the intersection vertices. The limitation of an ordered partition of a set $S$ to a subset $S^{'} \subset S$ is defined as an ordered partition of $S^{'}$ generated by removing all processes not in $S^{'}$ from the original ordered partition. We can choose the $n$-simplex $C_{2}$ reached from $C$ by the schedule generated by concatenating $ids(U_{2})$, $ids(U_{1}) - ids(U_{2})$, the limitation of the same ordered partition of $(ids(U_{2}) - ids(U_{1}))\bigcup D$ to $D$ and the same ordered partition of $(\Pi - ids(U_{1}) \cup ids(U_{2}) - D)$. For any $n$-simplex $C_{3}$ reached from $C_{1}$ by a schedule $\gamma$ that has an intersection with $\mathbb{F}_{r_{m}}(U_{2})$, we choose the $n$-simplex $C_{4}$ reached from $C_{2}$ by the same schedule $\gamma$ as its canonical neighbor. The first requirement of the canonical neighbor is satisfied. 

Only processes in $ids(U_{1}) \bigcup ids(U_{2}) \bigcup D - D_{l}$ can see a different set of input values after changing the first partition of $\Pi$. If a process $p \in ids(U_{1}) \bigcup ids(U_{2}) \bigcup D - D_{l}$ sees the input values of a subset of $ids(U_{1}) \bigcup ids(U_{2}) \bigcup D$ in $C_{1}$, then $p$ sees the input values of a different subset of $ids(U_{1}) \bigcup ids(U_{2}) \bigcup D$ in $C_{2}$. 

Let $P$ be the process ids of the vertices shared by $C_{3}$ and $\mathbb{F}_{r_{m}}(U_{2})$. Then $P$ is a subset of $(\Pi - ids(U_{1}) \bigcup ids(U_{2}) - D + D_{l})$. The schedule $\gamma$ can be expressed as $\gamma_{1}\gamma_{2}$ where $\gamma_{1}$ is a $P$-only schedule and $\gamma_{2}$ does not contain any process in $P$. The second requirement of the canonical neighbor is satisfied. 
In $Ch^{2}(\sigma)$, if an n-simplex $C_{3}$ reached from some $C_{1}$ has an intersection with $\mathbb{F}_{2}(U_{2})$, then the second partition of $\Pi$ in the schedule from $C$ to $C_{3}$ must begin with a set of processes that contains some process in $P$ that has already seen the input values of $ids(U_{1}) \bigcup ids(U_{2}) \bigcup D$. Therefore, the mismatch of input values that a process in $(\Pi - ids(U_{1}) \bigcup ids(U_{2}) - D + D_{l})$ sees in $C_{1}$ and $C_{2}$ disappears.  The third requirement of the canonical neighbor is satisfied. We have discussed all n-simplices in $\mathbb{F}_{r_{m}}(U_{1})$ having an intersection with $\mathbb{F}_{r_{m}}(U_{2})$ and constructed a canonical neighbor for each of them. In fact, the simplices discussed in the previous paragraphs belong to a special case where $D$ is an empty set. By symmetry, we can also find a canonical neighbor for all n-simplices in $\mathbb{F}_{r_{m}}(U_{2})$ that has an intersection with $\mathbb{F}_{r_{m}}(U_{1})$.

Then we will prove the case that $U_{1}$ is a subsimplex of $U_{2}$. The proof techniques are almost identical to those in the former case. The intersection $\mathbb{F}_{1}(U_{1}) \bigcap \allowbreak \mathbb{F}_{1}(U_{2})$ is $(\mathbb{F}_{1} \downarrow U_{1})(U_{1} * U_{2})$. In $Ch^{1}(\sigma)$, the schedule of an n-simplex representing $C_{1}$ in $\mathbb{F}_{1}(U_{1})$ having some vertices at the intersection can be expressed as a concatenation of $ids(U_{1})$, an ordered partition of $(ids(U_{2}) - ids(U_{1}))\bigcup D$ and an ordered partition of $(\Pi - ids(U_{2}) - D)$ for some subset $D$ of $(\Pi -  ids(U_{2}))$. We can choose the $n$-simplex representing $C_{2}$ reached by the schedule generated by concatenating $ids(U_{2})$, the limitation of the same ordered partition of $(ids(U_{2}) - ids(U_{1}))\bigcup D$ to $D$ and the same ordered partition of $(\Pi - ids(U_{2}) - D)$. For any $n$-simplex representing a configuration $C_{3}$ reached from $C_{1}$ by a schedule $\gamma$, the $n$-simplex representing $C_{4}$ reached from $C_{2}$ by the same schedule $\gamma$ is chosen as its canonical neighbor with the label $U_{2}$. Three requirements of a canonical neighbor can be checked, as in the former case. 

Finally, we prove the case that $U_{2}$ is a subsimplex of $U_{1}$. In fact, we can use the same construction. If the schedule from $C$ to $C_{1}$ is a concatenation of $ids(U_{1})$ and an ordered partition of $\Pi - ids(U_{1})$, then the schedule from $C$ to $C_{2}$ is $ids(U_{2})$, $ids(U_{1}) - ids(U_{2})$ and the same ordered partition of $\Pi - ids(U_{1})$.
\end{proof}

According to Lemma \ref{the:canonical_neighbor_existence_lemma}, for an n-simplex $s^{n}$ in $Ch^{r_{m}}(\mathcal{I})$ that has an intersection with $\mathbb{F}_{r_{m}}(U)$ , there exists a canonical neighbor of $s^{n}$ with the label $U$ satisfying the first, second and third conditions if $r_{m} \geq 2$.
If we have to terminate some vertices with the label $U$ in a subdivision of $s^{n}$, we can always find its canonical neighbor and use $\delta_{U}(N(s^{n}, U))$ as the terminated values without violating the task specification. But an important question remains to be answered: are these terminated values well-defined? In fact, this is equal to the fourth condition in the definition of canonical neighbor. We have to prove that our construction satisfies the fourth condition.

\begin{lemma}
\label{the:canonical_neighbor_common_boundary_lemma}
    If two $n$-simplices in $Ch^{r_{m}}(\sigma)$ both have canonical neighbors with the label $U_{2}$ and share some vertex $v$ with the process id $p$, then the two vertices of the id $p$ in their canonical neighbors with the label $U_{2}$ will have the same output value when using partial protocol $\delta_{U_{2}}$.
\end{lemma}

\begin{proof}

Suppose that the two $n$-simplices representing $C_{3}$ and $C_{3}^{'}$ are reached from the two $n$-simplices representing $C_{1}$ and $C_{1}^{'}$ in $Ch^{1}(\sigma)$ by the schedules $\gamma$ and $\gamma^{'}$, respectively. Note that the $n$-simplices representing $C_{1}$ and $C_{1}^{'}$ may not be in the same subcomplex $\mathbb{F}_{1}(U_{1})$. We will use the notations of the previous proof to simplify the understanding. Because the vertex $v$ is shared by the two $n$-simplices representing $C_{3}$ and $C_{3}^{'}$, the process $p$ will see the same set of processes in each partition of $\Pi$.

If $C_{1}$ and $C_{1}^{'}$ are represented by the same $n$-simplex in $\mathbb{F}_{1}(U_{1})$, then using the strategy in the proof of Lemma \ref{the:canonical_neighbor_existence_lemma}, the canonical neighbors of the two $n$-simplices representing $C_{3}$ and $C_{3}^{'}$ are reached from the same $n$-simplex in $\mathbb{F}_{1}(U_{2})$ by the schedules $\gamma$ and $\gamma^{'}$, respectively. Since the vertex $v$ is shared by the two $n$-simplices representing $C_{3}$ and $C_{3}^{'}$, their canonical neighbors will also share the vertex having the process id $p$ in our construction. The output values are given by the same protocol and are therefore identical.

Suppose that two simplices representing $C_{1}$ and $C_{1}^{'}$ are not the same but still in the same subcomplex $\mathbb{F}_{1}(U_{1})$. 
Consider the case where $U_{1}$ is not a subsimplex of $U_{2}$ and vice versa. The schedule from $C$ to $C_{1}$ can be expressed as a concatenation of $ids(U_{1})$, an ordered partition of $(ids(U_{2}) - ids(U_{1}))\bigcup D$ and an ordered partition of $(\Pi - ids(U_{1}) \bigcup ids(U_{2}) - D)$, for some subset $D$ of $(\Pi - ids(U_{1}) \bigcup ids(U_{2}))$. The schedule from $C$ to $C_{1}^{'}$ can also be presented for some subset $D^{'}$. 
The $n$-simplex representing $C_{2}$ is defined as the $n$-simplex reached from $C$ by the schedule generated by concatenating $ids(U_{2})$, $ids(U_{1}) - ids(U_{2})$, the limitation of the same ordered partition of $(ids(U_{2}) - ids(U_{1}))\bigcup D$ to $D$ and the same ordered partition of $(\Pi - ids(U_{1}) \bigcup ids(U_{2}) - D)$. And the $n$-simplex representing $C_{2}^{'}$ is defined in the same way. We will show that the canonical neighbors reached from $C_{2}$ and $C_{2}^{'}$ share the vertex with the process id $p$.

We need to analyze the first partition of $\Pi$. Suppose that the vertex with process id $p$ is shared by $C_{1}$ and $C_{1}^{'}$. There are four possible cases, depending on the set to which the process $p$ belongs.
\begin{itemize}
    \item[1)] If $p \in ids(U_{2})$, then the process $p$ will see the input values of the processes in $U_{2}$ in $C_{2}$ and $C_{2}^{'}$. 
    \item[2)] If $p \in ids(U_{1}) - ids(U_{2})$, then the process $p$ will see the input values of the processes in $ids(U_{1}) + ids(U_{2})$ in $C_{2}$ and $C_{2}^{'}$.
    \item[3)] If $p \in D \bigcup D^{'}$, the vertex $v$ in $C_{1}$ and $C_{1}^{'}$ will see the same subset of $\Pi - ids(U_{1}) \bigcup ids(U_{2})$, and the same subset of $ids(U_{1}) \bigcup ids(U_{2})$ since it is a shared vertex of $C_{1}$ and $C_{1}^{'}$ by assumption. By construction, the process $p$ will see $ids(U_{1}) \bigcup ids(U_{2})$(not a subset of it) and the same subset of $D \bigcup D^{'}$ in $C_{2}$ and $C_{2}^{'}$.  
    \item[4)] If $p \in \Pi - ids(U_{1}) \bigcup ids(U_{2}) - D \bigcup D^{'}$, then the adjustment of the first partition of $\Pi$ will not affect the set of processes that $p$ sees. 
\end{itemize}
In each case, the sets of processes that $p$ will see in $C_{2}$ and $C_{2}^{'}$ are the same. If any vertex is shared by the simplices representing $C_{1}$ and $C_{1}^{'}$, the simplices representing $C_{2}$ and $C_{2}^{'}$ will share the vertex with the same process id. Recall that the schedules $\gamma$ from $C_{1}$ to $C_{3}$ and $\gamma^{'}$ from $C_{1}^{'}$ to $C_{3}^{'}$ are indistinguishable to the process $p$. The canonical neighbors of the simplices representing $C_{3}$ and $C_{3}^{'}$ are reached from $C_{2}$ and $C_{2}^{'}$ by $\gamma$ and $\gamma^{'}$, respectively. Therefore, the canonical neighbors of the simplices representing $C_{3}$ and $C_{3}^{'}$ share the vertex with the process id $p$.

Suppose that $U_{1}$ is a subsimplex of $U_{2}$.
If $C_{1}$ and $C_{1}^{'}$ are in $\mathbb{F}_{1}(U_{1})$. The schedule from $C$ to $C_{1}$ can be expressed as a concatenation of $ids(U_{1})$, an ordered partition of $(ids(U_{2}) - ids(U_{1}))\bigcup D$ and an ordered partition of $(\Pi - ids(U_{2}) - ids(D)$, for some subset $D$ of $(\Pi - ids(U_{2}))$. We can use almost the same proof techniques here as in the previous paragraph. We can discuss different cases of adjustment of the first partition and prove that the simplices representing $C_{2}$ and $C_{2}^{'}$ will share the vertex with the process id $p$. A detailed proof is omitted to avoid redundancy. 
In fact, the proof here can be included in the proof presented in the previous paragraphs as a special case where $U_{1}$ is a subsimplex of $U_{2}$. 

The case that $U_{2}$ is a subsimplex of $U_{1}$ can also be seen as a special case. This observation can simplify the remaining proof.

Suppose that the simplices representing $C_{1}$ and $C_{1}^{'}$ are not in the same subcomplex $\mathbb{F}_{1}(U_{1})$. We will denote their first set of processes(and their input values) as $U_{1}$ and $U_{1}^{'}$. Then the schedule from $C$ to $C_{1}$ is expressed as above, and the schedule from $C$ to $C_{1}^{'}$ can be expressed as a concatenation of $ids(U_{1}^{'})$, an ordered partition of $(ids(U_{2}) - ids(U_{1}^{'}))\bigcup D^{'}$ and an ordered partition of $(\Pi - ids(U_{1}^{'}) \bigcup ids(U_{2}) - D^{'})$, for some subset $D^{'}$ of $(\Pi - ids(U_{1}^{'}) \bigcup ids(U_{2}))$. If the vertex $v$ with the process id $p$ is shared by the simplices representing $C_{1}$ and $C_{1}^{'}$, then the process $p$ is in $\Pi - ids(U_{1}) \bigcap ids(U_{1}^{'})$, and has seen the input values of $ids(U_{1}) \bigcup ids(U_{1}^{'})$ in $C_{1}$ and $C_{1}^{'}$. The process $p$ sees the same set of processes in $C_{1}$ and $C_{1}^{'}$, and we will argue that it sees the same set of processes in $C_{2}$ and $C_{2}^{'}$. 

Consider the set of processes that $p$ will see in $C_{2}$ (compared with the set of processes that $p$ will see in $C_{1}$): 
\begin{itemize}
    \item[1)] If $p \in ids(U_{2})$, then the process $p$ will see the input values of $ids(U_{2})$ in $C_{2}$ and $C_{2}^{'}$.
    \item[2)] If $p \in ids(U_{1}) - ids(U_{2})$, it sees the input values of $ids(U_{1})$ in $C_{1}$ and sees the input values of $ids(U_{1}) \bigcup ids(U_{2})$ in $C_{2}$. The additional processes that $p$ sees are those in $ids(U_{2}) - ids(U_{1})$.
    \item[3)] If $p \in ids(D_{1})$, the new processes that process $p$ will see are those in $ids(U_{2})$ that it does not see in $C_{1}$.
    \item[4)] If $p \in \Pi - ids(U_{1}) \bigcup ids(U_{2}) - D_{1}$, no additional information is obtained by the process $p$. 
\end{itemize}

In case 1 where $p \in ids(U_{2})$, $p$ will see $ids(U_{2})$ in both $C_{2}$ and $C_{2}^{'}$. In cases 2, 3 and 4 where $p \notin ids(U_{2})$, the additional processes are those in $ids(U_{2})$ that process $p$ does not see in $C_{1}$ and $C_{1}^{'}$. Therefore, the process $p$ will see the same set of processes in $C_{2}$ and $C_{2}^{'}$. If any vertex is shared by the simplices representing $C_{1}$ and $C_{1}^{'}$, the simplices representing $C_{2}$ and $C_{2}^{'}$ will share the vertex with the same process id. The schedule $\gamma$ from $C_{1}$ to $C_{3}$ and $\gamma^{'}$ from $C_{1}^{'}$ to $C_{3}^{'}$ are indistinguishable to the process $p$. Therefore, the schedule $\gamma$ from $C_{2}$ to $C_{4}$ and $\gamma^{'}$ from $C_{2}^{'}$ to $C_{4}^{'}$ are indistinguishable to the process $p$. The canonical neighbors of the simplices representing $C_{3}$ and $C_{3}^{'}$ will share the vertex with the process id $p$.

\end{proof}

Finally, we need to consider the case that two $n$-simplices in $Ch^{r_{m}}(\mathcal{I})$ are not reached from a single initial configuration $C$.

\begin{lemma}
\label{the:canonical_neighbor_common_boundary_lemma_2}
    If two $n$-simplices in $Ch^{r_{m}}(\sigma_{1})$ and $Ch^{r_{m}}(\sigma_{2})$ both have a canonical neighbor with label $U_{2}$ and share some vertex $v$ with the process id $p$, then the two vertices in their canonical neighbors having the process id $p$ will have the same output value when using partial protocol $\delta_{U_{2}}$.
\end{lemma}

\begin{proof}
    Let $s_{b}$ be the largest common simplex of $\sigma_{1}$ and $\sigma_{2}$.  The vertex $v$ is in $Ch^{r_{m}}(s_{b})$. So $carrier(v, \sigma_{1})$ and $carrier(v, \sigma_{2})$ are the same simplex $s_{c}$, which is a subsimplex of $s_{b}$. Therefore, the two vertices with the process id $p$ of two canonical neighbors are both in $Ch^{r_{m}}(s_{c})$. By the symmetry of our construction, we know that the two vertices are the same and have the same output value when using the partial protocol $\delta_{U_{2}}$.
\end{proof}

We summarize the above lemmas into a theorem:
\begin{theorem}
    
\label{the:canonical_neighbor_existence_theorem}
    For an $n$-simplex $s^{n}$ in $Ch^{r_{m}}(\mathcal{I})$ that has an intersection with $\mathbb{F}_{r_{m}}(U)$, there exists an canonical neighbor $N(s^{n}, U)$ of $s^{n}$ with the label $U$ if $r_{m} \geq 2$.
\end{theorem}

\subsection{Active distance}
Active distance is an essential tool in our adversarial strategy. The active distance between subgraphs $A$ and $B$ of a graph $G$ is defined to be the minimal number of edges between non-terminated vertices in any path from $A$ to $B$ in $G$. An important property is that the active distance between $\chi(A)$ and $\chi(B)$ in the non-uniform chromatic subdivision of $G$ is double the active distance between $A$ and $B$ in $G$. This property allows the adversary to subdivide a graph finitely to make certain subgraphs distant enough from each other.

\section{A necessary and sufficient condition to have no extension-based impossibility proofs}
\label{sec:nccondition}

\subsection{An enhanced necessary condition}
\label{sec:nccondition:glue_protocols}

In Section \ref{sec:preparations}, we give a necessary condition for finalization after phase 1. The conditions of Theorem \ref{the:partial_info_theorem} imply a partial protocol for each subcomplex $\mathbb{F}_{1}(U)$ of $Ch(\mathcal{I})$. But these partial protocols are not irrelevant, since the output values of some configurations should be the same. Taking as an example the initial configuration of the consensus task of three processes with input values $(0, 1, 2)$, there is a partial protocol for each subcomplex $\mathbb{F}_{1}(U)$ where $U$ is a 0-simplex whose process id is $p$. As shown in Figure \ref{img:3_processes_consensus_task}, the adversary can assign the input value of $p$ to each vertex. But in phase 1, the prover can submit a chain of queries corresponding to the schedule $\{p_{0}\}, \{p_{1}, p_{2}\}, \{p_{1}\}, \{p_{1}\}, ....$ where the process $p_{1}$ is scheduled until it outputs a value $b$. If the output value $b$ is 0, then the protocol for $\mathbb{F}_{1}(\{p_{2}\})$ is not possible when the prover chooses the schedule $\{p_{2}\}$ at the end of phase 1. If the output value $b$ is 2, then the protocol for $\mathbb{F}_{1}(\{p_{0}\})$ is not possible when the prover chooses the schedule $\{p_{0}\}$ at the end of phase 1. From this example, we know that there exist some requirements on the set of partial protocols. Will these partial protocols have to give the same output for each configuration on the boundary? The answer is no since if so, there will be a protocol for the task $(\mathcal{I}, \mathcal{O}, \Delta)$, which contradicts the definition of extension-based impossibility proofs. We will start with the simplest type, i.e. restricted extension-based proofs.

\begin{figure}[h]
  \centering
  \includegraphics[width=\linewidth]{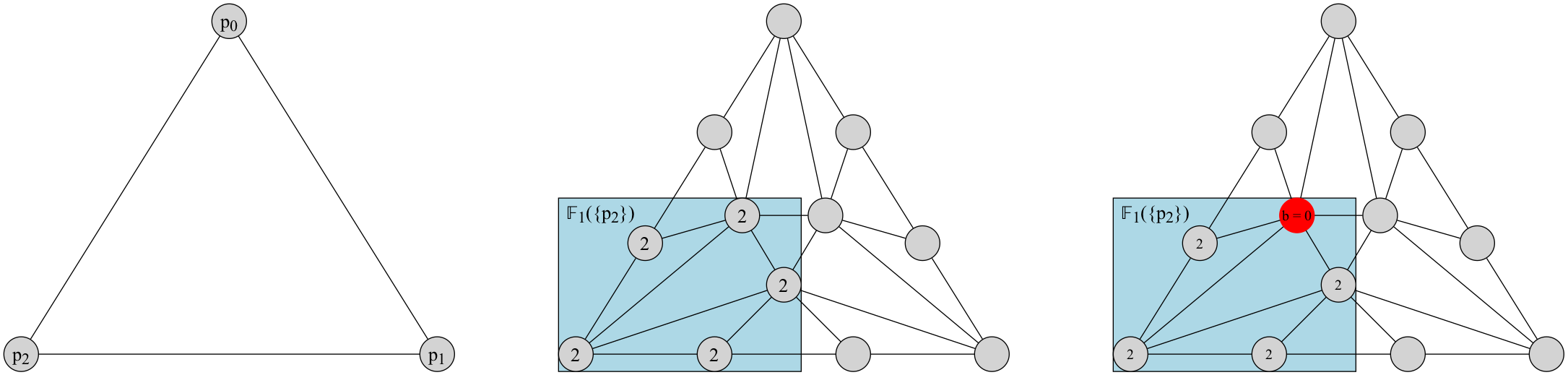}
  \caption{The 3-processes consensus task}
  \label{img:3_processes_consensus_task}
\end{figure}

First, we will discuss the most basic requirement for the set of partial protocols. For each simplex $s^{k}$ in
$\mathbb{F}_{1}(U_{1}) \cap \mathbb{F}_{1}(U_{2})$, $CEN(s^{k})$ is defined as the configuration reached from $s^{k}$ via a schedule that repeats the set of processes $ids(s^{k})$ until all processes in $ids(s^{k})$ terminate. We say that two partial protocols $\delta_{U_{1}}$ and $\delta_{U_{2}}$ are compatible if the output values of $CEN(s^{k})$ are the same for each possible $s^{k}$. A set of partial protocols is compatible if any two partial protocols are compatible. We can prove that the set of partial protocols should be compatible. Therefore, this is a necessary condition for finalization after phase 1.

Let $s^{k}$ be any such $k$-simplex in the standard chromatic subdivision of $C$, where $C$ is an initial configuration. In phase 1, the prover can submit a chain of queries so that the output of $CEN(s^{k})$ is returned by the protocol. Since the choice of $s^{k}$ is finite, there is only a finite number of chains of queries. The prover can end phase 1 by choosing any set of processes as the schedule, which means that any partial protocol for $\mathbb{F}_{1}(U)$ must adopt the existing output values for $CEN(s^{k})$. This is why the adversary for the consensus task cannot finalize after phase 1 as what the adversary for the $(n, 2)$-set agreement task can do: there exists a partial protocol with respect to each $U \in \mathcal{I}$, but some intersection simplices cannot be compatible.

We shall now prove that this requirement is sufficient for an adversary to finalize after phase 1, which is one of the main results in this paper: the adversary can win against any restricted extension-based prover after phase 1 if there is a partial protocol with respect to each $U$ and the set of partial protocols is compatible. Like in previous work \cite{Alistarh19, Brusse21}, we will design an adversary that can win against any extension-based prover. The adversary maintains a partial specification $\delta$ of a protocol and some invariants. All these previous sets of invariants are closely related to some specific task, such as the $(n, k)$-set agreement task.

\subsection{Adversarial strategy in phase 1}
\label{sec:nccondition:adversarial_strategy}

We define an adversarial strategy so that, before and after each query made by the prover in phase 1, the adversary is able to maintain some invariants. The adversary maintains a partial specification of $\delta$ (the protocol it is adaptively constructing) and an integer $t \leq 0$. The integer $t$ represents the number of non-uniform chromatic subdivisions of the input complex, $S^{0}$, that it has performed. The specification of $\delta$ is defined using a sequence of complexes indexed by an integer $t$, and each complex $S^{t}$ represents all the configurations that can be reached from the initial configurations by a $t$-round schedule. When the prover submits a query in the first phase, the adversary may increment the value of $t$ and set $\delta$ values for some vertices so that it can respond to the prover while maintaining the invariants. Note that our invariant (1) is the same as the invariant (1) in \cite{Alistarh19}, and our adversarial strategy is inspired by theirs. Invariants will use the concepts defined in Section \ref{sec:preparations}. 
\begin{itemize}
    \item[(1)] For each $0 \leq r < t$ and each vertex $v \in S^{r}$ , $\delta(v)$ is defined. If $v$ is a vertex in $S^{t}$, then $\delta(v)$ is undefined or $\delta(v) \neq \perp$. If $s$ is the state of a process in a configuration that has been reached by the prover, then $s$ is a vertex in $S^{r}$, for some $0 \leq r \leq t$, and $\delta(s)$ is defined.
    \item[(2)] Let $s^{k_{1}}$ be a $k_{1}$-simplex representing a terminated configuration in $S^{t}$. Suppose that $s^{k_{1}}$ is in the subdivision of some $n$-simplex $s^{n}$ of $\mathbb{F}_{r_{m}}(U)$. Then $s^{k_{1}}$ is given a simplex $U^{'} \in \mathcal{I}$ as its label indicating which partial protocol its $\delta$ value is from. If $U = U^{'}$, then $\delta(s^{k})$ is defined as the limitation of $\delta_{U}(s^{n})$ to $ids(s^{k})$ (that is, removing the output values of all processes not in $ids(s^{k})$). If $U \neq U^{'}$, then $\delta(s^{k})$ is defined as the limitation of $\delta_{U^{'}}(N(s^{n}, U^{'}))$ to $ids(s^{k})$ where $N(s^{n}, U^{'})$ is the canonical neighbor of $s^{n}$ with the label $U^{'}$, and there exists a $k_{2}$-simplex $s^{k_{2}}$ in $\mathbb{F}_{t}(U^{'})$ terminated with the label $U^{'}$ such that there is a path consisting of terminated vertices with the label $U^{'}$ from $s^{k_{1}}$ to $s^{k_{2}}$.
    \item[(3)] Let $s^{k_{1}}$ and $s^{k_{2}}$ be two terminated simplices with different labels in $S^{t}$. The active distance between $s^{k_{1}}$ and $s^{k_{2}}$ in $S^{t}$ is at least 3.
\end{itemize}
The idea behind the second invariant is that the canonical neighbor with the label $U^{'}$ is only defined for simplices in $Ch^{r_{m}}(\mathcal{I})$ that have an intersection with $\mathbb{F}_{r_{m}}(U^{'})$. Therefore, we require that the simplex $s^{k}$ be connected to some simplex of $\mathbb{F}_{t}(U^{'})$. We will prove if such a path exists, the simplex $s^{n}$ where $s^{k}$ is from will have an intersection with $\mathbb{F}_{r_{m}}(U^{'})$ (i.e. have a canonical neighbor with the label $U^{'}$ by Theorem \ref{the:canonical_neighbor_existence_theorem} ).

Initially, the adversary sets $\delta(v) = \perp$ for each vertex $v \in S^{0}$. Then it subdivides $S^{0}$ to construct $S^{1}$ and set $t = 1$. To keep invariant (2) true, repeat this subdivision $r_{m} + r_{a}$ times and set $t = r_{m} + r_{a} + 1$, for some constant number $r_{a} \geq 1$. The $\delta$ values of all vertices in $S^{r}$ where $r \leq r_{m} + r_{a}$ will be assigned the value $\perp$, and all vertices in $S^{r_{m} + r_{a} + 1}$ will be undefined. Before the first query, the prover has only reached the initial configurations and there are no defined vertices. Therefore, the invariants are satisfied at the beginning of phase 1.

Now suppose that the invariants are satisfied immediately prior to a query $(C, P)$ by the prover, where $C$ is a configuration previously reached by the prover and $P$ is a set of active processes in $C$ poised to access the same immediate snapshot object. There is an integer $0 \leq r \leq t$ such that the states of the processes in $P$ correspond to a simplex in $S^{r}$ by the first invariant. Since the processes in $P$ are still active in $C$, $\delta(v) = \perp$ for each such vertex $v$ and $r < t$. Let $\sigma$ be the simplex in $S^{r + 1}$ whose vertices represent the states of processes in $P$ in the resulting configuration.

If $0 \leq r < t - 1$, then $\delta$ is already defined for $\sigma$ in $S^{r  + 1}$ according to invariant (1). The adversary does not do anything and returns the defined values of $\delta$. Since $\delta$ has not changed, the invariants are still valid. 

If $r = t - 1$, then $\sigma$ in $S^{t}$ may contain some undefined vertices. The adversary has to define the $\delta$ values for these undefined vertices. For each undefined vertex $v$ having a process id $p$ in $\sigma$, there are three assignment rules:
\begin{itemize}
    \item[1)] If there is label $U$ for the vertex $v$ such that $v$ is in some subdivision of some $n$-simplex $s^{n}$ of $\mathbb{F}_{r_{m}}(U)$ and the active distance between $v$ and any terminated configuration with a different label is at least 3, the adversary will give $v$ the label $U$ and set $\delta(v) = \delta_{U}(v^{'})$, where $v^{'}$ is the vertex of $s^{n}$ with the same process id as $v$. 
    \item[2)] If there is some vertex that is adjacent to $v$ and already terminated with the label $U^{'}$, and the active distance between $v$ and any terminated configuration with a label different from $U^{'}$ is at least 3, the adversary can terminate the vertex $v$ with the label $U^{'}$ although $v$ is not a vertex in $\mathbb{F}_{t}(U^{'})$. The adversary sets $\delta(v) = \delta_{U^{'}}(v^{'})$, where $v^{'}$ is the vertex of $N(s^{n}, U^{'})$ with the same process id as $v$. We claim that the canonical neighbor of $s^{n}$ with the label $U^{'}$ is always defined. 
    \item[3)] Otherwise, the adversary will set $\delta(v) = \perp$. 
\end{itemize}

We introduce rule (2) to avoid an infinite chain of queries. However, this rule cannot be used frequently, since only $n$-simplices that have an intersection with $\mathbb{F}_{r_{m}}(U^{'})$ will have a canonical neighbor with the label $U^{'}$. We impose a restriction on this rule: If rule (2) is applied to a vertex in the $n_{1}$-th complex $S^{n_{1}}$, then it can only be used again in the $n_{2}$-th complex $S^{n_{2}}$ if $n_{2} - n_{1} >= r_{s}$ for some constant number $r_{s}$. Using this restriction, we prove the claim in rule (2). A geometric proof of Lemma \ref{the:glue_protocols:prover_fails:no_abuse} is given in Appendix \ref{appendix:geometric_proof}.

\begin{lemma}
\label{the:glue_protocols:prover_fails:terminated_vertices_not_reach}
If a $k$-simplex $s^{k}$ in $S^{t}$ has an active distance of at least 1 to every vertex that is terminated with some label $U^{'}$ and has an active distance of at least 2 to every undefined vertex in $\mathbb{F}_{t}(U^{'})$, then any vertex of $\chi^{t^{'} - t}(s^{k}, \delta)$ in $S^{t^{'}}$, where $t^{'} > t$, will not be terminated with the label $U^{'}$.
\end{lemma}

\begin{proof}
Let $Q(U^{'})$ be the subcomplex of $S^{t}$ consisting of $\mathbb{F}_{t}(U^{'})$ and all simplices that are terminated with the label $U^{'}$. 

Suppose that some vertex of $\chi^{t^{'} - t}(s^{k}, \delta)$ in $S^{t^{'}}$, for some $t^{'} > t$, is terminated with the label $U^{'}$. Then there is a finite path $v_{0}, v_{1}, .... v_{n}$ consisting of terminated vertices with the label $U^{'}$ from $\chi^{t^{'} - t}(Q(U^{'}), \delta)$ to $\chi^{t^{'} - t}(s^{k}, \delta)$. We assume this path $v_{0}, v_{1}, .... v_{n}$ is the shortest path that the adversary can construct from $\chi^{t^{'} - t}(Q(U^{'}), \delta)$ to $\chi^{t^{'} - t}(s^{k}, \delta)$ for some $t^{'}$ when processing all possible queries submitted by a prover. The length of the path $v_{0}, v_{1} \cdots v_{n}$ is more than 2, since the active distance is the number of edges between non-terminated vertices.

When the vertex $v_{n - 2}$ is terminated with the label $U^{'}$ in $S^{t_{1}}$ where $t \leq t_{1} < t^{'}$, $v_{n - 2}$ is not in the set $V$ of vertices, each of which is adjacent to $\chi^{t_{1} - t}(s^{k}, \delta)$. This can be proved as follows: If a vertex of $\chi^{t_{1} - t}(s^{k}, \delta)$ is terminated with some label $U^{''} \neq U^{'}$, then $v_{n - 2}$ cannot be adjacent to this vertex by invariant (3). If $v_{n - 2}$ is adjacent to an undefined vertex $v_{u}$ of $\chi^{t_{1} - t}(s^{k}, \delta)$, then the adversary can terminate $v_{u}$ in $\chi^{t_{2} - t}(s^{k}, \delta)$ for some $t_{2} \geq t_{1}$ (just consider what the adversary does when an extension-based prover submits a chain of queries to ask the $\delta$ value of $v_{u}$). This means a shorter path from some subdivision of $Q(U^{'})$ to some subdivision of $s^{k}$, a contradiction to the assumption that $v_{0}, v_{1}, .... v_{n}$ is the shortest path.

Therefore, in the complex $S^{t^{'}}$, $v_{n - 1}$ is in $\chi^{t^{'} - t_{1}}(V, \delta)$ and $v_{n}$ is in $\chi^{t^{'} - t}(s^{k}, \delta)$. When $v_{n - 1}$ is terminated in the complex $S^{t_{3}}$, where $t_{3} - t_{1} \geq r_{s}$, the length of a path consisting of undefined vertices between $\chi^{t_{3} - t_{1}}(V, \delta)$ and $\chi^{t_{3} - t}(s^{k}, \delta)$ is greater than 1, contradicting the assumption that $v_{n - 1}$ is adjacent to $v_{n}$ in $S^{t^{'}}$. Note that we do not use the active distance here, since the active distance between $V$ and $\chi^{t_{1} - t}(s^{k}, \delta)$ can be 0. Our concern is the paths in which each vertex is undefined. Therefore, no vertex of $\chi^{t^{'} - t}(s^{k}, \delta)$ in $S^{t^{'}}$, for each $t^{'} > t$, is terminated with the label $U^{'}$.

\end{proof}

\begin{lemma}
\label{the:glue_protocols:prover_fails:no_abuse}
    If a vertex $v$ is terminated by rule (2) with label $U^{'}$ and $v$ is in the subdivision of an $n$-simplex $s^{n}$ of $Ch^{r_{m}}(\mathcal{I})$, then $s^{n}$ has a canonical neighbor with label $U^{'}$.
\end{lemma}

\begin{proof}
Let $Q$ be the subcomplex of $Ch^{r_{m}}(\mathcal{I})$, which consists of all $n$-simplices that do not have an intersection with $\mathbb{F}_{r_{m}}(U^{'})$. In the complex $S^{r_{m} + r_{a}}$, $\delta(v)$ is defined as $\perp$ for each vertex $v$. Therefore, the active distance between $\chi^{r_{a}}(Q, \delta)$ and $\mathbb{F}_{r_{m} + r_{a}}(U^{'})$ is greater than 2. According to Lemma \ref{the:glue_protocols:prover_fails:terminated_vertices_not_reach}, no vertex in $\chi^{t^{'}}(Q, \delta)$, where $t^{'} > r_{m} + r_{a} $, is terminated with the label $U^{'}$.

\end{proof}

In our definition of canonical neighbors, only an $n$-simplex $s^{n}$ that shares some vertices with $\mathbb{F}_{U^{'}}$ has a canonical neighbor with the label $U^{'}$. But if a vertex $v$ in $\mathbb{F}_{t}(U)$ is in the subdivision an $n$-simplex that does not have a canonical neighbor with the label $U^{'}$ and is terminated using the second assignment rule, then its $\delta$ value is illegal. Lemma \ref{the:glue_protocols:prover_fails:no_abuse} shows that this situation is impossible. The vertex $v$ terminated by rule (2) will only be in the subdivision of $n$-simplices in $Ch^{r_{m}}(\mathcal{I})$ that have a canonical neighbor with the label $U^{'}$.

If the adversary sets $\delta(v) \neq \perp$ for each vertex $v \in \sigma$, invariant (1) holds. Otherwise, some vertex terminates with the value $\perp$, and the requirement for $S^{t}$ is not satisfied. The adversary will define $\delta(v) = \perp$ for each vertex $v \in S^{t}$ where $\delta(v)$ is undefined, subdivide $S^{t}$ to create $S^{t+1}$ and increment $t$. Invariant (1) continues to hold after the subdivision. 

If two vertices are connected by a terminated path, they will have the same active distance to any other terminated configuration. Invariant (3) holds since newly terminated vertices will be checked in rules (1) and (2), and a subdivision doubles the active distances between terminated configurations with different labels. 
Invariant (2) is more complicated to analyze. It is affected only when $r = t - 1$ and some undefined vertices in $S^{t}$ are terminated, that is, given some values other than $\perp$. By the induction hypothesis, we only have to discuss the simplex $\sigma_{n}$ consisting of all newly terminated vertices. If all vertices are terminated by rule (1), then invariant (2) still holds. Otherwise, some vertex of $\sigma_{n}$ is assigned by rule (2) with the label $U^{'}$. By construction, we know that there is an adjacent vertex $v_{ter}$ that is terminated with the label $U^{'}$ before this query. Using the induction hypothesis, there exists an $k_{2}$-simplex $s^{k_{2}}$ terminated with the label $U^{'}$ such that there is a path consisting of terminated vertices from $v_{ter}$ to $s^{k_{2}}$. There is also a path consisting of terminated vertices from $\sigma_{n}$ to $s^{k_{2}}$. Invariant (2) continues to hold.

Therefore, the adversary can maintain all invariants after the prover submits a query in the first phase.

\subsection{The prover cannot win in the first phase}
\label{sec:nccondition:prover_not_win_in_first_phase}
There are only two ways the prover can win in phase 1. First, the prover can find a violation of the task specification, i.e. some terminated configuration should not be assigned with its output values.

For each $n$-simplex $s^{n}$ in $Ch^{r_{m}}(C)$ , the outputs in a non-uniform chromatic subdivision of it will be $\delta_{U}(s^{n})$ or $\delta_{U^{'}}(N(s^{n}, U^{'}))$ for some label $U^{'}$. In both cases, there is no violation of the task specification. 
$\delta_{U}(s^{n})$ will not cause any problem since $\delta_{U}$ is assumed to be a partial protocol for $\mathbb{F}(U)$.
We have shown in Lemma \ref{the:canonical_neighbor_no_safety_issue} that if a $k$-simplex $s^{k}$ reached from some $n$-simplex $s^{n}$ in $\mathbb{F}_{r_{m}}(U)$ is assigned with the limitation of $\delta_{U^{'}}(N(s^{n}, U^{'}))$ to $ids(s^{k})$, then it will not cause any violation of the task specification. Furthermore, there is no conflict of output values at the intersection of different $n$-simplices according to Lemma \ref{the:canonical_neighbor_common_boundary_lemma}, \ref{the:canonical_neighbor_common_boundary_lemma_2} and \ref{the:glue_protocols:prover_fails:no_abuse}. Therefore, the prover cannot find a violation of the task specification during phase 1.

\begin{lemma}
\label{the:restrict_ebf_no_violation_of_task_specification}
    Each terminated configuration in phase 1 will not violate the task specification.
\end{lemma}

Second, the prover can win by constructing an infinite chain of queries. But this is also impossible.
\begin{lemma}
\label{the:restrict_ebf_finite_chain_of_queries_lemma}
     Every chain of queries in phase 1 is finite.
\end{lemma}

\begin{proof}
We will use similar proof techniques as in \cite{Alistarh19}. Assume that there is an infinite chain of queries $(C_{j}, P_{j})$ and let $P$ be the set of processes that are scheduled infinitely often. Then, there exists a round number $j_{0} \geq 0$ such that, for all $j \geq j_{0}$ , $P_{j} \subseteq P$. Let $t_{0} \geq 1$ be the value of $t$ held by the adversary immediately prior to query $(C_{j}, P_{j})$. By invariant (1), for any $t > t_{0}$, no process has accessed $S^{t}$ in $C_{j_{0}}$ and, during this chain of queries, only processes in $P$ access $S^{t}$. 

Consider the first $j_{1} \geq j_{0}$ such that each process in $P_{j_{1}}$ is poised to access $S^{t_{0} + k}$ in $C_{j_{1}}$ where $t_{0} + k \geq r_{m}$ for some $k \geq 1$. The states of the processes in $P_{j_{1}}$ in $C_{j_{1} + 1}$ correspond to a simplex $\sigma_{1}$ in $S^{t_{0} + k}$. Since no process in $P$ will terminate, the adversary will subdivide $S^{t_{0}}$ k times to construct $S^{t_{0} + k}$. By the property of active distance, the active distance between two terminated configurations with different labels is at least 6 in $S^{t_{0} + k}$.

1) Suppose that $\sigma_{1}$ has an active distance of at least 3 to any terminated configuration. The adversary will then set the $\delta$ values of $\sigma_{1}$ as the limitation of $\delta_{U}(s^{n})$ to $ids(\sigma_{1})$, where $\sigma_{1}$ is in the subdivision of $s^{n}$ of $Ch^{r_{m}}(\mathcal{I})$. Therefore, the processes in $\sigma_{1}$ will terminate, which is a contradiction.

2) Suppose that $\sigma_{1}$ has an active distance of less than 3 to a set of terminated configurations $\{\sigma_{U_{1}}\}$ with some label $U_{1}$. Note that $U_{1}$ is unique since the active distance between a configuration terminated with label $U_{1}$ and a configuration terminated with label $U_{2} \neq U_{1}$ is at least 6. In other words, $\sigma_{1}$ has an active distance of at least 3 to a terminated configuration with label $U_{2} \neq U_{1}$. Consider the first $j_{2} > j_{1}$ such that each process in $P_{j_{2}}$ is poised to access $S^{t_{0} + k + 1}$  in $C_{j_{2}}$. The set of states of the processes in $P_{j_{2}}$ in $C_{j_{2} + 1}$ corresponds to a simplex $\sigma_{2}$ in $S^{t_{0} + k + 1}$. Let $P^{'}$ be the set of processes that have accessed $S^{t_{0} + k}$ in $C_{j_{2}}$ . Since each process in $P_{j_{1}} \cup P_{j_{2}}$ has already accessed $S^{t_{0} + k}$, $P_{j_{1}} \cup P_{j_{2}} \subseteq P^{'}$. Therefore, the states of $P^{'}$ in $C_{j_{2}}$ form a simplex $\sigma_{1}^{'}$ in $S^{t_{0} + k}$ and $\sigma_{1} \subseteq \sigma_{1}^{'}$. Let $U_{2}$ be a label different from $U_{1}$. In $S^{t_{0} + k}$, the active distance between $\sigma_{1}^{'}$ and any terminated configuration $\sigma_{U_{2}}$ with label $U_{2}$ is at least 2, since the diameter of $\sigma_{1}^{'}$ is only 1 and the active distance between $\sigma_{1}$ and $\sigma_{U_{2}}$ is at least 3. In the complex $S^{t_{0} + k + 1} = \chi(S^{t_{0} + k})$, the active distance between $\chi(\sigma_{1}^{'})$ and any terminated configuration with the label $U_{2}$ is at least 4. 

If some vertex in $\sigma_{2}$ is in $\mathbb{F}_{t}(U_{1})$, then this vertex will be terminated with the label $U_{1}$. Hence the adversary defines $\delta(v) \neq \perp$ after query $(C_{j_{2}} , P_{j_{2}})$, i.e. some process in $P_{j_{2}} \subseteq P$ terminates and this is a contradiction.

Otherwise, no vertex in $\sigma_{2}$ is in $\mathbb{F}_{t}(U_{1})$. Consider the first $j_{3} > j_{2}$ such that each process in $P_{j_{3}}$ is poised to access $S^{t_{0} + k + 2}$ in $C_{j_{3}}$. Let $P^{''}$ be the set of processes that have accessed $S_{t_{0} + k + 1}$ in $C_{j_{3}}$, and the states of $P^{''}$ in $C_{j_{3}}$ form a simplex $\sigma_{2}^{'}$ in $S_{t_{0} + k + 1}$ and $\sigma_{2} \subseteq \sigma_{2}^{'}$. The active distance between $\sigma_{2}^{'}$ and any terminated configuration with the label $U_{2}$ is at least 3. In the complex $S^{t_{0} + k + 2} = \chi(S^{t_{0} + k + 1})$, the active distance between $\chi(\sigma_{2}^{'})$ and any terminated configuration with the label $U_{2}$ is at least 6. 
This remains true for any simplex $\sigma_{i}^{'}$ where $i \geq 3$, if we define $\sigma_{i}^{'}$ in the same way. Therefore, the active distance between $\chi(\sigma_{i}^{'})$ and terminated configurations with the label $U_{2}$ is always at least 3 (part of the conditions to use the second assignment rule).

Let $\sigma_{2}^{''}$ be the subsimplex of $\sigma_{2}^{'}$ whose process ids are $P_{j_{3}}$.
If the active distance between $\sigma_{2}^{''}$ and each terminated configuration $\sigma_{U_{1}}$ with label $U_{1}$ is greater than 0, then after several rounds of subdivisions, the active distance between $\sigma_{i}$ and $\{\sigma_{U_{1}}\}$ will be at least 3. The simplex $\sigma_{i}$ will be assigned with some label $U$, where $\sigma_{i}$ is in $\mathbb{F}_{t}(U)$, which is a contradiction.

Otherwise, the active distance between $\sigma_{2}^{''}$ and some terminated configurations $\sigma_{U_{1}}$ with the label $U_{1}$ is 0. There must be some vertex $v$ of $\sigma_{2}^{'}$ adjacent to a terminated vertex in $\sigma_{U_{1}}$. If assignment rule (2) is allowed to be used in the complex $S^{t_{0} + k + 2}$, $v$ can be terminated with the label $U_{1}$ after query $(C_{j_{3}} , P_{j_{3}})$. If rule (2) is not allowed in $S^{t_{0} + k + 2}$, we can continue the analysis in $S^{t_{0} + k + 3}$, $S^{t_{0} + k + 4}$, etc. Since $r_{s}$(the number of complexes that the adversary has to skip before using rule (2) again) is a constant number, this sequence will be finite and some process in $P$ will be terminated in $S^{t_{0} + k + k^{'}}$ for some $k^{'}$, which is a contradiction. 
\end{proof}

\subsection{Finalization after phase 1}
\label{sec:nccondition:finalization_after_phase_1}

The prover must end phase 1, since it cannot win in phase 1, by choosing a configuration $C^{'} \in \mathcal{A}^{'}(1)$. Let $C$ be an initial configuration from which $C^{'}$ is reached. Let the simplex $U$ be the simplex in $\mathcal{I}$ that has process ids as those of the first process set in $\alpha(2)$ and has input values of these processes in the initial configuration $C$. From now on, the adversary will focus on the subcomplex $\mathbb{F}_{t}(U)$.

By invariant (1), each vertex in $S^{t}$ is undefined or defined with an output value in $\delta$. If a $k$-simplex $s^{k}$ representing a terminated configuration is in the subdivision of an $n$-simplex $s^{n}$ of $\mathbb{F}_{t}(U)$, its $\delta$ value is either the limitation of $\delta_{U}(s^{n})$ to $ids(s^{k})$, or the limitation of $\delta_{U^{'}}(N(s^{n}, U^{'}))$ to $ids(s^{k})$ for some simplex $U^{'}$. Note that there can exist multiple labels different from $U$ simultaneously.

The complex $S^{t}$ has some undefined vertices. We will prove that the adversary can finalize by assigning a value to each undefined vertex in a non-uniform chromatic subdivision of $\mathbb{F}_{t}(U)$. Each configuration reached in phase 2 and later phases will be contained in this subdivision. Then the adversary can answer all queries submitted by the prover in the subsequent interaction. As in \cite{Alistarh19}, the prover will be forced to choose a terminal configuration at the end of some future phase and lose in the next phase.
Let $s^{k}$ be an $k$-simplex of $\mathbb{F}_{1}(U) \cap \mathbb{F}_{1}(U^{'})$ in $Ch(\mathcal{I})$, and $CEN(s^{k})$ be the $k$-simplex in $\mathbb{F}_{r_{m}}(U) \cap \mathbb{F}_{r_{m}}(U^{'})$ reached from $s^{k}$ via a schedule that repeats the set of processes $ids(s^{k})$. 
The output of $CEN(s^{k})$ is the same by the partial protocols $\delta_{U}$ and $\delta_{U^{'}}$ since $\delta_{U}$ and $\delta_{U^{'}}$ are compatible by assumption. The high-level idea of this section is to use these shared outputs to connect the output values $\delta_{U}(s^{n})$ and $\delta_{U^{'}}(N(s^{n}, U^{'}))$ in the subdivision of each $n$-simplex $s^{n}$ of $\mathbb{F}_{r_{m}}(U)$ to build a partial protocol for $\mathbb{F}_{t}(U)$.

Until now, we have not used the condition that the task $(\mathcal{I}, \mathcal{O}, \Delta)$ is colorless. Our adversarial strategy until now can be directly extended to colored tasks. The following discussions are based on the condition that the task is colorless as defined in Section \ref{sec:preliminaries}. Let $s^{n}$ be some $n$-simplex in $Ch^{r_{m}}(\sigma)$ where $\sigma$ represents an initial configuration. Recall that the colorless condition implies a property: for some subsimplex $s^{k}$ of $s^{n}$ and each other vertex whose id is not in $ids(s^{k})$ has an equal or larger carrier than $carrier(s^{k})$, each vertex not in $ids(s^{k})$ can use some values in $\delta_{U}(s^{k})$ or $\delta_{U^{'}}(s^{k})$ as its output value. 

For each simplex $U^{'}$, there is a complex $Q_{1}(U, U^{'}) = \mathbb{F}_{1}(U) \cap \mathbb{F}_{1}(U^{'})$ \label{syb:Q} consisting of all the intersection simplices in $Ch^{1}(\mathcal{I})$. Similarly, we define $Q_{i}(U, U^{'}) = \mathbb{F}_{i}(U) \cap \mathbb{F}_{i}(U^{'})$ for each integer $i$. For example, when $U_{1} \nsubseteq U_{2}$, $Q_{1}(U, U^{'}) = \mathbb{F}_{1}(U) \cap \mathbb{F}_{1}(U^{'})$ equals $(\mathbb{F} \downarrow U * U^{'})(U * U^{'})$.
Let $s^{n}$ be an $n$-simplex in $\mathbb{F}_{r_{m}}(U)$ that has an intersection with $\mathbb{F}_{r_{m}}(U^{'})$. Let $SHA(s^{n}, U^{'})$ \label{syb:SHA} be the simplex shared by $s^{n}$ and $Q_{r_{m}}(U, U^{'})$ that contains all vertices of $s^{n}$ that are in $\mathbb{F}(U^{'})$. 
For each simplex $s^{k}$ in $Q_{r_{m}}(U, U^{'})$, we define the {\em minimum carrier} of $s^{k}$ \label{syb:mc}, denoted as $mc(s^{k}, Q_{1}(U, U^{'}))$, as $\cap_{v \in s^{k}} carrier(v, Q_{1}(U, U^{'}))$. An example is given in Figure \ref{img:categories_of_subdivision}.

\begin{figure}[h]
  \centering
  \includegraphics[width=\linewidth]{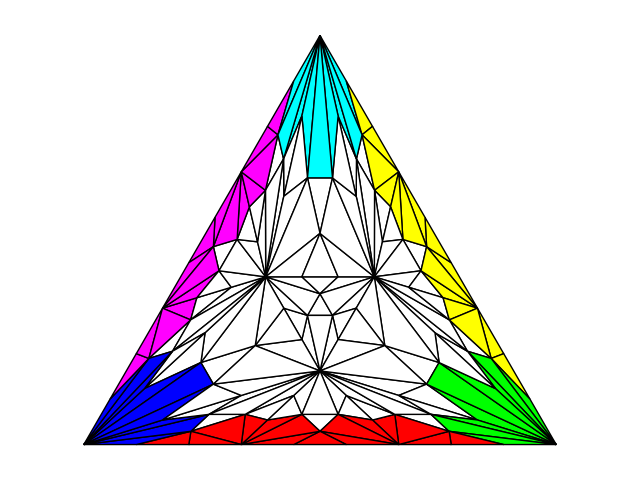}
  \caption{2-simplices with the same minimum carrier in $Ch^{2}(s^{2})$}
  \label{img:categories_of_subdivision}
\end{figure}

We say that an $n$-simplex in $\mathbb{F}_{r_{m}}(U)$ that has an intersection with $\mathbb{F}_{r_{m}}(U^{'})$ belongs to a {\em category $Cat(s^{k})$} whose label is a simplex $s^{k}$ in $Q_{1}(U, U^{'})$ if $mc(SHA(s^{n}, U^{'}), Q_{1}(U, U^{'}))$ is $s^{k}$. The projection of $s^{n}$ with the label $U^{'}$, denoted by $Prj(s^{n}, U^{'})$  \label{syb:prj}, is defined as the subsimplex of $s^{n}$ whose process ids are in $ids(s^{k})$, where $s^{k}$ is the label of the category that contains $s^{n}$. If $s^{k}$ is a facet of $Q_{1}(U, U^{'})$, $Prj(s^{n}, U^{'})$ = $SHA(s^{n}, U^{'})$. If $s^{k}$ is not a facet, $Prj(s^{n}, U^{'})$ is a subsimplex of $SHA(s^{n}, U^{'})$. A property is that if we consider the processes in the configuration $s^{n}$, processes in $Prj(s^{n}, U^{'})$ have the minimum carrier in $\mathcal{I}$, since in the last full round schedule of the full $r_{m}$-round from some initial configuration to $s^{n}$, processes in $Prj(s^{n}, U^{'})$ are the first set of processes to take a step.


Consider the complex $S^{t}$ and a set $S$ of vertices in $S^{t}$. We say that the adversary {\em adds a layer of $k$-dimensional output} $\{(p_{0}, o_{0}), (p_{1}, o_{1}) \cdots (p_{k}, o_{k})\}$ to $S$ in some $n$-simplices of $Ch^{r_{m}}(\mathcal{I})$(note that these $n$-simplices are in the complex $S^{r_{m}}$ rather than $S^{t}$) if given a vertex $v$ (in $S^{t}$) of the process id $p$ adjacent to $S$ in the subdivision of these $n$-simplices, the adversary sets $\delta(v) = o_{i}$ if $p = p_{i}$ for some $i$. Otherwise, let $P_{c}$ be the set of processes that are both in the ids of $carrier(v, Ch^{r_{m}}(\mathcal{I}))$ and $\{p_{0}, p_{1} \cdots p_{k}\}$, the adversary sets $\delta(v)$ as $o_{i}$ where $p_{i}$ is an arbitrary process in $P_{c}$. After the adversary adds a layer of output $\{(p_{0}, o_{0}), (p_{1}, o_{1}) \cdots (p_{k}, o_{k})\}$ to $S$, the outermost layer of $S$ in the subdivision of these $n$-simplices consists of the output values in $\{o_{0}, o_{1} \cdots o_{k}\}$.

To give readers an overview of what the adversary does, we present the simplest case where a set of vertices terminated with $U^{'}$ is in the subdivision of only one $n$-simplex.
If a set $S$ (in $S^{t}$) of terminated vertices with the label $U^{'}$ is totally contained in the subdivision of some $n$-simplex $s^{n}$ in $\mathbb{F}_{r_{m}}(U)$, and $s^{n}$ belongs to the category with the label $s^{k} \in \mathbb{F}_{1}(U) \cap \mathbb{F}_{1}(U^{'})$, then the adversary will add layers of output such that the output values of the outermost layer of $S$ change from $\delta_{U^{'}}(s^{n})$ to $\delta_{U^{'}}(Prj(s^{n}, U^{'}))$, then to $\delta_{U^{'}}(CEN(s^{k})) = \delta_{U}(CEN(s^{k}))$, finally to $\delta_{U}(Prj(s^{n}, U^{'}))$. Since $\delta_{U}(Prj(s^{n}, U^{'}))$ is a subset of $\delta_{U}(s^{n})$, we can say the output values of the outermost layer are in $\delta_{U}(s^{n})$ directly. The adversary then simply defines each undefined vertex $v$ in $\mathbb{F}_{t}(U)$ with the output value given by $\delta_{U}(v^{'})$ where $v^{'}$ is the vertex with the same process id as $v$ of the $n$-simplex in $\mathbb{F}_{r_{m}}(U)$ whose subdivision contains $v^{'}$.

Note that the example in the last paragraph is not a common situation. A set $S$ of vertices can be in the subdivision of multiple $n$-simplices of $\mathbb{F}_{r_{m}}(U)$. In this case, the added layers to $S$ will be more complicated due to possible conflicts. The details will be given in Section \ref{sec:nccondition:finalization_after_phase_1:stage3}.


The adversary constructs this partial protocol in four stages. In stage 1 and stage 2, the adversary makes some preparations for adding layers of outputs. In stage 3, the adversary adds enough layers to those terminated sets with a label $U^{'} \neq U$ such that the output values of the outermost layer are $\delta_{U}(Prj(s^{n}, U^{'}))$. In stage 4, the undefined vertices of $\mathbb{F}_{t}(U)$ will be assigned with the output values of $\delta_{U}(s^{n})$. Before presenting these details, we have to introduce some topological concepts that justify the addition of layers such that the output values of the outermost layer of $S$ change from $\delta_{U^{'}}(Prj(s^{n}, U^{'}))$ to $\delta_{U^{'}}(CEN(s^{k}))$.

\subsubsection{Some topological concepts}
\label{sec:nccondition:finalization_after_phase_1:concepts}
We use the definition in \cite{Spanier1966}. A $k$-dimensional simplicial complex is said to be {\em homogeneously $k$-dimensional} if every simplex is a face of some $n$-simplex of the complex. An n-dimensional {\em pseudomanifold} is defined to be a simplicial complex $\mathcal{K}$ such that
\begin{itemize}
    \item[(a)] $\mathcal{K}$ is homogeneously n-dimensional.
    \item[(b)] Every $(k - 1)$-simplex of $\mathcal{K}$ is the face of at most two $k$-simplices of $\mathcal{K}$.
    \item[(c)] If $s$ and $s^{'}$ are $k$-simplices of $\mathcal{K}$, there is a finite sequence $s = s_{1}, s_{2},\cdots, s_{m} = s^{'}$ of $k$-simplices of $\mathcal{K}$ such that $s_{i}$ and $s_{i+1}$ have an $(k - 1)$-simplex in common for $1 \leq i < m$.
\end{itemize}
There is a relation denoted by $s \sim s^{'}$ in the set of n-simplices of a pseudomanifold defined by $s \sim s^{'}$ if and only if there exists a finite sequence $s = s_{1}, s_{2},\cdots, s_{m} = s^{'}$ such that $s_{i}$ and $s_{i+1}$ have an $(n - 1)$-simplex in common, as in condition (c). Obviously, this relation is a reflexive, symmetric, and transitive relation, and therefore an equivalence relation. An important property proved in \cite{Spanier1966} is that any subdivision of a pseudomanifold is also a pseudomanifold. 
 
The second concept is about a subcomplex of the $i$-th chromatic subdivision of a $k$-simplex. We define the non-boundary complex of $Ch^{i}(s^{k})$ as the simplicial complex consisting of those simplices in $Ch^{i}(s^{k})$, each vertex of which has a carrier $s^{k}$. This non-boundary complex is denoted by $Ch\_int^{i}(s^{k})$. We want to sort the n-simplices of a non-boundary complex into a sequence. The following lemma proves that the non-boundary complex is a pseudomanifold which justifies this idea.

\begin{lemma}
\label{the:internal_simplices_of_a_subdivision_forms_a_pseudomanifold}
    The non-boundary complex of any iterated chromatic subdivision of a $k$-simplex is a pseudomanifold.
\end{lemma}

\begin{proof}
Let $Ch^{n}(s^{k})$ be an iterated chromatic subdivision of a simplex $s^{k}$. We prove this by induction on the number $n$. If $n = 1$, then the non-boundary complex has only one $k$-simplex which is the central $k$-simplex of the chromatic subdivision, and therefore is a pseudomanifold.

Assume that the non-boundary complex of $Ch^{n-1}(s^{k})$ is a pseudomanifold. Then the chromatic subdivision of it will be a pseudomanifold by the property given by \cite{Spanier1966}. In other words, those simplices in the non-boundary complex of $Ch^{n}(s^{k})$ reached from some $k$-simplex in the non-boundary complex of $Ch^{n - 1}(s^{k})$ are in one equivalence class. Any remaining $k$-simplex is reached from some $k$-simplex of $Ch^{n - 1}(s^{k})$ that has some vertices on the boundary of $Ch^{n - 1}(s^{k})$.

Let $\sigma$ be a $k$-simplex of $Ch^{n - 1}(s^{k})$ that has some vertices on the boundary. We will show that the non-boundary simplices of $Ch^{n}(s^{k})$ reached from $\sigma$ will be in the same equivalence class. 
Let $V$ denote the set of vertices of $\sigma$ whose carriers are $s^{k}$. In other words, $V$ are the vertices of $\sigma$ that are not on the boundary of $Ch^{n - 1}(s^{k})$. Then there is a bijection $\rho$ between the partitions of $\Pi$ whose first process set contains some vertex in $V$ and the non-boundary simplices reached from $\sigma$.
For each set $S$ of vertices containing some vertex in $V$,  $NBS(S)$ is defined as $\{\rho(Par)\}$ where $Par$ is a partition of $\Pi$ and the first partition set is $S$. $NBS(S)$ is the joining of a $|S|$-simplex and the chromatic subdivision of a $(|\Pi| - |S|)$-simplex. Since the latter is a pseudomanifold, $NBS(S)$ is a pseudomanifold.
For each vertex $v \in V$, $N(v)$ is defined as the union of all $NBS(S)$ where $v \in S$, as shown in Figure \ref{img:internal_complex_proof}. The intersection between $NBS(\{v\})$ and $NBS(\{v\} \bigcup V^{'})$, for some set $V^{'}$ of vertices, is a simplicial complex whose dimension is $(k - 1)$ by Lemma 10.4.4 in \cite{Herlihy13}. Therefore, all $k$-simplices of $N(v)$ belong to the same equivalence class. Each $NBS(v)$, where $v \in V$, contains $NBS(V)$ and therefore belongs to the same equivalence class. The set $\{NBS(v)\}$ forms a covering of non-boundary simplices of $Ch^{n}(s^{k})$ reached from $\sigma$, which are therefore in the same equivalence class. 

\begin{figure}[h]
  \centering
  \includegraphics[width=\linewidth]{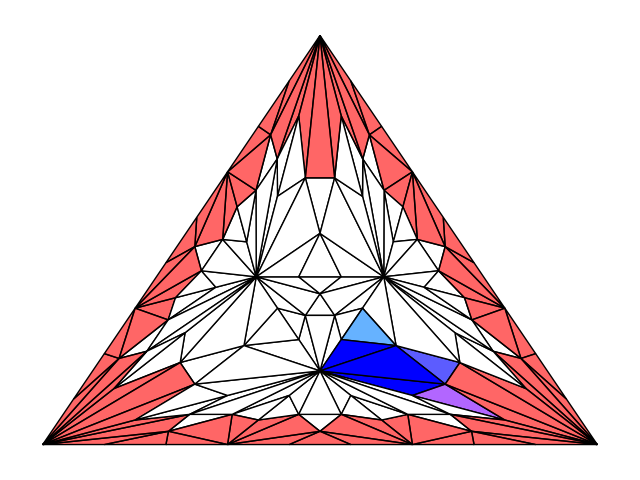}
  \caption{$N(v)$ in $Ch^{2}(s^{2})$}
  \label{img:internal_complex_proof}
\end{figure}

Then we show that the equivalence class is the same regardless of which $k$-simplex $\sigma$ of $Ch^{n-1}(s^{k})$ is chosen. All simplices in $Ch^{n - 1}(s^{k})$ form a pseudomanifold since $s^{k}$ is a pseudomanifold. If $s$ and $s^{'}$ are $k$-simplices of $Ch^{n - 1}(s^{k})$, by definition, there is a finite sequence $s = s_{1}, s_{2},\cdots, s_{m} = s^{'}$ of $k$-simplices of $Ch^{n - 1}(s^{k})$ such that $s_{i}$ and $s_{i+1}$ have an $(n - 1)$-face in common, for $1 \leq i < m$. After a subdivision of the common $(n - 1)$-face, the central simplex will be an internal $(n-1)$-face which is a common $(n-1)$-face of two non-boundary simplices reached from $s_{i}$ and $s_{i + 1}$ respectively. Therefore, the non-boundary simplices reached from $s_{i}$ and $s_{i + 1}$ are in the same equivalence class. We can get the same result for $s$ and $s^{'}$, which means that there exists one unique equivalence class for all non-boundary simplices of $Ch^{n}(s^{k})$.
\end{proof}

\subsubsection{Stage 1}
\label{sec:nccondition:finalization_after_phase_1:stage1}

In the following discussions, the adversary divides all vertices terminated with a label different from $U$ into sets in which each vertex has zero active distance to all other vertices in this set.

In the first stage, the adversary will subdivide $S^{t}$ and increment $t$ many times so that the layers added in later stages do not overlap and are far from $\chi^{t - r_{m}}(D, \delta)$, where $D$ is a subcomplex of $\mathbb{F}_{r_{m}}(U)$ consisting of $n$-simplices that do not have an intersection with $\mathbb{F}_{r_{m}}(U^{'})$. This is possible since the distance between different sets of vertices is at least 3 (by invariant 3 when two sets have different labels) or 1 (when two sets have the same label) and doubles every time the adversary subdivides $S^{t}$. There is an upper bound of the added layers, which is the number of facets in $Q_{r_{m}}(U, U^{'})$. After subdividing $S^{t}$ many times, the active distance between two sets of vertices is a large number. From now on, we will consider each set of vertices separately.

\subsubsection{Stage 2}
\label{sec:nccondition:finalization_after_phase_1:stage2}
In the second stage, the adversary will terminate the vertices adjacent to a terminated vertex $v$ with the label of $v$. The choice of label will be unique according to invariant (3). Note that we do not use the assignment rule (2) which may increase the number of vertices to terminate and therefore lead to an infinite sequence. The results assuming the usage of the assignment rule (2) will still be valid since the adversary will subdivide enough times in the first stage. The reasons to introduce this stage will be formulated in later discussions of finalization, and the discussion of assignment queries in Section \ref{sec:nccondition:augmented_extension_based_proofs}. 

\subsubsection{Stage 3}
\label{sec:nccondition:finalization_after_phase_1:stage3}

In this stage, the adversary starts to add layers of outputs to each set of vertices with a label $U^{'} \neq U$. Note that after enough subdivisions in stage 1, the adversary considers each label $U^{'} \neq U$ separately since the distance between a configuration terminated with the label $U^{'}$ and a configuration terminated with another label $U^{''} \neq U$ is large. For the $n$-simplices in $\mathbb{F}_{r_{m}}(U)$ that have an intersection with $\mathbb{F}_{r_{m}}(U^{'})$, we have divided these $n$-simplices into categories based on $mc(SHA(s^{n}, U^{'}), Q_{1}(U, U^{'}))$ in Section \ref{sec:nccondition:finalization_after_phase_1}. Each category has a simplex in $Q_{1}(U, U^{'})$ as its label.
Suppose that $Q_{1}(U, U^{'}) = \mathbb{F}_{1}(U) \cap \mathbb{F}_{1}(U^{'})$ has dimension $k$. Recall that a facet of a complex is a simplex with maximum dimension. In this section, the adversary starts with the $n$-simplices of $\mathbb{F}_{r_{m}}(U)$ in the category whose label is a facet ($k$-simplex) $s^{k}$ of $Q_{1}(U, U^{'})$ and considers the set of vertices with the label $U^{'}$ in the subdivision of these $n$-simplices. The projections of all $n$-simplices in this category form a pseudomanifold according to Lemma \ref{the:internal_simplices_of_a_subdivision_forms_a_pseudomanifold}. The $k$-simplices in the pseudomanifold can be organized into a sequence $s_{0}, s_{1}, ... s_{e}$ which ends with the $k$-simplex $CEN(s^{k})$ where $\delta_{U}(CEN(s^{k})) = \delta_{U^{'}}(CEN(s^{k}))$. If a set of vertices is in the subdivision of only one $n$-simplex $s^{n}$ of $\mathbb{F}_{r_{m}}(U)$, the adversary adds a layer of output $\delta_{U}(Prj(s^{n}, U^{'}))$ and then adds layers of outputs $\delta_{U^{'}}(s_{i + 1}), \delta_{U^{'}}(s_{i + 2}) ... \delta_{U^{'}}(s_{e})$ to this set, where $Prj(s^{n}, U^{'})$ is $i$-th element in the sequence. Then in reverse order, the adversary adds layers of outputs $\delta_{U}(s_{e}), \delta_{U}(s_{e - 1}) ... \delta_{U}(s_{i})$. The main difficulty of this section lies in the set of vertices that are in the subdivision of multiple $n$-simplices (which do not necessarily belong to one category).

The adversary then considers the set of vertices terminated with the label $U^{'}$ in the subdivision of $n$-simplices in a category whose label is a facet($(k-1)$-simplex) of the $(k-1)$-skeleton of $Q_{1}(U, U^{'})$. This ends when the category whose label is the $0$-skeleton of $Q_{1}(U, U^{'})$ has been considered.

The first step is to consider the category whose label is a facet($k$-simplex) $s^{k}$ of $Q_{1}(U, U^{'})$. By Lemma \ref{the:internal_simplices_of_a_subdivision_forms_a_pseudomanifold}, we know that the $k$-simplices, which are projections of $n$-simplices in this category, form a pseudomanifold. It is possible to enumerate these $k$-simplices in a finite sequence $s_{0}, s_{1},....,s_{i}, .... s_{e}$ and to find a sequence $s_{1}^{'},....,s_{i}^{'}, ....$ of $(k - 1)$-simplices such that for $i \geq 1$, $s_{i}^{'}$ is a face of $s_{i - 1}$ and also a face of $s_{i}$. Note that we only require every $k$-simplex to appear in this sequence, as a consequence of which a $k$-simplex may appear multiple times. The $k$-simplex $CEN(s^{k})$ whose output is assumed to be compatible when using $\delta_{U}$ and $\delta_{U^{'}}$ is chosen to be the last $k$-simplex $s_{e}$ of the sequence.

The adversary will maintain a finalization invariant after processing the $n$-simplices whose projections are in $s_{i}$: for each set of vertices terminated with the label $U^{'}$ in the subdivision of the $n$-simplices whose projections are in $s_{0} \cup s_{1} \cup s_{2}.... \cup s_{i}$, the outermost layer of the set is assigned with the output of $\delta_{U^{'}}(s_{i})$. 

Consider the $n$-simplices $s^{n}$ of $\mathbb{F}_{r_{m}}(U)$ whose projections $Prj(s^{n}, U^{'})$ are in $s_{0}$. The adversary adds a layer of $\delta_{U^{'}}(Prj(s^{n}, U^{'}))$ to each set terminated with the label $U^{'}$. In fact, this is the first step for any terminated set with the label $U^{'}$, which will be repeated many times in the following discussion. The dimension of $\delta_{U^{'}}(Prj(s^{n}, U^{'}))$ is at most $k$. Note that the $n$-simplices $s^{n}$ whose projections $Prj(s^{n}, U^{'})$ are a proper subsimplex of $s_{0}$ are included here, in the subdivision of which the adversary simply adds a layer of output $\delta_{U^{'}}(Prj(s^{n}, U^{'}))$ whose dimension is less than $k$. This added layer does not violate the task specification, as $carrier(Prj(s^{n}, U^{'}), Ch(\mathcal{I}))$ is a subsimplex of $carrier(v, Ch(\mathcal{I}))$ for each vertex $v$ of $s^{n}$. In other words, in the configuration represented by $s^{n}$ of $\mathbb{F}_{r_{m}}(U)$, if a process whose state is represented by a vertex in $Prj(s^{n}, U^{'})$ has seen the input value of some process $p$, then a process whose state is represented by a vertex in $s^{n} - Prj(s^{n}, U^{'})$ has seen the input value of $p$.

The adversary then considers the $n$-simplices whose projections are in $s_{1}$. The sets terminated with the label $U^{'}$ in the subdivision of $n$-simplices whose projections are in $s_{0} \cup s_{1}$ can be classified into three types: sets in $n$-simplices whose projections are in $s_{0}$, sets in $n$-simplices whose projections are in $s_{1}$, and sets in both types of $n$-simplices.

First, the adversary processes the sets in the subdivision of $n$-simplices whose projections are in $s_{0}$. The outermost layer of a set $S$ terminated with the label $U^{'}$ is assigned with the output $\delta_{U^{'}}(s_{0})$ by the finalization invariant. The adversary adds a layer of output $\delta_{U^{'}}(s_{1})$ to each set $S$. This is possible since $\delta_{U^{'}}(s_{0})$ is different from $\delta_{U^{'}}(s_{1})$ by the output value of a single process, denoted by $p$. Suppose that before the operation, $\sigma$ is a $n$-simplex with some vertices in this set $S$. If some vertex of $\sigma$ is assigned with the output value of the process $p$, then the adversary terminates the remaining vertices of $s^{n}$ with the output values of $\delta_{U^{'}}(s_{1}^{'})$. Otherwise, if $\sigma$ does not have any vertex terminated with the output value of $p$, then the adversary terminates the remaining vertices of $s^{n}$ with the output values of $\delta_{U^{'}}(s_{1})$. The outermost layer is assigned with the output values of $\delta_{U^{'}}(s_{1})$. Note that although we are assuming that the task is colorless, we try to minimize the usage of this condition. If a task is colorless, the adversary can add a layer of $s_{1}$ when $s_{0}$ and $s_{1}$ share only one vertex. We require that $s^{0}$ and $s^{1}$ share the $(k - 1)$ vertices since the same argument can be generalized to a task without the colorless condition.

Then we consider the sets in the $n$-simplices whose projections are in $s_{1}$. The adversary adds a layer of $\delta_{U^{'}}(s_{1})$ in the $n$-simplices whose projections are in $s_{1}$

The remaining sets are those sets in the subdivision of the $n$-simplices whose projections are in $s_{0}$ and the $n$-simplices whose projections are in $s_{1}$ at the same time. Again, the adversary adds a layer of $\delta_{U^{'}}(s_{1})$ in the $n$-simplices whose projections are in $s_{1}$. 
The adversary then adds a layer of $\delta_{U^{'}}(s_{1})$ in the $n$-simplices whose projections are in $s_{0} \cup s_{1}$. The outermost layer is assigned with the output values of $\delta_{U^{'}}(s_{1})$. The finalization invariant is maintained.

Suppose that the adversary has added layers to the set terminated with the label $U^{'}$ in the subdivision of the $n$-simplices whose projections are in $s_{0} \cup s_{1} \cup .... \cup s_{i - 1}$ and the finalization invariant is true. The adversary considers the sets in the subdivision of the $n$-simplices whose projections are in $s_{i}$, which are divided into three types as above.

First, the adversary looks back to the sets in the subdivision of the $n$-simplices whose projections are in $s_{0} \cup s_{1} \cup .... \cup s_{i - 1}$. Such a set can be in the subdivision of multiple $n$-simplices. The adversary adds a layer of $\delta_{U^{'}}(s_{i})$ in the $n$-simplices whose projections are in $s_{0} \cup s_{1} \cup .... \cup s_{i - 1}$, which is possible since the $\delta_{U^{'}}(s_{i - 1})$ is different from that of $\delta_{U^{'}}(s_{i})$ by the output value of only one vertex. Details have already been given in the previous discussion. 

The adversary then considers the sets totally in the subdivision of the $n$-simplices whose projections are in $s_{i}$. The adversary adds a layer of $\delta_{U^{'}}(s_{i})$ in the $n$-simplices whose projections are in $s_{i}$. In other words, the adversary will terminate the vertices adjacent to a terminated vertex with some value of $\delta_{U^{'}}(s_{i})$.

The main obstacle is to add layers to the sets at the intersection of the $n$-simplices whose projections are in $s_{i}$ and those whose projections are in $s_{0} \cup s_{1} \cup .... \cup s_{i - 1}$. Now if $s_{i}$ has appeared in the sequence before it (recall that we do not require that a $k$-simplex appear only once in the sequence), then the adversary will treat these sets as sets whose projections are in $s_{0} \cup s_{1} \cup .... \cup s_{i - 1}$. The adversary adds a layer of output value $\delta_{U^{'}}(s_{i})$.

If $s_{i}$ has not appeared in the sequence before it, then the adversary has added some layers to this set in the $n$-simplices whose projections are in $s_{0} \cup s_{1} \cup .... \cup s_{i - 1}$. The adversary needs to complete these layers in the $n$-simplices whose projections are in $s_{i}$. Consider the boundary between $n$-simplices whose projections are in $s_{0} \cup s_{1} \cup .... \cup s_{i - 1}$ and $n$-simplices whose projections are in $s_{i}$. In the general situation, the projection of this boundary is not restricted to $s_{i}^{'}$ as in the case of $s_{1}$ and this is the reason why the general discussion will be much more complicated. By the finalization invariant, for each terminated set with the label $U^{'}$ in the subdivision of $n$-simplices whose projections are in $s_{0} \cup s_{1} \cup .... \cup s_{i - 1}$, the outermost layer of it has the output values of $\delta_{U^{'}}(s_{i - 1})$.

If a set $S_{1}$ in the $n$-simplices whose projections are in $s_{i}$ shares some vertices with some previous set $S_{2}$, then these sets should be merged into a single set. The adversary has added some layers to $S_{2}$ in the $n$-simplices whose projections are in $s_{0} \cup s_{1} \cup .... \cup s_{i - 1}$. Now the adversary has to add some layers to $S_{1}$ in the $n$-simplices whose projections are in $s_{i}$. Suppose that after the adversary finishes stage 2, a vertex $v_{o}$ is on the boundary of an $n$-simplex $s^{n}$ whose projection is in $s_{i}$ and an $n$-simplex whose projection is in $s_{j}$ where $j < i$. Now the adversary has to terminate the vertices in the subdivision of $s^{n}$ reached from $v_{o}$. Will this operation cause any conflict of output values in these $n$-simplices whose projections are in $s_{i}$? Will the operation terminate more vertices on the boundary so that the adversary has to look back to $n$-simplices whose projections are in $s_{j}$ to fix the same problem? Before we present what the adversary will do, we have to analyze a property of the sets terminated with the label $U^{'}$. 

After the first stage, in which enough chromatic subdivisions are made, the sets in the subdivision of an $n$-simplex either have a large active distance to the boundary of the $n$-simplex or have a zero active distance. Sets of the former type have been discussed and will not affect the merging of sets. The sets near the boundary belong to the latter type, which means that after the second stage the set will contain some terminated vertices on the boundary, and the sets in the subdivision of different $n$-simplices should be merged into a single one. 

It is quite natural to think that the adversary can add layers to the merged set using the values of the existing added layers on the boundary. But this method does not work for our current adversarial strategy. Suppose that the dimension of $s_{i}$ is $k$. This idea will not work for sets containing vertices whose projections are in different $(k - 1)$-dimensional subsimplices of $s_{i}$. In other words, a set can be connected to many existing terminated sets, each of which is in the $n$-simplices whose projections are in $s_{j_{d}}$ where $j_{d} < i$. For example, this happens when each vertex in the subdivision of a facet of $Q_{r_{m}}(U, U^{'})$ is terminated with the label $U^{'}$, which is possible since $Q_{t}(U, U^{'})$ is a subcomplex $\mathbb{F}_{t}(U^{'})$ for each $t$, and is therefore not restricted (as the assignment rule (2)). In this situation, there will be multiple sequences of existing added layers on the boundary that cause a conflict of output values to be used. 

The solution to this disturbing problem is that each set $S_{1}$ in the $n$-simplices whose projections are in $s_{i}$ can be connected to at most one previous set $S_{2}$. To achieve this goal, we have to modify our adversarial strategy described in Section \ref{sec:nccondition:adversarial_strategy}. We place a restriction on the first assignment rule, just as we have done with the second assignment rule. Once the adversary terminates a vertex $v$ with the label $U$ where $v$ is in $\mathbb{F}_{t}(U)$, it can do it again after subdividing the current complex $r_{s}$ times. This restriction can only cause a liveness issue. However, as the proof of Lemma \ref{the:restrict_ebf_finite_chain_of_queries_lemma} shows, the restriction on the second assignment rule will not cause an infinite chain of queries. Due to the same reason, the restriction on the first assignment rule will not cause an infinite chain of queries, i.e., lead to a liveness problem. Perhaps we have to explain to the readers why this rule is not introduced in Section \ref{sec:nccondition:adversarial_strategy}. We impose a restriction on the second assignment rule, since the canonical neighbor with the label $U^{'}$ can only be defined for the $n$-simplices that have an intersection with $\mathbb{F}_{r_{m}}(U^{'})$. But this is not the case for those vertices terminated with $\delta_{U}$ when they are in $\mathbb{F}_{t}(U)$. Here, we introduce the same restriction for a completely different reason.

Let $S_{2}^{'}$ be a set of vertices after the second stage generated from $S_{2}$ by removing the added layers in stage 3. Let $S_{3}$ be a set of vertices that is shared by $S_{1}$ and $S_{2}^{'}$. The adversary will use the existing layers added to $S_{3}$ to terminate some vertices reached from $S_{1}$.

Given a set of vertices $S$, let $\mathcal{K}$ be the minimum subcomplex of $\mathbb{F}_{t}(U)$ whose subdivision contains $S$.
The projection of $S$ to $Q_{r_{m}}(U, U^{'})$ is defined as the subcomplex obtained from $\mathcal{K}$ by removing the vertices whose process ids are not in $ids(Q_{r_{m}}(U, U^{'}))$. 

If the projection of $S_{3}$ to $Q_{r_{m}}(U, U^{'})$ is within some $l$-dimensional ($l \leq k - 1$) subsimplex $s^{l}$ of $s_{i}$, the adversary will use the existing layers added to $S_{3}$. First, the adversary adds a layer of output values of $\delta_{U^{'}}(s^{l})$ (note that the first layer added here is not $\delta_{U^{'}}(s_{i})$ as in the above cases) to $S_{1}$ in the $n$-simplices whose projections are in $s_{i}$. The adversary can then add layers compatible with the existing layers on the boundary: for each vertex $v$ in $S_{1}$, the adversary terminates the vertices reached from $v$ by an undefined path whose length equals the number of existing layers. The concrete values are obtained from the existing layers added to $S_{3}$ at the boundary. In a high-level idea, we say that the adversary designates a previous $k$-simplex $s_{i^{'}}$, where $i^{'} < i$, for a set $S_{1}$ in the subdivision of an $n$-simplex $s_{i}$. This means that although $S_{1}$ is a set of vertices in the subdivision of $s^{n}$, the adversary sees it as a set of vertices in the subdivision of another $n$-simplex whose projection to $Q_{r_{m}}(U, U^{'})$ is $s_{i^{'}}$. The validity of these assigned values is shown in the following lemma.

In our adversarial strategy, any vertex $v_{1}$ in the subdivision of an $n$-simplex $s^{n}$ terminated with the label $U^{'}$ is connected to some vertex $v_{2}$ in $\mathbb{F}_{t}(U^{'})$(if there are multiple choices of $v_{2}$, choose the first $v_{2}$ in $\mathbb{F}_{t}(U^{'})$). The $n$-simplex $s^{n}$ can be seen as the joining of a simplex $s_{1} \in \mathbb{F}_{r_{m}}(U^{'})$ and a non-boundary simplex $s_{2} \notin \mathbb{F}_{r_{m}}(U^{'})$. 

\begin{lemma}
\label{the:relations_of_carrier}
$carrier(v_{1}, s^{n})$ is the joining of $carrier(v_{2}, s_{1})$ and some subsimplex of $s_{2}$, or a simplex containing this joining.
\end{lemma}
\begin{proof}
We will prove this result by contradiction. Suppose that the carrier of $v_{1}$ in $s^{n}$ is the joining of any strict subsimplex of $carrier(v_{2}, s_{1})$ and $s_{2}$, denoted by $s_{3}$. There exists a path from $v_{1}$ to $v_{2}$, where $v_{1}$ is in the subdivision of $s_{3}$ and $v_{2}$ is not in the subdivision of $s_{3}$. In the complex $S^{r_{m} + r_{a}}$, $\delta(v)$ is defined as $\perp$ for each vertex $v$. Let $v$ be a vertex in the subdivision of $carrier(v_{2}, s_{1})$ but not in the subdivision of any subsimplex of $carrier(v_{2}, s_{1})$. The length of an undefined path from $v$ to a vertex in the subdivision of $s_{3}$ is a large number or 1. In the latter case, the path goes through the strict subsimplex of $carrier(v_{2}, s_{1})$ which contradicts the definition of $v_{2}$. By Lemma \ref{the:glue_protocols:prover_fails:terminated_vertices_not_reach}, any undefined path in $S^{t}$, where $t > r_{m} + r_{a}$, from a vertex in the subdivision of $carrier(v_{2}, s_{1})$ but not a subsimplex of $carrier(v_{2}, s_{1})$ to a vertex in the subdivision of $s_{3}$ cannot be terminated with the label $U^{'}$, which contradicts the fact that $v_{2}$ is connected to $v_{1}$ by a path consisting of terminated vertices.
\end{proof}

This lemma shows that if an output value is valid for a vertex $v_{2}$ in $\mathbb{F}_{t}(U^{'})$ terminated with the label $U^{'}$, then any vertex $v_{1}$ connected to $v_{2}$ can be terminated with the same output value. The output values of the existing added layers to $S_{3}$ are obtained from $\delta_{U^{'}}(s_{0}), \delta_{U^{'}}(s_{1}) \cdots \delta_{U^{'}}(s_{e})$ where $s_{i} \in \mathbb{F}_{r_{m}}(U^{'})$. In the subdivision of all $n$-simplices in the category with label $s^{k}$, the vertices with the minimum carrier in $Ch(\mathcal{I})$ contain vertices of $s_{0}, s_{1} ... s_{e}$. For each vertex $v$ in $S_{1}$, $carrier(s_{i}, Ch(\mathcal{I}))$ is a subsimplex of $carrier(v, Ch(\mathcal{I}))$. Therefore, terminating a vertex adjacent to $S_{1}$ using the output values from the existing added layers to $S_{3}$ does not violate the task specification.

Although we do not use the second assignment rule to add layers to the terminated sets(i.e. we do not subdivide the current complex each time we terminate a vertex with the label $U^{'}$), the number of added layers has an upper bound, which means that if enough subdivisions are made in the first stage, the added layers will also satisfy the above lemma. 
Another requirement for validity is that there is no conflict of newly assigned values. Layers added to a vertex in $S_{1}$ will not terminate a new vertex on the boundary. Suppose that after the second stage, $v_{0}$ is a terminated vertex on the boundary and $v_{1}$ is a terminated vertex not on the boundary. The two vertices are both undefined vertices after the first stage. Let $\sigma$ denote an undefined $n$-simplex containing $v_{0}$ and $v_{1}$ before stage 1. Note that at this time, $v_{1}$ may not exist in $S^{t}$. The simplex $\sigma$ will be subdivided sufficiently using the standard chromatic subdivision so that the active distance from $v_{1}$ to the boundary will be 1 or a large number. So, after the second stage where $v_{0}$ and $v_{1}$ are terminated, any path consisting of undefined vertices from $v_{1}$ to a vertex at the boundary will have a long length. We give a simple example to present our idea in Figure \ref{img:added_layers_no_reached}. When the adversary merges the sets on the boundary, the added layers to $v_{1}$ will not reach the boundary, which means that the adversary will not have to change the output values of the terminated vertices on the boundary or terminate some new vertices on the boundary. Therefore, the adversary will not alternate infinitely between different sides of the boundary, and it is valid for us to choose the form of induction to discuss sets terminated with the label $U^{'}$.

\begin{figure}[h]
  \centering
  \includegraphics[width=\linewidth]{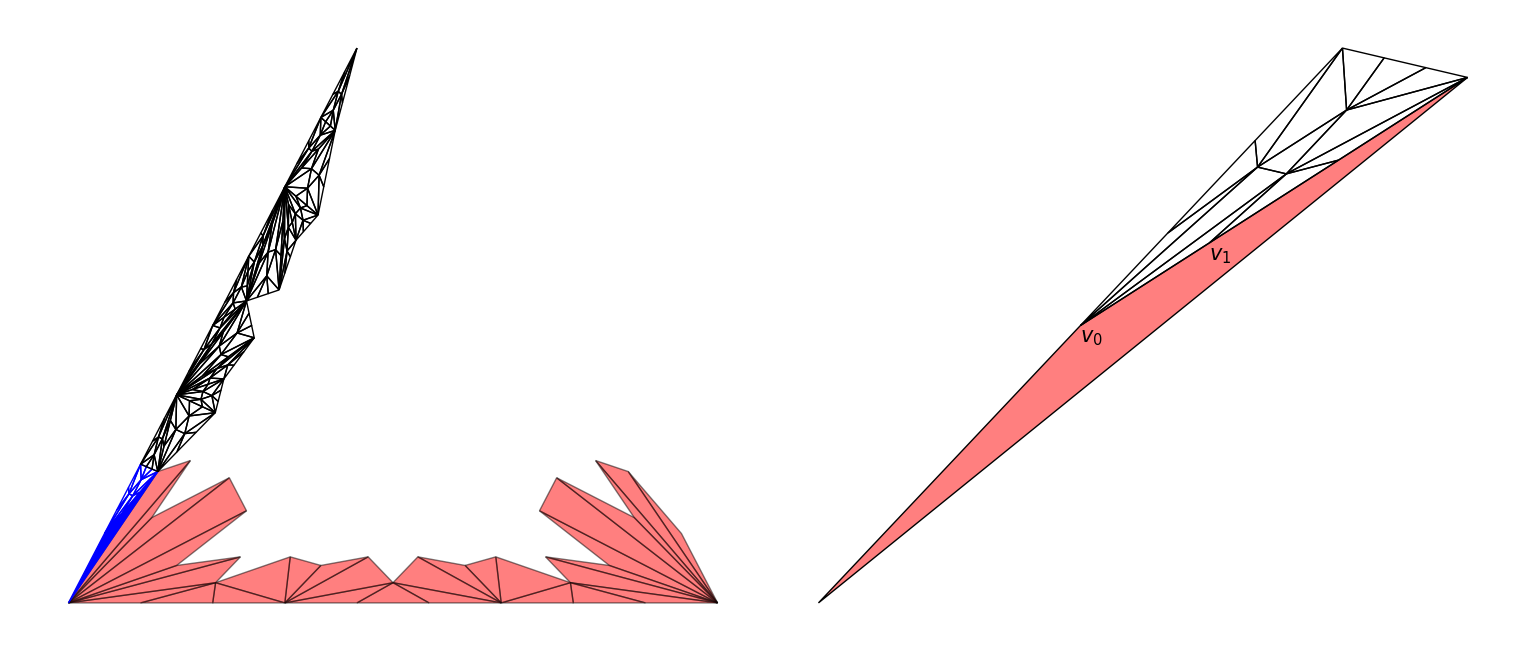}
  \caption{$v_{1}$ will have a large distance to the boundary}
  \label{img:added_layers_no_reached}
\end{figure}

A more complicated situation is that the projection of $S_{3}$ to $Q_{r_{m}}(U, U^{'})$ will be in multiple subsimplices of $s_{i}$, denoted by a complex $L$. The added layers to $S_{3}$ is not unique. For each simplex $s_{l}$ of $L$, there is a sequence of layers on the boundary added to the vertices whose projections are $s_{l}$. 

We first prove that each maximal simplex in $L$ contains some simplex $s^{j}$.

\begin{lemma}
\label{the:projection_is_an_open_star}
Let $s^{n}$ be an $n$-simplex in $\mathbb{F}_{r_{m}}(U)$, which is the joining of a subsimplex $s_{1} \in \mathbb{F}_{r_{m}}(U^{'})$ and a subsimplex $s_{2} \notin \mathbb{F}_{r_{m}}(U^{'})$. For each set terminated with the label $U^{'}$ after the second stage in the subdivision of $s^{n}$, there exists a terminated vertex whose projection to $Q_{r_{m}}(U, U^{'})$ is $s^{j}$ such that the projection to $Q_{r_{m}}(U, U^{'})$  of each terminated vertex in this set contains $s^{j}$.
\end{lemma}

\begin{proof}
Suppose that the projections of two terminated vertices $v_{1}$ and $v_{1}^{'}$ to $Q_{r_{m}}(U, U^{'})$ share a common simplex $s_{s}$ in $s_{1}$, and no terminated vertex has a projection $s_{s}$ to $Q_{r_{m}}(U, U^{'})$. By Lemma \ref{the:relations_of_carrier}, there exist $v_{2}$ and $v_{2}^{'}$ in the subdivision of $s_{1}$ that are connected to $v_{1}$ and $v_{1}^{'}$ respectively. 

By our new restriction on the first assignment rule, if the active distance between two subgraphs is at least 2 in $S^{t_{1}}$, then there is no terminated path from the subdivision of one subgraph to the subdivision of the other subgraph in $S^{t_{2}}$, where $t_{2} > t_{1}$. We can prove this using an argument similar to the proof of \ref{the:glue_protocols:prover_fails:terminated_vertices_not_reach}.

The two vertices $v_{2}$ and $v_{2}^{'}$ are connected by a terminated path with the label $U^{'}$, which is possible only when each vertex has zero active distance to a vertex whose projections to $Q_{r_{m}}(U, U^{'})$ is $s^{j}$. By our adversarial strategy, some vertices in the subdivision of $s_{s}$ will terminate in the second stage, which is a contradiction.
\end{proof}

In this part, we associate each vertex of $S_{1}$ with a simplex of $L$. For each vertex $v_{1}$ in $S_{1}$, there is a vertex $v_{2}$ in $\mathbb{F}_{r_{m}}(U^{'})$ that is connected to $v_{1}$. $v_{2}$ is connected to one simplex, denoted by $s_{v}$, in $L$ as each $v_{1}$ is connected to some vertex in $S_{3} \subset S_{1}$ ($v_{2}$ is the projection of $v_{1}$ and $L$ is the projection of $S_{3}$).

Suppose that the adversary designates a previous $k$-simplex $s_{i^{'}}$ for the set $S_{2}$ of vertices, where $i^{'} < i$. Similarly, the adversary designates $s_{i^{'}}$ for the set of vertices $S_{1}$. 
At first, the adversary terminates the vertices adjacent to $v_{1}$ in $S_{1}$ with the output $\delta_{U^{'}}(s_{v})$. The adversary then uses the sequence of layers added to $s_{v}$ to terminate the vertices adjacent to $v_{1}$. The adversary can terminate the vertices reached from $v_{1}$ by an undefined path whose length equals the number of existing layers added to $s_{v}$. Consider two vertices $v_{1}$ and $v_{1}^{'}$ such that $s_{v}$ and $s_{v}^{'}$ have a different carrier in $L$. Since there will be some simplex $s^{j}$ shared by $s_{v}$ and $s_{v}^{'}$ in $L$ by Lemma \ref{the:projection_is_an_open_star}, the layers added to the vertices whose projections are $s^{j}$ will serve as the intermediate part. Therefore, the newly terminated vertices will not have any conflict of output values, since the existing added layers on the boundary do not have a conflict of values. The output values of the outermost layer are those of $\delta_{U^{'}}(s_{i - 1})$. Finally, the adversary will add a layer of $\delta_{U^{'}}(s_{i})$. All the added layers mentioned above will not cause any violation of the task specification, according to Lemma \ref{the:relations_of_carrier}.

By induction, the adversary will process the last $k$-simplex in the sequence, i.e. the $k$-simplex that is assigned with the same output values by $\delta_{U}$ and $\delta_{U^{'}}$. By the finalization invariant, the output values of the outermost layer only contain these shared output values. From now on, these shared output values will be seen as output values with the label $U$ rather than $U^{'}$. 

The adversary will add more layers to connect the shared output values to $\delta_{U}(s^{n})$, dual to what it has done. This time, the adversary will reverse the sequence of $k$-simplices as $s_{e}, s_{e - 1},....,s_{i}, .... s_{0}$. Note that we do not keep an invariant (like the finalization invariant) here. When processing the $n$-simplices whose projections are in $s_{i}$, the adversary will add a layer of output $\delta_{U}(s_{i})$ to all sets in the $n$-simplices whose projections are in $s_{i} \cup s_{i - 1} ... \cup s_{0}$. Then the outermost layer of the set will be the output $\delta_{U}(s_{i})$, which is exactly what we want.

If a set is in the subdivision of $n$-simplices whose projections are both in $s_{e} \cup s_{e - 1},.... \cup s_{i + 1}$ and $s_{i} \cup s_{i - 1} ... \cup s_{0}$, the adversary only adds a layer in the $n$-simplices whose projections are in $s_{i} \cup s_{i - 1} ... \cup s_{0}$ but not in $(s_{e} \cup s_{e - 1},.... \cup s_{i + 1}) \cap (s_{i} \cup s_{i - 1} ... \cup s_{0})$.
Consider a set $S$ of vertices whose projections are in $s_{i - i^{'}}$. Let $S_{1}$ be the set of vertices in $S$ in $n$-simplices whose projections are in $s_{i - i^{'}}$ but not in $s_{i - i^{'}} \cap (s_{e} \cup s_{e - 1},.... \cup s_{i + 1})$, $S_{2}$ be the set of vertices in $S$ in $n$-simplices whose projections are in $s_{i - i^{'}} \cap (s_{e} \cup s_{e - 1},.... \cup s_{i + 1})$. The adversary adds a layer to $S_{1}$, but not to $S_{2}$. Later, when the adversary finishes the induction on $s_{i - i^{'}}$, the outermost layer of $S_{1}$ and $S_{2}$ consists of the output $\delta_{U}(s_{i - i^{'}})$. 

The adversary now processes the categories whose labels are facets of the $(k - 1)$-skeleton of $s^{k}$, denoted by $s^{k - 1}$. Each of these categories has a shared boundary with the category whose label is $s^{k}$. Before the adversary repeats what it has done to the category whose label is $s^{k}$, it has to deal with those vertices on the boundary with the $n$-simplices in the category whose label is $s^{k}$. Note that the boundary between different categories is not simply a simplex like the boundary of their projections in $s^{k}$, but contains each simplex whose projection is within the boundary of their projections.  Stage 1 and stage 2 are the same as in the previous category(including how many times the chromatic subdivisions are made). In fact, what we have done is to bring stage 1 and stage 2 of all categories forward to merge into a single stage 1 and stage 2.

Since the adversary has processed the category whose label is $s^{k}$, the subdivision of some $n$-simplices in the category whose label is $s^{k - 1}$ contains some added layers on the boundary. The adversary will designate some $k$-simplex in the category whose label is $s^{k}$ for such a set of vertices (just as the adversary does in the previous discussions). 
The adversary will add layers compatible with existing layers by assigning values to those undefined vertices reached from some terminated vertex after the second stage by an undefined path whose length equals the number of needed layers, similar to what the adversary has done before. A key observation is that for each terminated set in the subdivisions of $n$-simplices in the category whose label is $s^{k - 1}$, if it has some vertices $v_{1}$ on the boundary with the category whose label is $s^{k}$, then $carrier(v_{2}, Ch(\mathcal{I}))$ is $s^{k}$ rather than any subsimplex of $s^{k}$ where $v_{2}$ is a vertex in $\mathbb{F}_{r}(U^{'})$ connected to $v_{1}$. By Lemma \ref{the:relations_of_carrier}, the added layers will not violate the carrier map.

Now, these sets will be omitted, and the adversary uses the same techniques to the category with label $s^{k - 1}$. As a consequence, the outermost layer of the sets with the label $U^{'}$ consists of the output of $\delta_{U}$. A difference we must explicitly mention is that the adversary projects the $n$-simplices in $\mathbb{F}_{r_{m}}(U)$ to the label of the category, that is, $s^{k - 1}$ rather than $s^{k}$.

A general induction can be given. Suppose that the categories whose labels are facets of the $k - q + 1$-skeleton of $s^{k}$ for some $0 < q \leq k$ have been processed and the outermost layer of the sets terminated with the label $U^{'}$ in the subdivision of $n$-simplices of those categories will be assigned by the partial protocol $\delta_{U}$.
Let $s^{k - q} \in s^{k}$ be the label of some category. Then it has some shared boundary with some previous categories. For a set terminated with the label $U^{'}$ whose vertices are in the category whose label is $s^{k - q}$ and some previous categories, the adversary designates an $(k - k^{'})$-simplex in the category whose label has the highest dimension $(k - k^{'})$ for this set. After this, this set of vertices is omitted. The adversary then uses the techniques for the category whose label is $s^{k - q}$. No task specification will be violated, as we have shown for the category whose label is $s^{k-1}$. Until now, we have processed categories, each of which has a label of a subsimplex of a $k$-facet.


Therefore, after processing each possible category, the outermost layer of a set terminated with the label $U^{'}$ in the subdivision of an $n$-simplex $s^{n}$ consists of the output $\delta_{U}(s^{n})$.

\subsubsection{Stage 4}
\label{sec:nccondition:finalization_after_phase_1:stage4}

Finally, the adversary arrives at stage 4 where it fills the gap between a set terminated with the label $U$ and a set terminated with the label $U^{'}$. The adversary sets the output $\delta_{U}(s^{n})$ for each undefined vertex $v$ in $\chi^{t - r_{m}}(s^{n}, \delta)$ for each $n$-simplex $s^{n}$ in $\mathbb{F}_{r_{m}}(U)$. Each vertex in $\mathbb{F}_{t}(U)$ has been defined with an output value, which means that the adversary has constructed a partial protocol with respect to $U$. In phases $\varphi \geq 2$, the adversary will respond to queries according to this partial protocol and the prover cannot find any safety or liveness problem. Eventually, the prover will choose a configuration in which every process has terminated to end some phase and lose in the next phase.

We summarize the above results into a theorem. 
\begin{theorem}
\label{the:adversary_finalize_after_the_first_round_theorem}
For a colorless task$(\mathcal{I}, \mathcal{O}, \Delta)$, there exists an adversary that can finalize after the first round to win against any restricted extension-based prover if and only if there exists a partial protocol with respect to each simplex $U \in \mathcal{I}$ and all these partial protocols are compatible.
\end{theorem}

\subsection{Augmented extension-based proofs}
\label{sec:nccondition:augmented_extension_based_proofs}

Now, we consider augmented extension-based proofs that allow assignment queries $(C, P, f)$. The prover has to reply whether there exists a $P$-only schedule from the configuration $C$ such that the output values of $ Q \subseteq P $ satisfy the requirement of $f$.

In our current adversarial strategy, only the outputs $\delta_{U^{'}}(N(s^{n}, U^{'}))$ and $\delta_{U}(s^{n})$ are assigned during phase 1. The adversary now has to assign an output not equal to $\delta_{U}(s^{n})$ for any label $U$ to a simplex in the subdivision of $s^{n}$ for each $n$-simplex in $Ch^{r_{m}}(\mathcal{I})$ {\em during} phase 1. More specifically, the output values $((p_{0}, v_{0}), (p_{1}, v_{1}) \cdots (p_{m}, v_{m}))$ of $(m + 1)$ processes that the adversary uses after the end of phase 1 to fill the gap between $\delta_{U^{'}}(N(s^{n}, U^{'}))$ and $\delta_{U}(s^{n})$ must be assigned to some vertices in the subdivision of $s^{n}$ in some situation. This can be proved by having the prover submit an assignment query $(s^{n}, (p_{0}, p_{1} \cdots p_{m}), f: p_{i} \rightarrow v_{i})$. The main difficulty of this section is to integrate these output values into our adversarial strategy given in Section \ref{sec:nccondition:adversarial_strategy}, which assumes that the output value of a vertex in the subdivision of an $n$-simplex in $Ch^{r_{m}}(\mathcal{I})$ is determined by its label and its process id. We either have to maintain this assumption by defining some new labels for these new output values or modify this assumption. In this section, we adopt the latter method without redesigning most of our adversarial strategy, which means that assignment queries do not give the prover more power.

Let $s^{n}$ be an $n$-simplex in $\mathbb{F}_{r_{m}}(U)$. Let $s^{l}$ be any $l$-simplex in the subdivision of it. We define possible output values $POV(s^{l}, U^{'})$ with the label $U^{'}$ of $s^{l}$ as all possible output values of some $l$-simplex reached from $s^{l}$ terminated with the label $U^{'}$. By Lemma \ref{the:glue_protocols:prover_fails:terminated_vertices_not_reach}, only the simplex $s^{l}$ that has zero active distance to some vertex terminated with the label $U^{'}$ or some undefined vertex in $\mathbb{F}_{t}(U^{'})$ can have possible output values of the label $U^{'}$. For example, if the projection of $s^{l}$ to $Q_{t}(U, U^{'})$ is only in one category whose label is a $k$-simplex $s_{1}^{k}$, then $POV(s^{l}, U^{'})$ are the limitations to $ids(s^{l})$ of all $\delta_{U}(s_{2}^{k})$ and $\delta_{U^{'}}(s_{2}^{k})$ for all $s_{2}^{k}$ in $Ch\_int^{r_{m} - 1}(s_{1}^{k})$. In this section, we define the limitation of some output $O = \{(p_{0}, o_{0}), (p_{0}, o_{0}) \cdots (p_{i}, o_{i})\}$ to a set $P$ of processes as an assignment to each process $p \in P$ with an output value. If $p = p_{j}$ for some $0 \leq j \leq i$, then the process $p$ is assigned with the output $o_{j}$. Otherwise, the process $p$ is assigned an output $o_{j}$ for an arbitray $j$. By definition, if the adversary terminates $s^{l}$ with one of its possible output values, there is no violation of the task specification. Note that if the projections of $s^{l}$ to $Q_{r_{m}}(U, U^{'})$ are in multiple categories, the possible output values of $s^{l}$ are obtained from the category whose label has the highest dimension $(k - k^{'})$. This is because the adversary can designate a $(k - k^{'})$-simplex in this category for a set of vertices that contain $s^{l}$. An important property is that the space formed by all possible output values of $s^{l}$ is a connected space. In other words, each possible output value with the label $U^{'}$ can be connected to the output values of $(s^{l})$ given by $\delta_{U^{'}}$ by a path consisting of possible output values.

Suppose that the invariants hold before an assignment query $(C, P, f)$ in phase 1 and $C$ is some configuration reached from some simplex in $\mathbb{F}_{1}(U)$ for some simplex $U \in \mathcal{I}$. By invariant (1), we know that the states of processes in $P$ will correspond to a simplex in $S^{r}$, where $r \leq t$. If the round number $r$ is less than $r_{m}$, the result of a single assignment query $(C, P, f)$ can be generated from the results of a set of assignment queries $(C^{'}, P, f)$ where $C^{'}$ is a $k$-simplex reached from $C$ in $Ch^{r_{m}}(\mathcal{I})$. Therefore, we will only discuss the assignment queries where $r \geq r_{m}$ and the configuration $C$ that is reached from some $n$-simplex $s^{n}$ in $\mathbb{F}_{r_{m}}(U)$. Let $R^{t}$ denote the subcomplex of $S^{t}$ consisting of all simplices $s^{k}$ reached from this simplex by a $P$-only $(t - r)$-round schedule. We will discuss the $k$-simplices in $R^{t}$. To keep the notation simple, we define the canonical neighbor with the label $U$ of an $n$-simplex $s^{n}$ in $\mathbb{F}_{r_{m}}(U)$ as the $n$-simplex $s^{n}$ itself. 

Analysis is relatively straightforward when some vertices of a $k$-simplex $\tau$ in $R^{t}$ have already been terminated with some label $U_{e}$($U_{e}$ can be $U$ or some other label $U^{'}$). In this case, the output value of each configuration reached from $\tau$ is the limitation to $ids(\tau)$ of $\delta_{U^{e}}(N(s^{n}, U_{e}))$. Recall that in stage 2, the adversary terminates each vertex adjacent to a terminated vertex. Each vertex in the subdivision of $\tau$ is terminated with the label $U_{e}$ before phase 3, which means that we do not use the possible output values here. 
If there exists a $k$-simplex $\tau$ in $R^{t}$ that has some vertex $v$ terminated with some label $U_{e}$ and the limitation to $ids(\tau)$ of $\delta_{U_{e}}(N(s^{n}, U_{e}))$ satisfies $f$ where $\tau$ is in the subdivision of $s^{n} \in \mathbb{F}_{r_{m}}(U)$, then the adversary can submit a chain of queries from this $\tau$ to itself. By Lemma \ref{the:restrict_ebf_finite_chain_of_queries_lemma} this chain of queries must be finite and end with a $(k-1)$-simplex terminated with label $U_{e}$. And the chain of queries will not violate the invariants. The output values of the joining of this $(k-1)$-simplex and $v$ satisfy $f$. Because the adversary does not change $\delta$ and will not change the defined values once they have been set, the invariants will still hold.


If each $k$-simplex $\tau$ in $R^{t}$ has some vertices terminated with some label $U_{e}$ and the limitation to $ids(\tau)$ of $\delta_{U_{e}}(N(s^{n}, U_{e}))$ does not satisfy the requirement of $f$, then no $k$-simplex reached from $\tau$ will satisfy $f$. This is because in Section \ref{sec:nccondition:finalization_after_phase_1:stage2}, each vertex adjacent to a terminated vertex will be terminated with the same label. Each $k$-simplex reached from $\tau$ will terminate with the output $\delta_{U_{e}}(N(s^{n}, U_{e}))$.

Therefore, we will only discuss the $k$-simplices in $R^{t}$ that do not have a terminated vertex and $R^{t}$ contains at least one undefined $k$-simplex. For each undefined $k$-simplex  $\tau \in R^{t}$, let $A_{\tau}$ denote the subcomplex of $S^{t}$ that contains vertices at a distance of at most 1 to each vertex in $\tau$. In other words, $A_{\tau}$ consists of all n-simplices that contain $\tau$. There are several cases according to the status of the vertices of $A_{\tau}$. In cases 1 and 2, some vertices of $A_{\tau}$ have been terminated. These two situations are quite similar to the situation in which a vertex in $\tau$ is terminated. Therefore, the argument given in the last paragraph will be used again, and no possible output values are involved. In cases 3 and 4, all the vertices in $A_{\tau}$ are undefined. In case 3 where the surrounding vertices of $\tau$ are not defined, we will show that $\tau$ will be separated from other terminated sets. This means that the simplices in the subdivision of $\tau$ will not be connected to other terminated sets. But this does not mean that the possible output values we define are not used here. Consider the situation where a subsimplex of $\tau$ is in $\mathbb{F}_{t}(U^{in})$ for some label $U^{in}$. The output values of the configurations reached from $\tau$ can be in $POV(\tau, U^{in})$. In case 4, some vertex in $\tau$ is adjacent to a vertex terminated with the label $U_{e}$. Therefore, in addition to the possible output values with $U^{in}$, the possible output values with the label $U_{e}$ can also be assigned.

Case 1: There is an undefined $k$-simplex $\tau$ in $R^{t}$, some vertex in $A_{\tau}$ is terminated with the label $U_{e}$ and the limitation to $ids(\tau)$ of $\delta_{U_{e}}(N(s^{n}, U_{e}))$ satisfies $f$. 
The adversary will submit a chain of queries such that some $k$-simplex reached from $\tau$ will terminate with the output value that satisfies the function $f$. The adversary can return the schedule from configuration $C$ to this $k$-simplex. 

Case 2: There is an undefined $k$-simplex $\tau$ in $R^{t}$ and some vertex in $A_{\tau}$ is terminated with some label $U_{e}$, the limitation to $ids(\tau)$ of $\delta_{U_{e}}(N(s^{n}, U_{e}))$ does not satisfy $f$. Any simplex in a subdivision of $\tau$ will be adjacent to this terminated vertex. In Section \ref{sec:nccondition:finalization_after_phase_1:stage2}, each undefined vertex adjacent to a terminated vertex will be terminated with its label. Therefore, all configurations reached from the $k$-simplex $\tau$ will be terminated with the label $U_{e}$ and their output values will not satisfy the function $f$.

Case 3: There is an undefined $k$-simplex $\tau$ in $R^{t}$ such that every vertex in $A_{\tau}$ is undefined and all vertices adjacent to $\tau$ are undefined. 

Case 3.1: If $\tau$ is in $\mathbb{F}_{t}(U^{in})$ for some label $U^{in}$ and the limitation to $ids(\tau)$ of $\delta_{U^{in}}(N(s^{n}, U^{in}))$ satisfies $f$, the adversary can subdivide the current complex $S^{t}$ $r$ times, where $r = max(k * r_{s}, 3)$ (just to ensure that the adversary does not violate the restriction on assignment rules). We define the simplex $\rho = \{(i, \vec{\tau}):i \in Id(\tau)\}$ in $Ch(\tau)$ as the central simplex of the standard chromatic subdivision. Then the central simplex of $Ch^{r}(\tau)$ will have an active distance of at least 3 to all terminated configurations with other labels. The adversary can terminate this central simplex with the label $U^{in}$ and return the schedule from $C$ to this simplex.

Case 3.2: The discussions in case 3.1 are also true when only a subset $\tau^{'}$ of $\tau$ is in $\mathbb{F}_{t}(U^{in})$ for some label $U^{in}$. If the limitation to $ids(\tau)$ of $\delta_{U^{in}}(N(s^{n}, U^{in}))$ satisfies $f$, the adversary will still subdivide $S^{t}$ $r$ times, where $r = max(k * r_{s}, 3)$. In fact, any simplex in $Ch^{r}(\tau)$ will have an active distance of at least 3 to all terminated configurations with other labels. We will choose a simplex $s_{int}$ in $Ch^{r}(\tau)$ that has some subsimplex in the subdivision of $\tau^{'}$. The adversary can then terminate $s_{int}$ with the label $U^{in}$ and return the schedule from $C$ to this simplex.

Now, we have to discuss the output values introduced in the finalization after phase 1. 

Case 3.3: If some possible output value $O$ of $\tau$ with the label $U^{in}$ satisfies the function $f$, when a subset $\tau^{'}$ of $\tau$ is in $\mathbb{F}_{t}(U^{in})$ for some label $U^{in}$, then the adversary will construct a configuration whose output value is $O$. First, the adversary will construct a simplex $s_{int}$ terminated with the output $\delta_{U^{in}}(N(s^{n}, U^{in}))$, as in the last paragraph. By the property of possible output values, there is a path of output values $O_{0}, O_{1} \cdots O_{e}$ from $\delta_{U^{in}}(N(s^{n}, U^{in}))$ to the output $O$.

Let $R_{0}$ be a simplicial complex that contains only $s_{int}$. The adversary assigns the output values of $O_{1}$ to a simplex $s_{new}$ adjacent to $R_{0}$ in the subdivision of $\tau$ and sets $R_{i + 1} = s_{new}, i = i + 1$ such that the distance between $R_{i}$ and $R_{i - j}$ is $j - 1$ for each $j$.  This simplex is given the label $U^{in}$. Note that each terminated vertex should follow the restriction on the second assignment rule: The adversary subdivides the current complex many times before terminating one more vertex with the label $U^{in}$. The adversary repeats this procedure for $O_{2}, O_{3} \cdots O_{e} = O$ which means that a simplex $s_{d}$ in the subdivision of $\tau$ is assigned with an output $O$ satisfying $f$. These terminated simplices are far enough away from other terminated sets. But the adversary cannot stop here, since it has terminated some vertices with the label $U^{in}$ but not with the output $\delta_{U^{in}}(N(s^{n}, U^{in}))$, while this is the most basic assumption of our adversarial strategy. Let $Q$ be a simplicial complex that contains only $s_{d}$. The adversary hides this mismatch by assigning the output values of $O_{e - 1}$ to {\em each} simplex $s_{new}$ in the subdivision of $A_{\tau}$ adjacent to $Q$ and sets $Q = \cup s_{new}$. Repeat this procedure for $O_{e - 2}, O_{e - 3} \cdots O_{0}$. The adversary does not follow the restriction here (like in stage 2). The active distance between two sets terminated with different labels increases exponentially as the adversary subdivides the current complex when it constructs a path from $s_{int}$ to $s_{d}$. So, the active distance between a newly terminated simplex and a set terminated with a different label is at least 3. Note that when the adversary constructs $s_{d}$, it only assigns the output to {\em only one} simplex adjacent to $Q$ each time. This hiding is possible since each $n$-simplex containing $\tau$ is included in $A_{\tau}$. After this, the output values on the path from $s_{int}$ to $s_{d}$ will be sealed. If all the $n$-simplices in $A_{\tau}$ are in the subdivision of the $n$-simplices of one category, then the assignment of the simplices adjacent to $Q$ with $O_{e - 2}, O_{e - 1} \cdots O_{0}$ will not violate the task specification, as $O_{e - 2}, O_{e - 1} \cdots O_{0}$ are also possible output values for these simplices. The adversary will return the schedule from $C$ to the configuration $s_{d}$. If the $n$-simplices in $A_{\tau}$ are in the subdivision of the $n$-simplices of multiple categories, then the vertices in $Q$ will be designated with some $(k - k^{'})$-simplex in the category that has the highest dimension $k - k^{'}$. Therefore, a simplex adjacent to $Q$ will be designated with the same $(k - k^{'})$-simplex, which means that $O_{e - 2}, O_{e - 1} \cdots O_{0}$ will not violate the task specification. Since we have constructed a path from any simplex terminated with the label $U^{in}$ to a vertex in $\mathbb{F}_{t}(U^{in})$, invariant (2) remains true.

Case 3.4: For each label $U^{in}$ where a subset $\tau^{'}$ of $\tau$ is in $\mathbb{F}_{t}(U^{in})$, the limitations of $POV(\tau, U^{in})$ to $ids(\tau)$ do not satisfy the function $f$. Even if the projection to $ids(\tau)$ of $POV(\tau, U_{e})$ satisfies the function $f$, where no vertex of $\tau$ is in $\mathbb{F}_{t}(U_{e})$, the adversary cannot simply terminate some simplex reached from $\tau$ with the label $U_{e}$, since the adversary has to maintain invariant (2). However, it is also not safe to say that there will be no simplex reached from $\tau$ and terminated with such a label $U_{e}$ (i.e., that satisfies the function $f$). In previous work \cite{Alistarh19}, the adversary will terminate all vertices reached from $\tau$ with the limitation of $\delta_{U}(s^{n})$ to $ids(\tau)$ to eliminate ambiguity. If we adopt their method, it will look as follows: Since all vertices in $\tau$ are not adjacent to the terminated vertices, the active distance between $\tau$ and every terminated configuration is at least 1. The adversary can subdivide $S^{t}$ and increase $t$ many times so that all simplices reached from $\tau$ have an active distance of at least 3 to any terminated configuration. The adversary will terminate all vertices with the label $U$ to ensure that they do not terminate with any other label $U^{'}$ in the subsequent interaction. But this method does not work for our adversarial strategy, since it can generate a set that is too large to contain all vertices in the subdivision of a $k$-facet $s^{k} \in Q_{r_{m}}(U, U^{'})$. Recall that in Section \ref{sec:nccondition:finalization_after_phase_1}, we introduce a new rule to terminate a vertex with the output $\delta_{U}(s^{n})$. In fact, any simplex reached from $\tau$ will not terminate with the label $U^{'}$ is a direct consequence of Lemma \ref{the:glue_protocols:prover_fails:terminated_vertices_not_reach}.

Case 4: There is a $k$-simplex $\tau$ in $R^{t}$, every vertex in $A_{\tau}$ is undefined, but some vertex $v$ in $\tau$ is adjacent to a terminated vertex $v^{'}$ with the label $U_{e}$. Note that the vertex $v$ will not be in $A_{\tau}$ i.e. will not be adjacent to all vertices in $\tau$. By invariants, the labels of two possible terminated vertices $v^{'}$ will be the same, since the active distances between the configurations terminated with two different labels are at least 3. Cases 4.1 and 4.2 are the situations in which some output with the label $U_{e}$ satisfies the function $f$. Cases 4.3, 4.4 and 4.5 are situations in which no output with the label $U_{e}$ does.

Case 4.1: If the limitation of $\delta_{U_{e}}(N(s^{n}, U_{e}))$ to $ids(\tau)$ satisfies $f$, the adversary can submit a chain of queries corresponding to the solo execution of the vertex $v$. The vertex $v$ will be terminated with the label $U_{e}$ in later subdivisions of $\tau$. The adversary can terminate some vertices adjacent to $v$ in a further subdivision of $\tau$ with the label $U_{e}$ by submitting a chain of queries. This is possible since according to Lemma \ref{the:restrict_ebf_finite_chain_of_queries_lemma} this chain of queries must be finite and end with a $(k-1)$-simplex with label $U_{e}$. The adversary will return the schedule from $C$ to the configuration reached from $\tau$. Since we have constructed a path consisting of vertices terminated with the label $U_{e}$ to the resulting simplex, invariant (2) will remain true.

Case 4.2: If the label $U_{e}$ is not equal to $U$ and some possible output value to $\tau$ with the label $U_{e}$ satisfies $f$, the adversary can submit a chain of queries corresponding to the solo execution of the vertex $v$ that will be terminated with the label $U_{e}$. The adversary will then use the techniques in Case 3.3 to construct the configuration whose output satisfies $f$ and return it as a result.

Otherwise, the output of any configuration terminated with the label $U^{'}$ does not satisfy the function $f$. We have to return to the discussion of the label $U^{in}$ where some subset of $\tau$ is in $\mathbb{F}_{t}(U^{in})$, as in Case 3.

Case 4.3 (corresponding to Case 3.1): If $\tau$ is in $\mathbb{F}_{t}(N(s^{n}, U^{in}))$ for some label $U^{in}$ and the limitation of $\delta_{U^{in}}(s^{n})$ to $ids(\tau)$ satisfies $f$, the adversary can terminate a central simplex of $Ch^{r}(\tau)$ with the label $U$ for some $r$. The central simplex after several subdivisions will have an active distance of at least 3 to all configurations terminated with different labels. The adversary can return the schedule from $C$ to this simplex. 

Case 4.4 (corresponding to Case 3.2 and Case 3.3): Suppose that a subset $\tau^{'}$ of $\tau$ is in $\mathbb{F}_{t}(U^{in})$ and that the limitation of $\delta_{U^{in}}(N(s^{n}, U^{in}))$ or some output of $POV(\tau, U^{in})$ to $ids(\tau)$ satisfies the function $f$. A difference here from Case 3 is that all vertices in $\tau^{'}$ should not be adjacent to a single terminated vertex with some label $U_{e}$. Otherwise, it should be discussed in Case 4.5. The adversary can terminate some internal simplex $s_{int}$ reached from $\tau$ with the label $U^{in}$ just as in Case 3. Furthermore, if the limitation of some output $O$ of $POV(\tau, U^{in})$ to $ids(\tau)$ can satisfy the function $f$, then the adversary will first terminate an internal $k$-simplex $s_{int}$ with the output $\delta_{U^{in}}(N(s^{n}, U^{in}))$. By the property of possible output values, we can construct a path from $s_{int}$ to $s_{d}$ whose output values are $O$. The adversary has constructed a simplex whose output satisfies the function $f$. Again, it cannot stop here, since it has terminated some vertices with the label $U^{in}$, but not with the output $\delta_{U^{in}}(N(s^{n}, U^{in}))$. The operation of sealing the vertices whose outputs are not $\delta_{U^{in}}(N(s^{n}, U^{in}))$ will be repeated here.

Case 4.5 (corresponding to Case 3.4): For each label $U^{in}$ where a subset $\tau^{'}$ of $\tau$ is in $\mathbb{F}_{t}(U^{in})$, the limitation of any possible output value of $\tau$ with label $U^{in}$ does not satisfy the function $f$. Then, for the same reason as in Case 3.4, it is not possible to assign any configuration reached from the simplex $\tau$ with an output that satisfies the function $f$.

After checking each $k$-simplex in $R^{t}$, the adversary is able to answer the assignment query. If the adversary manages to find an simplex in the subdivision of some $k$-simplex $\tau$ whose output satisfies $f$, it can return the schedule from $C$ to it. But if each $k$-simplex $\tau \in R^{t}$ cannot reach a configuration whose output satisfies the function $f$, the adversary will return $NULL$. By the strategy described above, the result of assignment queries will not be violated. If an assignment query returns a schedule, the schedule will remain valid in the later interaction, since the adversary will not change $\delta(v)$ for each vertex $v$ once it has been set. On the other hand, if the assignment query returns $NULL$, we have already proved that any configuration reached from $C$ will not satisfy the function $f$.

\begin{lemma}
\label{the:augmented_ebf_no_violation_lemma}
No assignment queries made in phase 1 will be violated.
\end{lemma}

Note that in all circumstances, the invariants will remain after the assignment query, and all terminated vertices with a label $U^{in}$ can be seen as terminated with the output $\delta_{U^{in}}(N(s^{n}, U^{in}))$. Furthermore, since the output of each simplex $\tau$ assigned in the execution of an assignment query is obtained from possible output values of $\tau$ , no task specification will be violated.

\begin{theorem}    
\label{the:assignment_queries_no_more_power}
Assignment queries do not give an extension-based prover more power in its interaction with a protocol.
\end{theorem}

Before proceeding to the next section, we want to talk about something proposed in \cite{Brusse21}. In the original proof in \cite{Alistarh19}, one of the invariants guarantees that any two vertices terminated with different output values are sufficiently far apart in $S^{t}$, which means that at the end of phase 1 no simplex in $S^{t}$ contains vertices that output different values. In \cite{Brusse21} it is pointed out that it is impossible to maintain this invariant when the prover is allowed to submit assignment queries, since if so, the adversary will tell the prover that it can solve the consensus task. What they do to fix this issue is that the output configurations with two different values are allowed in the interaction of phase 1. The idea behind this solution is not explicitly given in \cite{Brusse21}. We can interpret this solution in a simpler way using our consequence: output values introduced in the finalization after phase 1 have to be and can be used in the interaction during phase 1 when the assignment queries are allowed.

\subsection{Finalization after phase r}
\label{sec:nccondition:finalization_after_phase_r}

We have discussed the conditions for the finalization after the first phase. But will the adversary have to finalize in the latter phases when dealing with more complicated tasks, e.g. tasks more difficult than $(n, 2)$-set agreement? As described in the definition of extension-based proofs, the adaptive protocol specified by $\delta$ can be finalized after any phase $r$.

But we face a new challenge when discussing the finalization after the phase $r$. Before we go directly to finalization after phase $r$, it is better to discuss finalization after phase two to clarify the differences. Take a 3-process task as an example. To keep the settings simple, we assume that the input complex $\mathcal{I}$ contains only one $n$-simplex. The complex $\mathbb{F}_{1}(\{(v_{0}, 0)\})$ is the joining of a 0-simplex $s^{top}$ and the chromatic subdivision of a 1-simplex $s^{bot}$. We denote the three $1$-simplices in this subdivision of $s^{bot}$ by $s_{0}$, $s_{1}$ and $s_{2}$. Suppose that there exists no partial protocol with respect to the 0-simplex $\{(v_{0}, 0)\}$, but a partial protocol whose configurations are reached from each 2-simplex in $\mathbb{F}_{1}(\{(v_{0}, 0)\})$. We say that each such partial protocol is specified by two simplices in $\mathcal{I}$. For example, the partial protocol defined for the joining of $s^{top}$ and $s_{1}$ is specified by $[\{(v_{0}, 0)\}, \{(v_{1}, 0), (v_{2}, 0)\}]$ which are, respectively, the first and second sets of processes to take a step.  In previous sections, we used the canonical neighbors to obtain the $\delta$ value of a vertex with some label $U^{'}$: If some $n$-simplex is reached from some $n$-simplex $C_{1}$ in $Ch(\mathcal{I})$ by a schedule $\alpha$, then its canonical neighbor with the label $U^{'}$ is defined as the $n$-simplex reached from $C_{2}$ in $Ch(\mathcal{I})$ by $\alpha$, where $C_{2}$ is calculated from $C_{1}$ according to some rule in Section \ref{sec:preparations:canonical_neighbor}. This analysis can be used for the simplicial complex consisting of $s_{0}$, $s_{1}$ and $s_{2}$ since this complex is the first chromatic subdivision of $s^{bot}$. We use some notation about canonical neighbors here: Let $C_{1}$ be the $1$-simplex $s_{0}$ and $U^{'}$ be the label $[\{(v_{0}, 0)\}, \{(v_{1}, 0), (v_{2}, 0)\}]$, then $C_{2}$ is the $1$-simplex $s_{1}$.

Let $s^{n}$ be an $n$-simplex reached from the joining of $s^{top}$ and the 1-simplex $C_{1} = s_{0}$ by schedule $\alpha$. It seems quite natural to define the canonical neighbor of $s^{n}$ with the label $[\{(v_{0}, 0)\}, \{(v_{1}, 0), (v_{2}, 0)\}]$ as the simplex reached by the schedule $\alpha$ from the joining of $s^{top}$ and $C_{2} = s_{1}$. However, the third requirement of canonical neighbors may not be satisfied. In fact, this has happened before in the proof of Lemma \ref{the:canonical_neighbor_existence_lemma} in which we solve this problem by assuming $r_{m} \geq 2$. Perhaps we should recall the details of Lemma \ref{the:canonical_neighbor_existence_lemma} and present the idea behind the differences. In $Ch^{1}(\mathcal{I})$, $C_{2}$ is not the canonical neighbor of $C_{1}$ as some processes $P$ will see different sets of input values. But(since we assume that $r_{m} \geq 2$) the second partition of $\Pi$ starts with a set of processes that contains some shared vertex of $C_{1}$ and $C_{2}$ that eliminates the mismatch of input values that the processes in $P$ see. This argument does not work here: not each intersection vertex can help to eliminate the mismatch. More specifically, the vertices in $s^{top}$ cannot help to eliminate the mismatch. In Figure \ref{img:canonical_neighbor_with_cone}, we give a simpler example of a 3-process task that. The canonical neighbor of a 2-simplex of red color satisfies the third requirement, while the canonical neighbor of the 2-simplex of blue color does not.

\begin{figure}[h]
  \centering
  \includegraphics[width=\linewidth]{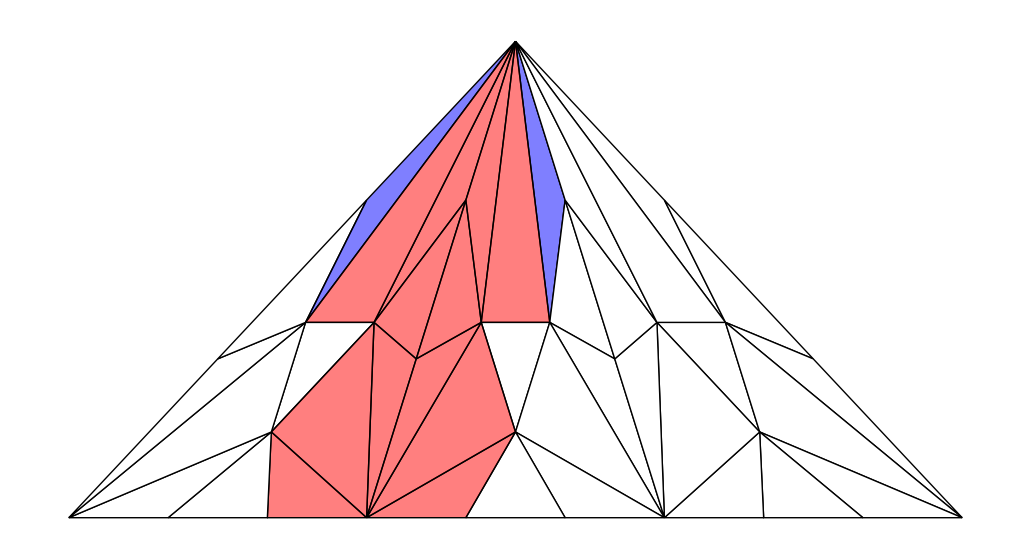}
  \caption{Canonical neighbors in $\mathbb{F}_{2}(\{v_{0}\}, 0)$ of a 3-process task}
  \label{img:canonical_neighbor_with_cone}
\end{figure}

For each simplex $U \in \mathcal{I}$, the complex $\mathbb{F}_{1}(U)$ can be seen as the joining of a $|ids(U)|$-simplex $s^{top}$ and the chromatic subdivision of a $|\Pi - ids(U)|$-simplex. $\mathbb{F}_{1}(U)$ can be divided further into a set of simplicial complexes, each of which corresponds to a subset of $(\Pi - U)$. Suppose that there exists a partial protocol for each such input complex. To distinguish these partial protocols, we assign a label $[U_{1}, U_{2}]$ to each partial protocol, where $ids(U_{2})$ is a subset of the set $(\Pi - ids(U))$ and is the second set of processes to execute. In other words, the label is now an ordered list of simplices in $\mathcal{I}$.

Let $s^{n} \in \mathbb{F}_{r_{m}}([U_{1}, U_{2}])$ be an n-simplex reached from the joining of $s^{top}$ and a $|\Pi - U|$-simplex $C_{1}$ by the schedule $\alpha$.
We define the canonical neighbor of $s^{n}$ with the label $[U_{1}, U^{'}_{2}]$ as an $n$-simplex reached by $\alpha$ from the joining of $s^{top}$ and $C_{2}$ where $C_{2}$ is calculated from $C_{1}$(as we do in the definition of canonical neighbors). An $n$-simplex whose intersection with $\mathbb{F}_{r_{m}}([U_{1}, U_{2}^{'}])$ is in the subdivision of $s^{top}$ may have a canonical neighbor with a larger carrier, as the example shows. Note that using the techniques in the proof of Lemma \ref{the:canonical_neighbor_existence_lemma}, other $n$-simplices will not have this issue. We introduce a rule, which is referred to as the carrier rule in subsequent discussions, for such a type of $n$-simplices whose canonical neighbor has a larger carrier for some processes. 

The vertices in $s^{n}$ can be divided into two types: those in the subdivision of $s^{top}$ and those not. Consider a vertex of the latter type. If the vertex $v$ with the process id $p$ has a smaller carrier compared to the vertex $v^{'}$ with the same process id in the canonical neighbor of $s^{n}$ with the label $[U_{1}, U^{'}_{2}]$, then it is not safe to use $\delta_{[U_{1}, U^{'}_{2}]}(v^{'})$ to terminate the vertices with the process id $p$ in the subdivision of $s^{n}$.
Instead, a vertex $v$ with the process id $p$ in the subdivision of $s^{n}$ will terminate with $\delta_{[U_{1}, U^{'}_{2}]}(v^{''})$, where $v^{''}$ is some vertex of the first type in $s^{n}$. Let $\tau^{'}$ be the carrier of $v$ in $Ch^{r_{m}}(\mathcal{I})$. $v^{''}$ is chosen from the vertices in $\tau^{'}$ that are also in $s^{top}$. There will be no violation of the task specification, since the vertex $v^{'}$ has only seen the input values of the processes in $U$, which all the vertices in $\mathbb{F}_{1}(U_{1})$ have seen, and we have Property \ref{pro:colorless_property}. We have to show that the carrier rule will not cause any conflict of values. Recall that the canonical neighbors of two $n$-simplices sharing some vertex will share the vertex with the same process id. If some other $n$-simplex $s^{n}_{1}$ contains a vertex $v$ terminated by the carrier rule, then by the property of canonical neighbors, the vertex $v$ will also have a larger carrier in the canonical neighbor of $s^{n}_{1}$. This means that $v$ will terminate with the carrier rule, and the terminated value will be identical since the choice of $v^{''}$ will be the same.

So what is the difference of phase 1 when the partial protocol with respect to $U_{1}$ is replaced by a sequence of "smaller" partial protocols for $[U_{1}, U_{2}], [U_{1}, U^{'}_{2}] \cdots$? We have to check whether the previous results can still be used. First, we will show that Lemma \ref{the:canonical_neighbor_existence_lemma} can be used for these partial protocols: if an $n$-simplex $s^{n}$ has an intersection with $\mathbb{F}_{r_{m}}([U_{1}, U_{2}])$, then there exists a canonical neighbor of $s^{n}$ with this label. We use the canonical neighbor $s^{n}_{c}$ of $s^{n}$ with the label $U_{1}$ to construct its canonical neighbor with the label $[U_{1}, U_{2}]$. If $s^{n}_{c}$ is in $\mathbb{F}_{r_{m}}([U_{1}, U_{2}])$, then $s^{n}_{c}$ is the canonical neighbor of $s^{n}$ with the label $[U_{1}, U_{2}]$. Otherwise, $s^{n}_{c}$ is in another protocol complex $\mathbb{F}_{r_{m}}([U_{1}, U_{2}^{'}]$). Let $V_{int}$ be the set of vertices shared by $s^{n}$ and $\mathbb{F}_{r_{m}}([U_{1}, U_{2}])$. Then $s^{n}_{c}$ will also contain $V_{int}$. Therefore, there exists a canonical neighbor of $s^{n}_{c}$ with the label $[U_{1}, U_{2}]$. We will use the canonical neighbor of $s^{n}_{c}$ with the label $[U_{1}, U_{2}]$ as the canonical neighbor of $s^{n}$ with the label $[U_{1}, U_{2}]$.

Next, we prove that there will be no conflict of output values. If the vertex $v$ with process id $p$ is shared by two $n$-simplices having an intersection with $\mathbb{F}_{r_{m}}([U_{1}, U_{2}])$, then the canonical neighbors with $U_{1}$ will share the vertex with the same process id. When the canonical neighbors with the label $U_{1}$ are both in or not in $\mathbb{F}_{r_{m}}([U_{1}, U_{2}])$, the canonical neighbors with the label $[U_{1}, U_{2}]$ will share the vertex with the process id $p$. If only one of the canonical neighbors with label $U_{1}$ denoted by $s^{n}_{c}$ is in $\mathbb{F}_{r_{m}}([U_{1}, U_{2}^{'}])$ for some different $U_{2}^{'}$, then the vertex $v$ will be at the intersection of $\mathbb{F}_{r_{m}}([U_{1}, U_{2}])$ and $\mathbb{F}_{r_{m}}([U_{1}, U_{2}^{'}])$. The canonical neighbor of $s^{n}_{c}$ with the label $[U_{1}, U_{2}]$ will not change the intersection vertices, including the vertex $v$. Therefore, the canonical neighbors with the label $[U_{1}, U_{2}]$ will still share the vertex with the process id $p$. 
Therefore, the partial protocol $\delta_{U_{1}}$ will be replaced by a sequence of partial protocol $\delta_{[U_{1}, U_{2}]}$ while our interaction framework is kept valid.


After the second phase, the first two sets of processes are determined, which means that the adversary can limit its view to $\mathbb{F}_{t}([U_{1}, U_{2}])$. The adversary can use the techniques in Section \ref{sec:nccondition} to construct a partial protocol for $\mathbb{F}_{t}([U_{1}, U_{2}])$. The condition required to connect output values of different labels is that $CEN(s^{k})$, where $s^{k} \in \mathbb{F}_{1}(\mathcal{I})$ has the same output in different partial protocols. The carrier rule will not affect these procedures since the carrier rule will change the values of vertices that are not shared by $s^{n}$ and its canonical neighbor. Only the output values of the shared vertices are used to add layers to the terminated sets.

\begin{theorem}
\label{the:adversary_finalize_after_the_second_phase_theorem}
For a colorless task$(\mathcal{I}, \mathcal{O}, \Delta)$, there exists an adversary that can finalize after the second phase and win against any restricted extension-based prover if and only if there exists a partial protocol with respect to $[U_{1}, U_{2}])$ for each possible sequence $[U_{1}, U_{2}])$  and all partial protocols are compatible.
\end{theorem}


The decomposition of the partial protocol with respect to $[U_{1}, U_{2}]$ can continue. Suppose that some protocol complex is labeled by $[U_{1}, \allowbreak ... U_{k}]$ where all $U_{i}$ are subsimplices of an $n$-simplex in $\mathcal{I}$ and the process set $ids(U_{i}) \subseteq \Pi - \cup_{1}^{i - 1}ids(U_{l})$ for each $i$. All the properties we have reviewed above are also satisfied. For each sequence, the protocol complex $\mathbb{F}_{1}(U_{1}, U_{2}\allowbreak... U_{k})$ can be seen as the joining of a $|\cup_{1}^{k} ids(U_{i})|$-simplex $s^{top}$ and the chromatic subdivision of a $|\Pi - \cup_{1}^{k}ids(U_{i})|$-simplex. It can also be divided into a sequence of protocol complexes, each of which corresponds to a simplex $U_{k + 1}$ such that $ids(U_{k + 1}) \in \Pi - \sum_{1}^{k}ids(U_{i})$. Suppose that there exists a partial protocol for each such protocol complex, which is assigned with a label $[U_{1}, \allowbreak... U_{k}, U_{k + 1})]$. For convenience, we will denote the old and new labels as $l_{o}$ and $l_{n}$. Let $s^{n}$ be an $n$-simplex reached from the joining of $s^{top}$ and a $|\Pi - \sum_{1}^{k}ids(U_{i})|$-simplex $C_{1}$ by the schedule $\alpha$. We define the canonical neighbor of $s^{n}$ as the $n$-simplex reached from the joining of $s^{top}$ and $C_{2}$ by the schedule $\alpha$, where $C_{2}$ is generated from $C_{1}$ as defined in the construction of canonical neighbors. The carrier rule will be used to avoid violating the task specification.

Lemma \ref{the:canonical_neighbor_existence_lemma} can be used for the new protocol complexes: if an $n$-simplex $s^{n}$ has an intersection with $\mathbb{F}_{r_{m}}(l_{n})$, then there exists a canonical neighbor of $s^{n}$ with the label $l_{n}$. Suppose that an $n$-simplex $s^{n}$ has some intersection vertices with $\mathbb{F}_{r_{m}}(l_{n})$, denoted by $V_{int}$. According to the induction hypothesis, it will have a canonical neighbor with the label $l_{o}$. If the canonical neighbor $s^{n}_{c}$ with the label $l_{o}$ of $s^{n}$ is in the protocol complex of the label $l_{n}$, then we will use it as the canonical neighbor of $s^{n}$ with the label $l_{n}$. Otherwise, $s^{n}_{c}$ is in another protocol complex, then $s^{n}_{c}$ also contains the vertices set $V_{int}$. So, there exists a canonical neighbor of $s^{n}_{c}$ with the label $l_{n}$. We will use the canonical neighbor of $s^{n}_{c}$ with the label $l_{n}$ as the canonical neighbor of $s^{n}$ with the label $l_{n}$. 
There will be no conflict of output values. If the vertex $v$ with the process id $p$ is shared by two $n$-simplices having an intersection with $\mathbb{F}_{r_{m}}(l_{n})$, then the canonical neighbors with the label $l_{o}$ will share the vertex with the same process id. If the canonical neighbors with the label $l_{o}$ are both in or not in $\mathbb{F}_{r_{m}}(l_{n})$, the canonical neighbors with the label $l_{n}$ will share the vertex with the process id $p$. If only one of the canonical neighbors with the label $l_{o}$ denoted by $s^{n}_{c}$ is {\em not} in $\mathbb{F}_{r_{m}}(l_{n})$, then the vertex $v$ will be at the intersection of $s^{n}_{c}$ and $\mathbb{F}_{r_{m}}(l_{n})$. The canonical neighbor of $s^{n}_{c}$ with the label $l_{n}$ shares the vertex $v$ with $s^{n}$. Therefore, the canonical neighbors with the label $l_{o}$ of two $n$-simplices share the vertex with the process id $p$. 

After phase $k$, the first $k$ process sets are determined which means the adversary can limit itself to $\mathbb{F}_{t}([U_{1}, \allowbreak ... U_{k}])$, a subcomplex of $S^{t}$. The adversary can use the techniques in Section \ref{sec:nccondition} to generate a partial protocol.

\begin{theorem}
\label{the:adversary_finalize_after_the_k_phase_theorem}
For a colorless task $(\mathcal{I}, \mathcal{O}, \Delta)$, there exists an adversary that can finalize after the $k$-th phase and win against any restricted extension-based prover if and only if there exists a partial protocol with respect to $[U_{1}, \allowbreak ... U_{k}]$ for each possible sequence $[U_{1}, \allowbreak ... U_{k}]$ and all partial protocols are compatible.
\end{theorem}

Now we will discuss the case where $ids(U_{1}), ids(U_{2}) \cdots ids(U_{k})$ is a full 1-round schedule. We will show that any further division of the protocol complex in this case is unnecessary. Any further division will still imply the existence of a partial protocol with respect to $[U_{1}, \allowbreak ... U_{k}]$. Consider the configurations reached from the schedule that repeats the first partition $ids(U_{1}), ids(U_{2}), \allowbreak... ids(U_{k})$ of $\Pi$. For example, if the protocol complex with respect to $[U_{1}, \allowbreak ... U_{k}]$ is further divided into a sequence of protocol complexes with respect to $[U_{1}, \allowbreak ... U_{k}, U_{k + 1}]$ where $U_{k+1}$ is some subsimplex of the joining of all $U_{i}$, then there is a partial protocol $\delta_{dp}$ for the configurations reached from the schedule which repeats the first partition of $\Pi$ twice. This partial protocol $\delta_{dp}$ can be transformed into a partial protocol $\delta_{p}$ for $[U_{1}, \allowbreak ... U_{k}]$ without violating the carrier map $\Delta$: given a schedule $\beta$, $\delta_{p}$ defines the output values of processes in the configuration reached by $\beta$ as the output values given by $\delta_{dp}$ of processes in the configuration reached by $ids(U_{1}), ids(U_{2}), \allowbreak... ids(U_{k}) \beta$. But this is not enough for our requirements: the partial protocol with respect to $[U_{1}, \allowbreak ... U_{k}]$ shall be compatible with other partial protocols. If we just use $\delta_{dp}$ as the partial protocol with respect to $[U_{1}, \allowbreak ... U_{k}]$, the output values of $CEN(s^{k})$, where $s^{k}$ is some $k$-simplex in $\mathcal{I}$ and $CEN(s^{k})$ is in $\mathbb{F}_{r_{m}}([U_{1}, \allowbreak ... U_{k}])$, will change. We show that we can overcome this problem.

Suppose that there is a sequence of partial protocols(each has a label $[U_{1}, \allowbreak ... U_{k}, U_{k + 1, i}]$) whose protocol complexes(here we mean $\mathbb{F}_{2}$) form a covering of $\mathbb{F}_{2}([U_{1}, \allowbreak ... U_{k}])$. Note that we allow the adversary to use a set of "smaller" partial protocol to pretend it has a partial protocol with respect to a label $[U_{1}, \allowbreak ... U_{k}, U_{k + 1, i}]$. What we will do is to construct a partial protocol with respect to $[U_{1}, \allowbreak ... U_{k}]$ from this sequence of partial protocols . A key observation is that for any $CEN(s^{k})$ that we mentioned in the last paragraph is contained in only one protocol complex with respect to $[U_{1}, \allowbreak ... U_{k}, U_{k + 1, i}]$. And each protocol complex with respect to $[U_{1}, \allowbreak ... U_{k}, U_{k + 1, i}]$ contains such a simplex $CEN(s^{k})$, which is reached from the $n$-simplex $U_{1} * U_{2} \cdots U_{k}$ by a schedule repeating $ids(U_{k + 1, i})$. We denote such a simplex as $s_{i}$. We will build a new partial protocol with respect to each $[U_{1}, \allowbreak ... U_{k}, U_{k + 1, i}]$ using the existing one. Note that the new partial protocol may not terminate after the same number of rounds as the old one. We denote this new value by $r_{m}^{'}$. We do not change the output values of $CEN(s^{k})$ corresponding to this partial protocol, but have the configurations in $\mathbb{F}_{r_{m}^{'}}([U_{1}, \allowbreak ... U_{k}, U_{k + 1, i}]) \cap \mathbb{F}_{r_{m}^{'}}([U_{1}, \allowbreak ... U_{k}, U_{k + 1, j}])$ output the same values for each $i$ and $j$. Therefore, the constructed partial protocols can be merged into a partial protocol with respect to $[U_{1}, \allowbreak ... U_{k}]$.

All processes have seen the input values of the processes in $U_{1}$. Therefore, we can use the output values of the processes in the configuration $s_{c}$ reached by the schedule that repeats $ids(U_{1})$ until termination to connect our constructed partial protocols.

For each simplex $U_{k + 1, i}$, there is a path $P$ of output values(as we construct in the next paragraph) from the output values $s_{i}$ to the output values of $s_{c}$ using the original sequence of partial protocols. Let this path be $o_{0}, o_{1} \cdots o_{e}$. Given an existing partial protocol, we will construct a new partial protocol with respect to $[U_{1}, \allowbreak ... U_{k}, U_{k + 1, i}]$ while keeping the output values of $s_{i}$ unchanged. First, we define the output values of $s_{i}$ with the output values given by the existing partial protocol. Then we terminate each vertex adjacent to the terminated vertices with the value $o_{0}$. Finally, we will add layers corresponding to the path from $o_{1}$ to $o_{e}$. These added layers should be included in $\mathbb{F}_{r_{m}^{'}}([U_{1}, \allowbreak ... U_{k}, U_{k + 1, i}])$(if not, we can increase the round number $r_{m}^{'}$ of our constructed partial protocol with respect to $[U_{1}, \allowbreak ... U_{k}, U_{k + 1, i}]$).

We can construct the path $P$ consisting of output values. At first, we construct a path $P_{v}$ in the union of all protocol complexes (here we mean $\mathbb{F}_{r_{m}}$, not $\mathbb{F}_{r_{m}^{'}}$). There is a path $P_{v}$ in the union of all protocol complexes from $s_{i}$ to $s_{c}$. If some vertex $v_{1}$ in $P_{v}$ is in different protocol complexes at the same time, we replace $v_{1}$ with an order sequence of vertices $v_{1}^{'}, v_{2}^{'} \cdots v_{l}^{'}$, such that all vertices have the same carrier in $Ch(\mathcal{I})$, $v_{1} = v_{1}^{'} = v_{l}^{'}$ and there is some $v_{m}$ that is assigned with the same output value by different partial protocols for some $1 \leq m \leq l$. This is always possible since the original sequence of partial protocols is compatible. Now we construct the path $P$ by defining $P(i)$ as the output value of $P_{v}(i)$ given by a partial protocol containing $P_{v}(i)$.

And the carrier of $s_{i}$ in $Ch(\mathcal{I})$ is the smallest compared to any other vertex in the subdivision of $\mathbb{F}_{r_{m}^{'}}([U_{1}, \allowbreak ... U_{k}, U_{k + 1, i}])$. Therefore, adding layers using output values in the path $P$ can avoid any violation of the task specification. 

Finally, the newly constructed partial protocols can be merged into a single partial protocol with respect to $[U_{1}, \allowbreak ... U_{k}]$ by defining each undefined vertex in $\mathbb{F}_{r_{m}^{'}}([U_{1}, \allowbreak ... U_{k}])$ with the output value of some vertex in $s_{c}$.


\begin{theorem}
\label{the:adversary_finalize_theorem}
For a colorless task$(\mathcal{I}, \mathcal{O}, \Delta)$, there exists an adversary that can win against any extension-based prover if and only if there exists a partial protocol with respect to $[U_{1}, \allowbreak ... U_{k}]$ for each possible sequence $[U_{1}, \allowbreak ... U_{k}]$, where $ids(U_{1}), ids(U_{2}) \cdots ids(U_{k})$ is a full 1-round schedule, and all partial protocols are compatible.
\end{theorem}

\fi

\section{Applications}
\label{sec:applications}
We present an application of Theorem \ref{the:adversary_finalize_theorem}: There is an extension-based proof for the impossibility of solving colorless covering tasks using wait-free protocols. \cite{Attiya23_2} gives a similar result for colorless covering tasks using round reduction, a concept proposed in their paper. In terms of techniques, our proofs are completely different from previous proofs. Using our consequence, we can understand this result from a combinatorial perspective.

We adopt the definition in \cite{Attiya23_2}. 
For two connected simplicial complexes $\mathcal{I}$ and $\mathcal{O}$, and for a simplicial map $f: \mathcal{O} \rightarrow \mathcal{I}$, the pair $(\mathcal{O}, f)$ is a covering complex of $\mathcal{I}$ if, for every $\tau \in\mathcal{I} $, $f^{-1}(\tau)$ is a union of pairwise disjoint simplices. These disjoint simplices are called the sheets of $\tau$.

Given a covering complex $(\mathcal{O^{*}}, f)$ of a complex $\mathcal{I}$, the colorless covering task $(\mathcal{I^{*}}, \mathcal{O^{*}}, \Delta^{*})$ is the task where $\Delta^{*}(\tau)$ is defined as $\{\tau \in \mathcal{O^{*}} | f(\tau) \subseteq \tau\}$, for every $\tau \in \mathcal{I^{*}}$. A covering complex is non-trivial if each simplex in $\mathcal{I^{*}}$ has more than one sheet. Note that this is a definition in colorless form. We can use the transformation given in Section \ref{sec:preliminaries} to convert it to a colored form $(\mathcal{I}, \mathcal{O}, \Delta)$.

The Hexagone task $HX$ is one example of colorless covering tasks. The input complex $I^{*}$ is defined as a 3-node circle $\{u_{0}, u_{1}, u_{2}\}$. The output complex is a 6-node circle $\{v_{0}, v_{1} \cdots v_{5}\}$. The map $\Delta$ is defined as follows. For every $i \in \{0, 1, 2\}$, $\Delta(u_{i}) = \{v_{i}, v_{(i + 3) mod 6} \}$ and $\Delta(\{u_{i}, u_{(i + 1) mod 3} \}) = \{\{v_{i}, v_{i + 1} \}, \{v_{i + 3}, v_{(i + 4) mod 6} \} \}$.

\begin{figure}[h]
  \centering
  \includegraphics[width=\linewidth]{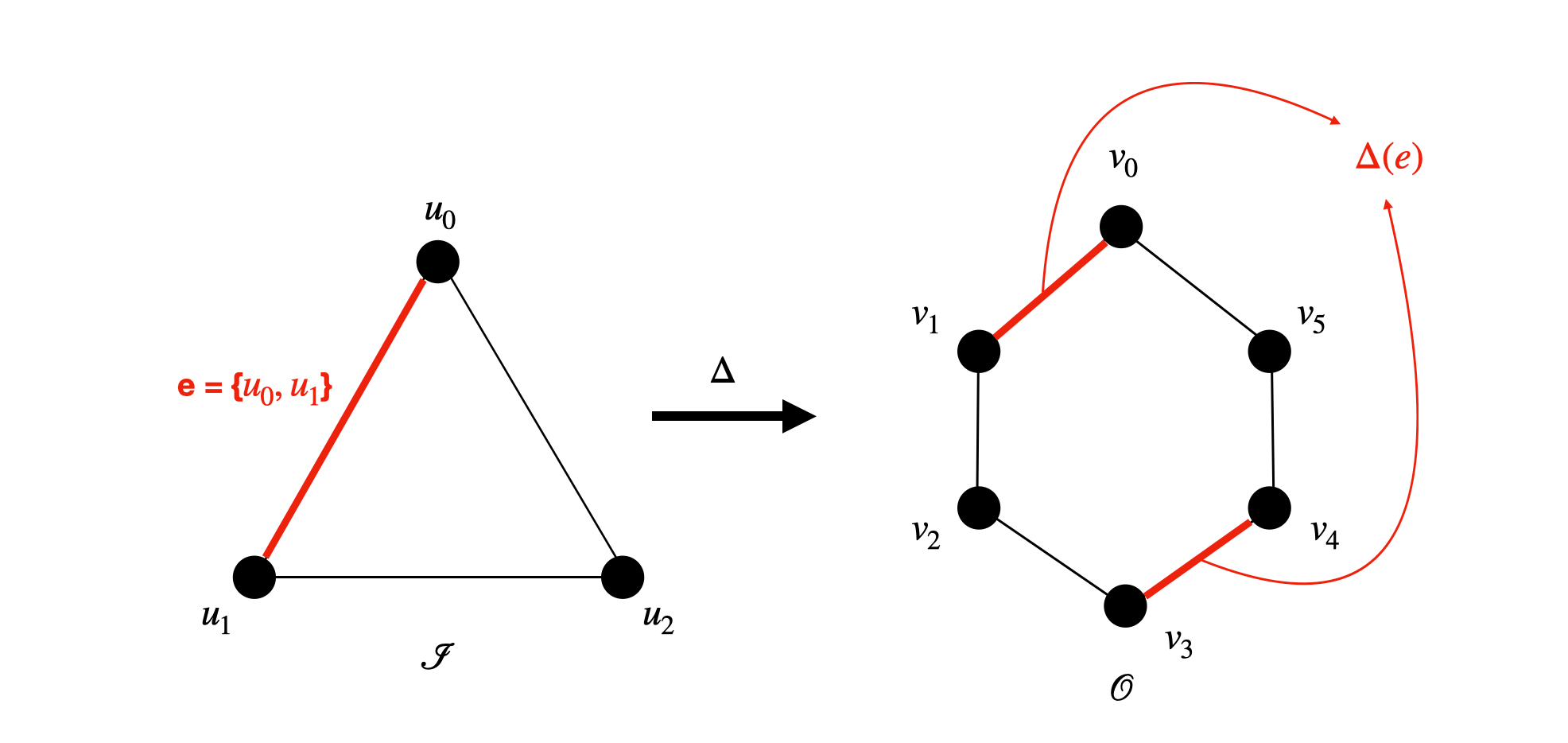}
  \caption{The Hexagone task}
  \label{img:application:Hexagone}
\end{figure}

\begin{lemma}
\label{the:covering_lemma}
    Given any non-trival colorless task $(\mathcal{I}, \mathcal{O}, \Delta)$, there exists a 1-dimensional subcomplex $\mathcal{I}^{'}$ of $\mathcal{I}$ such that $(\mathcal{I}^{'}, \Delta(\mathcal{I}^{'}), \Delta)$ has no wait-free algorithms, when $n \geq 2$.
\end{lemma}

\begin{proof}
    This proof is greatly inspired by \cite{Attiya23_2}. 
    Let $x$ be a vertex in $\mathcal{I^{*}}$, there are at least two vertices $y$ and $y^{'}$ in $\mathcal{O^{*}}$ such that $f(y) = f(y^{'}) = x$ since the task is not trival. As $\mathcal{O^{*}}$ is connected, there is a path $y_{0} = y, y_{1}, \cdots y_{k} = y^{'}$. The image of this path is a cycle $x_{0} = f(y_{0}), \cdots, x_{k} = f(y_{k})$ in $\mathcal{I^{*}}$ since $f(y_{0}) = f(y_{k})$.

    We build a 1-dimensional subcomplex $\mathcal{I}^{'}$ of $\mathcal{I}$ as follows: two processes $p_{0}$ and $p_{1}$ have inputs $(x_{i}, x_{(i + 1) mod (k + 1)})$, and $(x_{(i + 1) mod (k + 1)}, x_{(i + 1) mod (k + 1)})$ for all $i \leq k$. $\mathcal{I}^{'}$ is therefore a circle $C = (v_{0}, v_{1} \cdots v_{2k  + 1})$.  Assume that $(\mathcal{I}^{'}, \Delta(\mathcal{I}^{'}), \Delta)$ has a wait-free algorithm, this means that there exists a subdivision of $\mathcal{I}^{'}$ and a decision map $m$ that maps each vertex in the subdivision of $\mathcal{I}^{'}$ to $\Delta(\mathcal{I}^{'})$. Let $v$ be a vertex in $\mathcal{I}^{'}$, we use $\{v\}$ to denote the vertex in the subdivision of $\mathcal{I}^{'}$ reached from $v$ by a solo execution of the process $ids(v)$. 

    Then $m(\{v_{0}\}) = m(\{v_{1}\}) = y_{0}$ but $m(\{v_{2k + 1}\}) \neq y_{k}$, then there exist a minimal index $0 \leq i < k$ such that $m(\{v_{2i}\}) = m(\{v_{2i + 1}\}) = y_{i}$ and $m(\{v_{2i + 2}\}) \neq y_{i + 1}$.  We have $f(y_{i}, y_{i + 1}) = (x_{i}, x_{i + 1})$ by construction. The map $m$ has to satisfy the carrier map $\Delta$, so $f(y_{i}, m(\{v_{2i + 2}\})) = (x_{i}, x_{i + 1})$. Hence $f^{-1}(x_{i}, x_{i + 1})$ contains $(y_{i}, y_{i + 1})$ and $(y_{i}, m(\{v_{2i + 2}\}))$, which means that $(x_{i}, x_{i + 1})$ has two intersecting sheets. This is a contradiction to the specification of the covering task.

\end{proof}

\begin{theorem}
    There is an extension-based proof for the impossibility of solving colorless covering tasks $(\mathcal{I}, \mathcal{O}, \Delta)$  using wait-free protocols.
\end{theorem}
\begin{proof}
    By Lemma \ref{the:covering_lemma}, we know that there exists a 1-dimensional subcomplex $\mathcal{I}^{'}$ of $\mathcal{I}$ such that $(\mathcal{I}^{'}, \mathcal{O}, \Delta)$ has no wait-free algorithms. So the task $(\mathcal{I}, \Delta(\mathcal{I}^{'}), \Delta)$ has no wait-free protocols. 

    Assume that $(\mathcal{I}, \mathcal{O}, \Delta)$ does not have extension-based proofs. By Theorem \ref{the:adversary_finalize_theorem}, we know that there is a set of compatible partial protocols. We restrict these partial protocols to $(\mathcal{I}^{'}, \mathcal{O}, \Delta)$. If the dimension of the input complex is just 1, then any two compatible partial protocols can be merged into a partial protocol whose input complex is the union of input complexes of the two compatible partial protocols. This implies that $(\mathcal{I}^{'}, \Delta(\mathcal{I}^{'}), \Delta)$ has a wait-free protocol. This contradicts Lemma \ref{the:covering_lemma}.
\end{proof}

In fact, this argument works for any 1-dimensional colorless task $(\mathcal{I^{*}}, \mathcal{O^{*}}, \Delta^{*})$ that does not have any wait-free protocol. We can similarly transform it into a colored task $(\mathcal{I}, \mathcal{O}, \Delta)$ and construct a 1-dimensional subcomplex $\mathcal{I}^{'}$(each edge $v_{1} - v_{2}$ in $\mathcal{I^{*}}$ is replaced by two edges $(p_{0}, v_{1}) - (p_{1}, v_{1}) - (p_{0}, v_{2})$ in $\mathcal{I}^{'}$) so that the task $(\mathcal{I^{'}}, \Delta(\mathcal{I}^{'}), \Delta)$ also has no wait-free protocols. In this case, having a set of compatible partial protocols(i.e. having no extension-based proofs) is equal to having a wait-free protocol. Therefore, there exists an extension-based proof for this task.

\section{Conclusions}
\label{sec:conclusions}
In this paper, we solve some open questions proposed in \cite{Brusse21}. An important contribution is that we give a necessary and sufficient condition for colorless tasks to have no extension-based proofs. Previous work is related to some specific tasks, such as the $(n, k)$-set agreement task, and this is the first attempt for colorless tasks. We introduce a more general interaction procedure for a prover and an adversary. The concept of partial protocols can help us understand the local solvable tasks \cite{Attiya23_1}. We have also shown that assignment queries will not give more power to the prover. Different versions of extension-based proofs are equivalent in power for colorless tasks. Our adversarial strategy does not require the condition that the task is colorless before the finalization part, which means that part of our strategy can be generalized to general tasks directly. An open problem that remains is to extend our results to colored tasks or to colored tasks with the minimum requirements. We hope that the concept of canonical neighbors can be used in future research of extension-based proofs and perhaps other questions of distributed computing.




\bibliography{sample-base}

\appendix

\newpage
\section{List of important symbols}
\label{appendix:symbols}

\newcolumntype{Y}{>{\centering\arraybackslash}X}

\begin{table}[h]
    \centering
    \caption{List of important symbols}
    \label{tab:Hg}
    \begin{tabularx}{\textwidth}{|c|Y|c|}
      \toprule
        Symbol & Meaning & Page \\
      \midrule
        $n + 1$ & The number of processes & \pageref{syb:n_plus_1} \\
        $p_{0}, p_{1} \cdots p_{n}$ & processes  & \pageref{syb:processes} \\
        $\Pi$ & the set of all processes & \pageref{syb:pi}\\
        $(\mathcal{I}, \mathcal{O}, \Delta)$ & a tuple(having process ids) that fully determines a task & \pageref{syb:task} \\ 
        $(\mathcal{I^{*}}, \mathcal{O^{*}}, \Delta^{*})$ & a tuple(without process ids) that fully determines a colorless task & \pageref{syb:colorless_task} \\ 
        IS & Immediate snapshot object & \pageref{syb:IS} \\
        NIIS &  The non-uniform iterated immediate snapshot model & \pageref{syb:NIIS} \\
        $\delta$ & A map that specifies an NIIS protocol & \pageref{syb:delta}  \\
        $C$ & A configuration & \pageref{syb:C}\\
        $\alpha$ & A schedule & \pageref{syb:alpha}\\
        
        $S, s, s^{i}$ & A simplex(a superscript represents its dimension) & \pageref{syb:simplex} \\
        $\mathcal{K}$ & A complex & \pageref{syb:complex} \\
        $dim(s), dim(\mathcal{K})$ & The dimension of a simplex $s$ or a complex $\mathcal{K}$ & \pageref{syb:dim} \\
        $lk(S, \mathcal{K})$ & The link of a simplex $S$ in $\mathcal{K}$ & \pageref{syb:link} \\
        $S * T$ & The joining of two simplices $S$ and $T$ & \pageref{syb:joining} \\
        
        $carrier(S, \mathcal{K})$ & The carrier of a simplex $S$ in a complex $\mathcal{K}$ & \pageref{syb:carrier} \\
        $Ch(\mathcal{K})$ & The standard chromatic subdivision of a complex $\mathcal{K}$ & \pageref{syb:Ch} \\
        $\chi(\mathcal{K}, \delta)$ & The non-uniform chromatic subdivision of a complex $\mathcal{K}$ & \pageref{syb:chi} \\

        $U, U_{1}, U_{2} \cdots$ & A simplex in $\mathcal{I}$ & \pageref{syb:U}\\
        $\delta_{U}$ & A partial protocol with respect to $U$ & \pageref{syb:delta_U}\\
        $\mathbb{F}_{t}(U)$ & The $t$-th protocol complex of $\delta_{U}$ & \pageref{syb:mathbb_F_U} \\
        $\delta_{[U_{1} \cdots U_{k}]}$ & A partial protocol with respect to $[U_{1} \cdots U_{k}]$ & \pageref{syb:delta_U_1_2} \\
        $\mathbb{F}_{t}([U_{1} \cdots U_{k}])$ & The $t$-th protocol complex of $\delta_{[U_{1} \cdots U_{k}]}$ & \pageref{syb:mathbb_F_U_1_2} \\

        $\alpha(\varphi)$& The initial schedule of phase $\varphi$ &\pageref{syb:initial_schedule_of_phase_varphi}{} \\
        $\mathcal{A}(\varphi)$& The initial configurations of phase $\varphi$ &\pageref{syb:initial_configurations_of_phase_varphi} \\
        $\mathcal{A}^{'}(\varphi)$& The reached configurations in phase $\varphi$ &\pageref{syb:reached_configurations_in_phase_varphi} \\
 
        $r_{m}$ & The round number of all partial protocols & \pageref{syb:r_m}\\
        $CEN(S)$ & The configuration reached from $S$ until termination by a \allowbreak schedule that repeats $ids(S)$, where $S$ is a simplex in $Ch(\mathcal{I})$  & \pageref{syb:CEN_S} \\

        $t$ & Current complex & \pageref{syb:t} \\
        $S^{0}, S^{1} \cdots S^{t}$ & A sequence of complexes that represent a partially specified protocol & \pageref{syb:complexes} \\
        $N(s^{n}, U)$ & The canonical neighbor of $s^{n}$ with the label $U$, where $s^{n} \in Ch^{r_{m}}(\mathcal{I})$ & \pageref{syb:canonical_neigbhor} \\

        $r_s$ & The interval between two applications of the special rule & \pageref{syb:r_s} \\

        $Q_{i}(U, U^{'})$ & The intersection of $\mathbb{F}_{i}(U)$ and $\mathbb{F}_{i}(U^{'})$ & \pageref{syb:Q} \\
        $SHA(s^{n}, U^{'})$ & The simplex shared by $s^{n}$ and $\mathbb{F}_{r_{m}}(U^{'})$ & \pageref{syb:SHA} \\
        $mc(s^{k}, \mathcal{K})$ & the minimum carrier of $s^{k}$ in the complex $\mathcal{K}$ & \pageref{syb:mc} \\
        $Prj(s^{n}, U^{'})$ & The projection of $s^{n}$ to $\mathbb{F}_{r_{m}}(U^{'})$ & \pageref{syb:prj} \\
        
      \bottomrule
    \end{tabularx}
\end{table}
\newpage

\section{A transformation from  $(\mathcal{I}^{*}, \mathcal{O}^{*}, \Delta^{*})$ to $(\mathcal{I}, \mathcal{O}, \Delta)$}
\label{appendix:transformation}
A {\em pseudosphere} \cite{Her98} is a combinatorial structure in which each process is independently assigned a value from a set of values. Let $s^{n} = (t_{0}, t_{1} \cdots t_{n})$ be an $n$-simplex and $V_{0}, V_{1} \cdots V_{n}$ be a sequence of finite sets of values. In the {\em pseudosphere} $\psi(s^{n}; V_{0}, V_{1} \cdots V_{n})$, each vertex is labeled with a pair $(t_{i}, v_{i})$, where $t_{i}$ is a vertex of $s^{n}$ and $v_{i} \in V_{i}$. Vertices $(t_{i, 0}, v_{i, 0}), \cdots, (t_{i, j}, v_{i, j})$ form a simplex if and only if $t_{i, k}$ for all $k$ are distinct. For an uncolored complex $\mathcal{K}^{*}$ and processes $\Pi = \{p_{0},..., p_{n}\}$, define the {\em colorized complex} $\mathcal{K}$ by replacing each simplex $\tau^{*}$ of $\mathcal{K}^{*}$ with the pseudosphere $\psi(\Pi; \tau^{*}, \tau^{*} \cdots)$.
Let $\pi:\mathcal{K} \rightarrow \mathcal{K}^{*}$ be the map that discards the color component: $\pi(val, color) = val$.
A task $(\mathcal{I}^{*}, \mathcal{O}^{*}, \Delta^{*})$  defined by the colorless form can be transformed into the colored task $(\mathcal{I}, \mathcal{O}, \Delta)$, where $\mathcal{I}$ is the colorized complex of $\mathcal{I}^{*}$, $\mathcal{O}$ is the colorized complex of $\mathcal{O}^{*}$ and $\pi(\Delta(\sigma)) = \Delta^{*}(\pi(\sigma))$.

\section{A geometric proof of Lemma \ref{the:glue_protocols:prover_fails:no_abuse}}
\label{appendix:geometric_proof}

\begin{lemma}
    If a vertex $v$ is terminated by rule (2) with label $U^{'}$ and $v$ is in the subdivision of an $n$-simplex $s^{n}$ of $Ch^{r_{m}}(\mathcal{I})$, then $s^{n}$ has a canonical neighbor with label $U^{'}$.
\end{lemma}

\begin{proof}
Invariant (2) guarantees that a vertex terminated with the label $U^{'}$ by rule (2) has a path consisting of terminated vertices to some $k_{2}$-simplex in $\mathbb{F}_{t}(U^{'})$ terminated with label $U^{'}$. This path consists of vertices terminated by rule (2). We will argue that this path is in the subdivision of some $n$-simplices in $Ch^{r_{m}}(\mathcal{I})$, each of which has a canonical neighbor with the label $U^{'}$.

It has been shown that there exists an equivalent geometric definition of the standard chromatic subdivision in each algebraic topology textbook. Let $K$ be a geometric simplicial complex. The mesh of $K$, denoted $mesh(K)$, is the maximum diameter of any of its simplices or, equivalently, the length of its longest edge. The standard chromatic subdivision is a mesh-shrinking subdivision which means that the mesh of a simplex will decrease after such a subdivision. And in \cite{Herlihy99} it is shown that $mesh(Ch(s^{k})) \leq \frac{n}{n+1} mesh(s^{k})$ for any $k$-simplex $s^k$. Let $l_{m}$ be the mesh of $Ch^{r_{m}}(\mathcal{I})$. Since no vertex in $Ch^{r_{m}}(\mathcal{I})$ has terminated, $l_{m}$ is well-defined. Note that the adversary uses the non-uniform chromatic subdivision to construct $S^{r + 1}$ from $S^{r}$ for each $r \leq r_{m} + r_{a}$. The length of an edge between a terminated vertex and an adjacent vertex will not change after a non-uniform chromatic subdivision. But the length of an edge between two active vertices in $S^{r}$ where $r \geq r_{m}$ is at most $(\frac{n}{n+1})^{r - r_{m}}*l_{m}$. Consider the path $v_{0}, v_{1}, v_{2}, \cdots v_{k}$ where $v_{0}$ terminates with one of its possible configurations. When $v_{i}$ where $i \geq 0$ is defined in $S_{r}$, the distance between $v_{i}$ and $v_{i+1}$ will have an upper limit $(\frac{n}{n+1})^{r - r_{m}}*l_{m}$. If $v_{i+1}$ is a vertex adjacent to $v_{i}$ in $S^{r}$, the upper bound is reached. Otherwise, $v_{i+1}$ is defined in a later subdivision and the distance between $v_{i}$ and $v_{i+1}$ is less than the distance between $v_{i}$ and an adjacent vertex in $S_{r}$. Therefore, the total length of the path has a limitation $\lim\limits_{k \rightarrow \infty}\sum_{i = 0}^{k}|v_{i + 1} - v_{i}| \leq l_{m} * (\frac{n}{n+1})^{r_{a}} / (1 - (\frac{n}{n+1})^{r_{s}})$. Let $D$ be the subcomplex of $Ch^{r_{m}}(\mathcal{I})$, which consists of all simplices that do not have an intersection with $U^{'}$. Then there is a minimum distance from $U^{'}$ to $D$. If we choose $r_{a}$ and $r_{s}$ large enough, we can avoid any path consisting of vertices terminated by rule (2) to reach $D$. All vertices terminated by rule (2) are contained in the $n$-simplices in $Ch^{r_{m}}(\mathcal{I})$ that have an intersection with $U^{'}$.
\end{proof}

\end{document}